\documentclass[a4paper,USenglish,cleveref,autoref]{lipics-v2021}
\hideLIPIcs
\nolinenumbers

\usepackage{amsmath,amssymb,amsthm,mathtools}
\usepackage{algorithm}
\usepackage{algorithmic}
\usepackage{aliascnt}
\usepackage{subcaption}
\usepackage{hyperref}
\usepackage{xcolor}
\usepackage{microtype}
\usepackage{framed}
\hypersetup{
  colorlinks=true,
  linkcolor=blue!60!black,
  citecolor=blue!60!black,
  urlcolor=blue!60!black
}
\makeatletter
\let\lipics@original@ps@headings\ps@headings
\def\ps@headings{%
  \lipics@original@ps@headings
  \ifx\@hideLIPIcs\@undefined\else
    \let\@oddfoot\@empty
  \fi
}
\makeatother
\usepackage{tikz}
\usetikzlibrary{patterns,arrows.meta,calc}

\newaliascnt{question}{theorem}
\newtheorem{question}[question]{Question}
\aliascntresetthe{question}

\crefname{theorem}{Theorem}{Theorems}
\Crefname{theorem}{Theorem}{Theorems}
\crefname{corollary}{Corollary}{Corollaries}
\Crefname{corollary}{Corollary}{Corollaries}
\crefname{definition}{Definition}{Definitions}
\Crefname{definition}{Definition}{Definitions}
\crefname{proposition}{Proposition}{Propositions}
\Crefname{proposition}{Proposition}{Propositions}
\crefname{observation}{Observation}{Observations}
\Crefname{observation}{Observation}{Observations}
\crefname{lemma}{Lemma}{Lemmas}
\Crefname{lemma}{Lemma}{Lemmas}
\crefname{claim}{Claim}{Claims}
\Crefname{claim}{Claim}{Claims}
\crefname{problem}{Problem}{Problems}
\Crefname{problem}{Problem}{Problems}
\crefname{conjecture}{Conjecture}{Conjectures}
\Crefname{conjecture}{Conjecture}{Conjectures}
\crefname{question}{Question}{Questions}
\Crefname{question}{Question}{Questions}
\crefname{example}{Example}{Examples}
\Crefname{example}{Example}{Examples}
\crefname{fact}{Fact}{Facts}
\Crefname{fact}{Fact}{Facts}
\crefname{remark}{Remark}{Remarks}
\Crefname{remark}{Remark}{Remarks}
\crefname{section}{Section}{Sections}
\Crefname{section}{Section}{Sections}
\crefname{subsection}{Section}{Sections}
\Crefname{subsection}{Section}{Sections}
\crefname{subsubsection}{Section}{Sections}
\Crefname{subsubsection}{Section}{Sections}
\crefname{equation}{Equation}{Equations}
\Crefname{equation}{Equation}{Equations}
\crefname{algorithm}{Algorithm}{Algorithms}
\Crefname{algorithm}{Algorithm}{Algorithms}

\newcommand{\LOCAL}{\mathsf{LOCAL}}
\newcommand{\CONGEST}{\mathsf{CONGEST}}
\newcommand{\Vol}{\mathsf{VOLUME}}

\newcommand{\xor}{\oplus}
\newcommand{\insub}{_{\mathrm{in}}}
\newcommand{\outsub}{_{\mathrm{out}}}
\newcommand{\LeafColoring}{\mathsf{LeafColoring}}
\newcommand{\HTHC}{\mathsf{Hierarchical}\text{-}\mathsf{THC}}

\newcommand{\PathToLeafK}{\Pi_\mathsf{PTL}^k}
\newcommand{\PathToLeaf}{\Pi_\mathsf{PTL}}
\newcommand{\HC}{\Pi_{\mathsf{HC}}}
\renewcommand{\L}{\mathsf{L}}
\newcommand{\R}{\mathsf{R}}
\renewcommand{\P}{\mathrm{P}}
\newcommand{\D}{\mathsf{D}}
\newcommand{\X}{\mathsf{X}}
\newcommand{\B}{\mathsf{B}}
\newcommand{\U}{\mathsf{U}}
\newcommand{\Cyc}{\mathsf{Cyc}}
\newcommand{\LC}{\mathrm{LC}}
\newcommand{\RC}{\mathrm{RC}}
\newcommand{\NC}{\mathrm{NC}}
\newcommand{\DC}{\mathrm{DC}}
\newcommand{\val}{\operatorname{val}}
\newcommand{\child}{\operatorname{child}}

\title{New Complexity Classes in Locally Checkable Labeling for Local Computation Algorithms}
\titlerunning{New LCL Complexity Classes for LCAs}

\author{Sijin Peng}{CSAIL, MIT, Cambridge, United States}{sijinp@mit.edu}{https://orcid.org/0009-0007-4377-1437}{}

\authorrunning{S. Peng}

\Copyright{Sijin Peng}
\DOIPrefix{}



\acknowledgements{The LCL construction, lower bound analysis, and high-level ideas for upper bound algorithms in this paper were discovered by the author. At the same time, most of the proof details in \cref{sec:log,sec:poly}, including \cref{def:poly-strong-algo}, and figures in \cref{sec:Intro,sec:tech-overview} were assisted by the use of OpenAI Codex. All AI outputs were reviewed and edited by the author, who takes full responsibility for the correctness, originality, and integrity of this paper.}

\begin{document}
\maketitle

\begin{abstract}
Local Computation Algorithms (LCAs), introduced by Rubinfeld, Tamir, Vardi, and Xie (2011), are a special type of sublinear algorithms that, given probing access to a possibly massive input, are required to provide query access to a consistent solution, without maintaining a state between different queries. In this paper, we try to understand LCA through the lens of complexity classifications, described by the following question: Given a target complexity function $f(n)$, is there a problem whose local computation complexity is $f(n)$, up to polylogarithmic factors?

We restrict our focus to Locally Checkable Labeling (LCL) problems, which can be seen as constant-degree constraint satisfaction problems. Possible complexity classes of this problem family have been extensively studied in various distributed computation models, including the $\Vol$ model proposed by Rosenbaum and Suomela (2020), which is an invariant of local computation algorithms with additional locality requirements.

In this paper, we provide new LCL complexity constructions in the $\Vol$ model, and generalize the results to LCAs. Specifically, we show that there are LCLs whose probe complexities in the $\Vol$ and LCA models are $\Theta(\log^k n)$ and $\tilde \Theta(n^{p/q})$ for any positive integer $k \ge 1$ and rational $p/q \in (0,1]$. Our approach, completely different from the approach to a similar result in the distributed $\LOCAL$ model by Balliu et al. (2018), is to \emph{stack} instances of complexity $\Theta(\log n)$ and $\tilde \Theta(n^{1/k})$ in the $\Vol$ model constructed by Rosenbaum and Suomela (2020).
\end{abstract}

\section{Introduction}
\label{sec:Intro}

Local computation algorithms (LCAs), proposed in \cite{alon2012space, rubinfeld2011fast}, are a special type of sublinear algorithms typically applied to graph problems involving multiple output bits. An LCA is expected to find only a small part of the solution each time, specified by \emph{queries}, through a small number of \emph{probes} that give the algorithm access to the input and adjacency list of a vertex in the graph. Following most of the LCA literature, in this paper, the input graph has a constant maximum degree, so it is the same up to a constant factor for a probe to return the complete adjacency list of a vertex as to return one entry of the adjacency list.

The main challenge in designing a local computation algorithm is consistency: when there are multiple feasible solutions to the instance, multiple queries to different parts of the solution should provide consistent local solutions that point to a \emph{single} globally feasible solution without maintaining state across queries. For graph coloring, for example, a query may ask for the color of one vertex, and the answers to multiple queries must together form a single feasible coloring. To make this possible, if an LCA needs randomness, the same randomness must be shared across queries.

The motivation behind LCA is to provide an efficient and parallelizable way to gain access to a small part of a large solution to a massive problem. This shares the spirit with distributed computing and property testing. See \cite{levi2017centralized} for a survey on common methods in LCA and its connection to property testing and various distributed models. LCAs have also been connected to several other fields of theoretical computer science, including online algorithms \cite{mansour2012converting}, graph sparsification \cite{azarmehr2025stochastic}, distributed optimization \cite{london2017distributed}, coding theory \cite{li2025constructing, rubinfeld2011fast}, and algorithmic game theory \cite{hassidim2016local}. However, most existing literature on LCA takes an algorithm-designer perspective, seeking the optimal probe complexity of concrete problems. In this paper, we ask the question the other way around: Given a target complexity $f(n)$, is there a problem whose probe complexity is exactly or roughly $f(n)$ up to a polylogarithmic factor?

This question becomes meaningless if we take arbitrary problems into consideration. For example, one can take a problem with $\Theta(n)$ probe complexity and add a promise to the input instance that each connected component only has size $f(n)$, likely leading to a problem of probe complexity $\Theta(f(n))$. Another candidate with probe complexity $\Theta(f(n))$ is to give each vertex $v$ a radius $r_v \le \log_\Delta f(n)$ and let each vertex collect the exact number of vertices within distance $r_v$ from $v$. From the examples above, it is essential to select a suitable family of problems that excludes trivial issues, while also striving to include a large class of meaningful problems.

Our approach, as is standard in distributed graph algorithms, is to restrict ourselves to the family of locally checkable labeling (LCL) problems. To talk about LCL, we should first define the graph family these problems are defined on. In this paper, we will analyze LCLs on the family of all bounded-degree graphs and the family of all bounded-degree trees. We cannot make additional promise about the input structure, so the first example above does not fall into our consideration.

In a typical LCL setting, each vertex in a \emph{constant}-degree graph receives its unique ID and possibly an input label from a \emph{constant} number of possible inputs, and each vertex should provide an output label from a \emph{constant} number of choices. Furthermore, the feasibility of a solution is \emph{locally checkable}. This means that each vertex can check whether its local constraint is satisfied with information from a constant-hop neighborhood, and a global solution satisfies the requirement of the problem if all vertices are locally satisfied. One may also think of LCLs as CSPs on constant-degree constraint graphs. The second example given above is not an LCL because its label sets or checkability radius are not constant. At the same time, LCLs still include classic problems, such as graph coloring and maximal independent set.

Since the work by Naor and Stockmeyer \cite{naor1993can}, the possible complexity of LCLs has been extensively studied in the context of distributed graph algorithms \cite{balliu2026distributed2,balliu2026distributed, bousquet2024local,chang2019exponential,chang2023distributed2,dhar2024local}. There is now a substantial understanding of LCL complexities in the $\LOCAL$ model on bounded-degree graphs (see \cite{balliu2020much} for the almost-complete complexity spectrum), trees \cite{balliu2022efficient, balliu2021locally, chang2020complexity, chang2023distributed, grunau2022landscape}, and other special types of graphs \cite{balliu2019distributed, berlow2025separating, brandt2017lcl, grebik2023local}. This research paradigm has also extended to other models and settings \cite{akbari2021locality, balliu2021locally2, chang2023distributed, schmid2026lcls}, among which the $\Vol$ model proposed by Rosenbaum and Suomela \cite{rosenbaum2020seeing} has a close connection to local computation algorithms.

The $\Vol$ model is motivated to provide a more fine-grained resource analysis compared to other common distributed models such as $\LOCAL$ and $\CONGEST$. This model shares a similar spirit with LCA, using probes to measure the cost of an algorithm, but they are different in two aspects: 
\begin{itemize}
    \item First, any time in the execution of a $\Vol$ algorithm, the set of vertices probed by the algorithm should always be a connected component that includes the current vertex being queried. In other words, $\Vol$ algorithms cannot perform a far probe with low cost.
    \item Another difference lies in the randomness: As mentioned above, in local computation algorithms, different queries share the single random string $r$, and the algorithm can freely access each bit of $r$. In the $\Vol$ model, different queries still share the same randomness, but the algorithm only has restricted access to it. Specifically, each vertex $v$ in the graph has an independent random string $r_v$ that jointly constitutes $r$, and for a query, the algorithm can access $r_v$ only after it probes $v$.
\end{itemize}
In these two dimensions, the $\Vol$ model more closely resembles common distributed models in which information only propagates through network channels abstracted as edges in the graph, and the algorithm needs to use additional resources to get information from vertices far away.

\subsection{Main Results}

In this paper, we consider LCLs with probe complexity $\Omega(\log n)$ in both the LCA and $\Vol$ models for randomized algorithms only. To the best of the author's knowledge, no prior work addresses the LCL complexity classification for LCAs. Several existing papers have studied randomized $\Vol$ complexities of LCLs \cite{brandt2021randomized,grunau2022landscape,rosenbaum2020seeing}. See \cref{fig:rand-vol-complexity} for an illustration of the status of the randomized $\Vol$ complexity landscape for bounded-degree graphs before our work. We note here that \cite{rosenbaum2020seeing} also includes results for the deterministic $\Vol$ model, but current knowledge on this model is largely limited by the fact that it is unknown whether there is either a gap theorem or an explicit LCL construction between $\Omega(\log^\star n)$ and $o(n)$ for the deterministic $\Vol$ model.

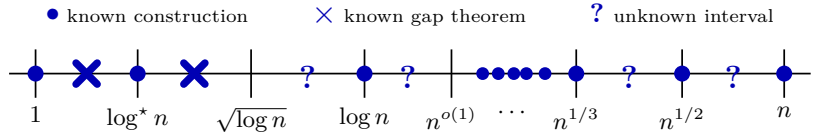
\begin{figure}[htbp]
    \centering
    \begin{tikzpicture}[
    x=1cm,y=1cm,
    axis/.style={line width=0.65pt},
    tick/.style={line width=0.65pt},
    construction/.style={blue!70!black, fill=blue!70!black},
    gapmark/.style={
        blue!70!black,
        line width=2.8pt,
        line cap=round
    },
    ticklabel/.style={font=\small},
    openmark/.style={blue!70!black, font=\bfseries\Large, inner sep=0pt},
    legendlabel/.style={font=\scriptsize, anchor=west}
]

\coordinate (one)      at (0,0);
\coordinate (lstar)    at (1.35,0);
\coordinate (sqrtlog)  at (2.85,0);
\coordinate (logn)     at (4.35,0);
\coordinate (subpoly)  at (5.50,0);
\coordinate (nthird)   at (7.15,0);
\coordinate (nhalf)    at (8.55,0);
\coordinate (n)        at (9.90,0);

\draw[axis] (-0.35,0) -- (10.25,0);

\foreach \p/\lab in {
    one/{$1$},
    lstar/{$\log^\star n$},
    sqrtlog/{$\sqrt{\log n}$},
    logn/{$\log n$},
    subpoly/{$n^{o(1)}$},
    nthird/{$n^{1/3}$},
    nhalf/{$n^{1/2}$},
    n/{$n$}
}{
    \draw[tick] ($(\p)+(0,-0.32)$) -- ($(\p)+(0,0.32)$);
    \node[ticklabel,below=0.32cm] at (\p) {\lab};
}

\foreach \p in {one,lstar,logn,nthird,nhalf,n}{
    \fill[construction] (\p) circle (3.0pt);
}

\foreach \x in {0.25,0.375,0.50,0.60,0.75}{
    \fill[construction] ($(subpoly)!\x!(nthird)$) circle (2.4pt);
}

\newcommand{\volgapcross}[1]{%
    \draw[gapmark] ($(#1)+(-0.14,-0.14)$) -- ($(#1)+(0.14,0.14)$);
    \draw[gapmark] ($(#1)+(-0.14,0.14)$) -- ($(#1)+(0.14,-0.14)$);
}

\coordinate (gapone)      at ($(one)!0.5!(lstar)$);
\coordinate (gaplstar)    at ($(lstar)!0.5!(sqrtlog)$);
\coordinate (opensqrt)    at ($(sqrtlog)!0.5!(logn)$);
\coordinate (openpolylog) at ($(logn)!0.5!(subpoly)$);
\coordinate (openpolyone) at ($(nthird)!0.5!(nhalf)$);
\coordinate (openpolytwo) at ($(nhalf)!0.5!(n)$);

\volgapcross{gapone}
\volgapcross{gaplstar}

\node[openmark] at ($(opensqrt)+(0,-0.04)$) {?};
\node[openmark] at ($(openpolylog)+(0,-0.04)$) {?};
\node[openmark] at ($(openpolyone)+(0,-0.04)$) {?};
\node[openmark] at ($(openpolytwo)+(0,-0.04)$) {?};

\node[ticklabel,below=0.32cm] at ($(subpoly)!0.50!(nthird)$) {$\cdots$};

\node[legendlabel,anchor=center] at ($(one)!0.5!(n)+(0,0.78)$) {%
    {\color{blue!70!black}\raisebox{0.12ex}{\large$\bullet$}} known construction
    \hspace{0.65cm}
    {\color{blue!70!black}\raisebox{-0.12ex}{\Large\bfseries$\times$}} known gap theorem
    \hspace{0.65cm}
    {\color{blue!70!black}\bfseries\normalsize ?} unknown interval%
};

\end{tikzpicture}
    \caption{The LCL complexity landscape of randomized $\Vol$ algorithms from previous works for bounded-degree graphs. The blue dots between $n^{o(1)}$ and $n^{1/3}$ correspond to constructions for $n^{1/k}$ for any integer $k \ge 3$.}
    \label{fig:rand-vol-complexity}
\end{figure}

Before our work, one can construct LCLs with randomized $\Vol$ complexity $\Theta(1)$, $\Theta(\log^\star n)$, $\Theta(\log n)$, and $\tilde \Theta(n^{1/k})$ for any positive integer $k$, where $\tilde \Theta$ ignores polylogarithmic factors. The constructions for $\Theta(1)$ and $\Theta(\log^\star n)$ complexity come from constructions in the $\LOCAL$ model, while other constructions are given in \cite{rosenbaum2020seeing}. It turns out that all LCL constructions in \cref{fig:rand-vol-complexity} have the same complexity in the $\Vol$ model when we restrict the instance to bounded-degree trees. Two gap theorems in the randomized $\Vol$ model, $\omega(1) - o(\log^\star n)$ gap and $\omega(\log^\star n) - o(\sqrt{\log n})$ gap, are given in \cite{grunau2022landscape} and \cite{brandt2021randomized} respectively. Unfortunately, these constructions and gap theorems do not directly transfer to LCAs.

At the same time, there are still regions where we do not know either a separation theorem or an LCL construction. It is explicitly asked in \cite{rosenbaum2020seeing} whether more randomized $\Vol$ complexities exist in the polylogarithmic regime, and whether the set of possible complexities in the polynomial regime is ``dense'', meaning that for any two reals $0 < a < b \le 1$, there is a problem with complexity $\Omega(n^a)$ and $o(n^b)$.

In this paper, we provide affirmative answers to both questions, and extend the constructions and analysis further to the LCA setting and bounded-degree trees, as stated in \cref{thm:log-main} and \cref{thm:poly-main}. Adding our results gives the new complexity landscape for the $\Vol$ model shown in \cref{fig:rand-vol-complexity-new}.

\begin{theorem}\label{thm:log-main}
    For each positive integer $k$, there exists an LCL whose randomized LCA and $\Vol$ complexities are both $\Theta(\log^k n)$ for both bounded-degree graphs and trees.
\end{theorem}

\begin{theorem}\label{thm:poly-main}
    For each $x \in \mathbb{Q} \cap (0,1]$, there exists an LCL whose randomized LCA and $\Vol$ complexities are both $\tilde \Theta(n^x)$ for both bounded-degree graphs and trees.
\end{theorem}

\begin{figure}[htbp]
    \centering
    \begin{tikzpicture}[
    x=1cm,y=1cm,
    axis/.style={line width=0.65pt},
    tick/.style={line width=0.65pt},
    construction/.style={blue!70!black, fill=blue!70!black},
    denseconstruction/.style={
        blue!70!black,
        line width=3.2pt,
        line cap=round
    },
    gapmark/.style={
        blue!70!black,
        line width=2.8pt,
        line cap=round
    },
    ticklabel/.style={font=\small},
    openmark/.style={blue!70!black, font=\bfseries\Large, inner sep=0pt},
    legendlabel/.style={font=\scriptsize, anchor=west}
]

\coordinate (one)      at (0,0);
\coordinate (lstar)    at (1.20,0);
\coordinate (sqrtlog)  at (2.40,0);
\coordinate (logn)     at (3.55,0);
\coordinate (logtwo)   at (4.70,0);
\coordinate (logthree) at (5.85,0);
\coordinate (superlog) at (7.55,0);
\coordinate (subpoly)  at (9.35,0);
\coordinate (n)        at (10.65,0);

\draw[axis] (-0.35,0) -- (10.95,0);

\foreach \p/\lab in {
    one/{$1$},
    lstar/{$\log^\star n$},
    sqrtlog/{$\sqrt{\log n}$},
    logn/{$\log n$},
    logtwo/{$\log^2 n$},
    logthree/{$\log^3 n$},
    n/{$n$}
}{
    \draw[tick] ($(\p)+(0,-0.32)$) -- ($(\p)+(0,0.32)$);
    \node[ticklabel,below=0.32cm] at (\p) {\lab};
}

\draw[tick] ($(superlog)+(0,-0.32)$) -- ($(superlog)+(0,0.32)$);
\node[ticklabel,below=0.32cm] at (superlog) {$2^{\omega(\log\log n)}$};

\draw[tick] ($(subpoly)+(0,-0.32)$) -- ($(subpoly)+(0,0.32)$);
\node[ticklabel,below=0.32cm] at (subpoly) {$n^{o(1)}$};

\draw[denseconstruction] ($(subpoly)+(0.05,0)$) -- (n);

\foreach \p in {one,lstar,logn,logtwo,logthree}{
    \fill[construction] (\p) circle (3.0pt);
}

\foreach \x in {0.143,0.286,0.429,0.571,0.714,0.857}{
    \fill[construction] ($(logthree)!\x!(superlog)$) circle (2.4pt);
}

\newcommand{\volgapcross}[1]{%
    \draw[gapmark] ($(#1)+(-0.14,-0.14)$) -- ($(#1)+(0.14,0.14)$);
    \draw[gapmark] ($(#1)+(-0.14,0.14)$) -- ($(#1)+(0.14,-0.14)$);
}

\coordinate (gapone)     at ($(one)!0.5!(lstar)$);
\coordinate (gaplstar)   at ($(lstar)!0.5!(sqrtlog)$);
\coordinate (opensqrt)   at ($(sqrtlog)!0.5!(logn)$);
\coordinate (openlogone) at ($(logn)!0.5!(logtwo)$);
\coordinate (openlogtwo) at ($(logtwo)!0.5!(logthree)$);
\coordinate (opensubpoly) at ($(superlog)!0.5!(subpoly)$);

\volgapcross{gapone}
\volgapcross{gaplstar}

\node[openmark] at ($(opensqrt)+(0,-0.04)$) {?};
\node[openmark] at ($(openlogone)+(0,-0.04)$) {?};
\node[openmark] at ($(openlogtwo)+(0,-0.04)$) {?};
\node[openmark] at ($(opensubpoly)+(0,-0.04)$) {?};

\node[ticklabel,above=0.22cm] at ($(logthree)!0.50!(superlog)$) {$\cdots$};

\node[legendlabel,anchor=center] at ($(one)!0.5!(n)+(0,0.80)$) {%
    {\color{blue!70!black}\raisebox{0.12ex}{\large$\bullet$}} known construction
    \hspace{0.35cm}
    {\color{blue!70!black}\rule[0.45ex]{0.42cm}{1.1pt}} dense constructions
    \hspace{0.35cm}
    {\color{blue!70!black}\raisebox{-0.12ex}{\Large\bfseries$\times$}} known gap theorem
    \hspace{0.35cm}
    {\color{blue!70!black}\bfseries\normalsize ?} unknown interval%
};

\end{tikzpicture}
    \caption{The current LCL complexity landscape of randomized $\Vol$ algorithms for bounded-degree graphs. The blue dots between $\log^3 n$ and $2^{\omega(\log\log n)}$ correspond to constructions for $\log^k n$ for any integer $k \ge 3$; The solid blue segment in the polynomial regime represents constructions with complexity $\tilde\Theta(n^x)$ for every rational $x \in (0,1]$.}
    \label{fig:rand-vol-complexity-new}
\end{figure}
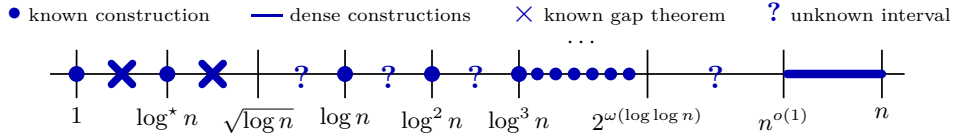

\subsection{Comparison with Results in the \texorpdfstring{$\LOCAL$}{LOCAL} Model}

It is worth noting the difference between our results in the LCA and $\Vol$ model and similar results in the $\LOCAL$ model. In the $\LOCAL$ model proposed by Linial \cite{linial1992locality}, the complexity is measured by the number of communication rounds, equivalently, the radius of the neighborhood that determines the output for each vertex.

For bounded-degree trees, it is shown in \cite{chang2020complexity,chang2019time} that no LCL has complexity in the range $\omega(\log n)-n^{o(1)}$ or the range $\omega(n^{1/(k+1)})-o(n^{1/k})$ for any integer $k \ge 1$. Our results indicate that the complexity landscape of LCLs in bounded-degree trees is more complex for the $\Vol$ model and LCAs compared to the $\LOCAL$ model in polylogarithmic and polynomial regimes. This is consistent with the intuition that probe complexity looks like a more fine-grained complexity measure compared to distance complexity.

Yet this intuition does not hold for general graphs, as it is shown in \cite{balliu2018almost,balliu2018new} that, for bounded-degree graphs, the possible LCL complexity for the $\LOCAL$ model is dense in the interval ranging from $\log n$ to $n$. More specifically, for any positive rational $r/s \le 1$ and $p/q \ge 1$, one can construct LCLs with distance complexity $\Theta(\log^{p/q}n)$, $2^{\Theta(\log^{r/s}n)}$ and $\Theta(n^{r/s})$. Our result exhibits a similar density in the polynomial regime, but the techniques are completely different. The constructions in \cite{balliu2018almost,balliu2018new} use a set of carefully crafted constraints to encode the execution of special Turing machines into the instance. Instead, our constructions use the constructions in \cite{rosenbaum2020seeing} with $\Vol$ complexity $\Theta(\log n)$ and $\tilde \Theta(n^{1/k})$ and provide a way to \emph{stack} multiple instances of these problems to create more complexity classes. We will briefly discuss how the stacking works in \cref{sec:tech-overview}.

\subsection{Organization of the Paper}

In \cref{sec:tech-overview}, we provide an overview of our approach. In \cref{sec:prelim}, we define basic notations and also models and problems of our interest. In \cref{sec:log} and \cref{sec:poly}, we provide the LCL construction and analysis in the polylogarithmic complexity setting and the polynomial complexity setting, respectively. Finally, in \cref{sec:discussion}, we discuss several aspects of our paper at a higher level and provide some future directions.

\section{Technical Overview}
\label{sec:tech-overview}

\subsection{The Construction in the Polylogarithmic Regime}

In the polylogarithmic regime, we use the $\LeafColoring$ LCL constructed in \cite{rosenbaum2020seeing} with $\Vol$ complexity $\Theta(\log n)$, and stack its instances for $k$ levels to provide an LCL whose $\Vol$ and LCA complexities are $\Theta(\log^k n)$ for any positive integer $k$. We first review the $\LeafColoring$ construction and then show how the stacking works. Finally, we briefly demonstrate how to generalize the lower bound analysis from $\Vol$ to LCAs.

\subsubsection{Review of the \texorpdfstring{$\LeafColoring$}{LeafColoring} LCL}

The main idea behind the $\LeafColoring$ LCL is that, on all possible rooted trees of size $n$ where each vertex has out-degree zero or two, the root needs $\Omega(\log n)$ steps to reach a leaf, while a random walk takes $O(\log n)$ steps from the root to a leaf with high probability for any such tree.

To make sure that any $\Vol$ algorithm working on the output label for the root must discover a path from the root to a leaf, we impose the following constraints: Each leaf is given a red-blue color as an input and required to output the same color as its input, while each internal vertex needs to output a color appearing as the output color of one of the two children. These locally checkable constraints together require the root to output the input color of one of the leaves, and if an algorithm does not probe any leaf, it cannot distinguish between cases where all leaves are red and all leaves are blue.

Finally, to generalize the construction to all bounded-degree trees and graphs, one needs to rule out graphs not following a rooted tree of out-degree zero or two without an explicit promise on the input structure. To address this issue, Rosenbaum and Suomela \cite{rosenbaum2020seeing} design a special set of input labels for each vertex and constraints called binary tree labeling. The binary tree labeling is a part of the LCL. If a vertex does not satisfy the locally checkable constraints for the binary tree labeling, we will output a special $\perp$ label for this vertex and remove it from the rest of the computation. For vertices satisfying all constraints for the binary tree labeling, we can show that the induced subgraph will be almost a rooted forest of out-degree zero or two, and we require these vertices to solve the coloring problem. We will use the same approach in \cref{def:log-binary-tree-labeling,def:log-tree-labeling-consistency} to generalize our construction to all bounded-degree trees and graphs. For ease of understanding, in this section, we can just assume that all vertices pass the binary tree labeling check.

\subsubsection{The Stacking}

First, we slightly generalize the $\LeafColoring$ construction by using the following set of constraints: Instead of the red/blue color, we give for each vertex an input bit $b\insub(v)$, and each vertex needs to output two symbols. The first symbol $d(v) \in \{\L, \R\}$ indicates the path from each vertex to a leaf. Define $\mathrm{child}(v,d(v))$ to be the left child of $v$ when $d(v) = \L$ and the right child of $v$ when $d(v) = \R$, then the path from some vertex $u$ will be $p_0 = u, p_1 = \child(p_0, d(p_0)), p_2 = \child(p_1, d(p_1)), \cdots,$ all the way down to a leaf $w$ where $d(w)$ should be a placeholder $\perp$. The second output symbol $b\outsub(v)$ is to ensure that the algorithm probes all vertices along the path, and we add the constraint that $b\outsub(u) = \xor_{i \ge 0} b\insub(p_i)$ for each vertex $u$. This constraint involves $\Omega(\log n)$ labels, so it is not local, but we notice that we can use the locally checkable constraint $b\outsub(v) = b\outsub(p_1) \oplus b\insub(v)$ on each internal vertex $v$ in the instance to jointly make sure that $b\outsub(u)$ equals the xor of input bits along the path. 

We denote the previous LCL as $\PathToLeaf^1$, where $\mathrm{PTL}$ abbreviates ``path to leaf''. Compared to $\LeafColoring$, the benefit of $\PathToLeaf^1$ is that the output of the root is always related to $\Theta(\log n)$ input bits rather than a single input color on some leaf. As a result, we would expect that any randomized LCA would also need $\Omega(\log n)$ probes.

We will now demonstrate the way to \emph{stack} $\PathToLeaf^1$ instances. Consider, for each vertex $v$ in a $\PathToLeaf^1$ instance of size $m$ called the base instance, we give it a ``twisted'' input bit $b\insub'(v)$ and additionally attach a fresh $\PathToLeaf^1$ instance of size $m$. Different vertices receive disjoint instances, so the final instance includes $(m+1)$ independent $\PathToLeaf^1$ instances where one of them is the base instance. For a vertex in the base instance, the actual input bit $b\insub(v)$ will be the xor of the twisted input bit $b'\insub(v)$ and the output bit of the root of the attached instance. The $\PathToLeaf^1$ constraints for each of the $(m+1)$ instances are the same as before.

Now, to get the output bit for the root of the base instance, apart from $\Theta(\log m)$ input bits along the path, the algorithm needs to additionally get the root output label for $\Theta(\log m)$ independent attached instances. This suggests that the LCL should provide a $\Theta(\log^2 m)$ probe complexity construction. The total instance size will blow up by a polynomial factor, but it does not affect the complexity since we are in the polylogarithmic regime. We call this problem $\PathToLeaf^2$, and we can recursively do the attachment to define $\PathToLeafK$ for any positive integer $k$.

\subsubsection{Upper Bound and Lower Bound}

The $\Theta(\log^2 n)$ $\Vol$ and LCA upper bound for $\PathToLeaf^2$ follows by generalizing the random-walk algorithm in \cite{rosenbaum2020seeing}. For the lower bound, one may first approach it using uniform random bits for $b\insub$, hoping that the output bit at the root of the base instance is then the xor sum of $\Theta(\log^2 n)$ uniform random bits. The issue is that the algorithm can choose $d(v)$ based on the input bits. For example, an algorithm may use a few probes to find a $d(v)$ assignment adaptively such that the correct $b\outsub$ value at the root is zero with high probability.

To resolve this issue, we construct a special distribution over the instances, so that each instance $I$ in the support is associated with a bit $b_I \in\{0,1\}$, and no matter how the output label $d$ is chosen for each vertex, the output bit at the root of the base instance in a globally correct solution will always be $b_I$. This removes the effect of $d$ outputs, and we can use an induction argument to show that, conditioned on the input bit of any set of $o(\log^2 n)$ vertices, the probability that the associated instance bit $b_I$ equals $1$ will be exactly $1/2$. This, along with Yao's minimax principle, proves the lower bound.

\subsection{The Construction in the Polynomial Regime}

In the polynomial regime, we utilize the $\HTHC$ LCL constructed in \cite{rosenbaum2020seeing} with $\Vol$ complexity $\tilde \Theta(n^{1/k})$ for any positive integer $k$, and stack its instances in a nontrivial way to provide an LCL of randomized LCA and $\Vol$ complexities $\tilde \Theta(n^{p/q})$ for any positive rational $0<p/q \le 1$. We first review the $\HTHC$ construction and then show how the stacking works. Finally, we demonstrate the difficulties in the upper and lower bound analysis.

\subsubsection{Review of the \texorpdfstring{$\HTHC$}{Hierarchical-THC} LCL}

Our presentation of the $\HTHC$ LCL will be different from previous works \cite{chang2019time,rosenbaum2020seeing}, and the purpose is to show that $\HTHC$ is constructed by stacking path two-coloring instances in a different way compared to the stacking in the polylogarithmic regime. This allows us to integrate two different stacking strategies to provide new complexity classes.

The $\HTHC$ construction is parameterized by levels of stacking. To start with, let us consider the construction with two levels, leading to an LCL with $\tilde \Theta(n^{1/2})$ probe complexity in the $\Vol$ model. In the construction, we have a directed \emph{top} path of length $\ell$, denoted by $v_1,v_2,\dots,v_\ell$, with each vertex $v_i$ in the top path possibly having a directed \emph{bottom} path $P_i$ attached. The total number of vertices in the top path and all bottom paths is $n$. See \cref{fig:hthc-top-bottom-paths}(a) for an illustration.

The constraints of the $\HTHC$ LCL will guarantee the following: The algorithm needs to either visit all vertices in the top path to get the output for $v_1$, or visit all vertices in one of the bottom paths to get the output of the head of the bottom path. This immediately implies $\Omega(n^{1/2})$ lower bound from the instance where $\ell = \Theta(\sqrt{n})$ and $|P_i| = \Theta(\sqrt n)$ for each $1 \le i \le \ell$. The matching upper bound requires additional work when bottom paths have varying lengths, and we will discuss this after introducing our LCL construction.

\newcommand{\HTHCPathTemplate}{%
    \begin{tikzpicture}[
        vertex/.style={circle,draw,inner sep=0pt,minimum size=6mm,font=\scriptsize},
        red vertex/.style={vertex,fill=red!35},
        blue vertex/.style={vertex,fill=blue!35},
        path edge/.style={->,thick,>=Stealth},
        attach edge/.style={->,densely dotted,thick,>=Stealth},
        note/.style={font=\scriptsize}
    ]
        \node[red vertex] (v1) at (0,0) {$v_1$};
        \node[red vertex] (v2) at (1.65,0) {$v_2$};
        \node[red vertex] (v3) at (3.3,0) {$v_3$};
        \node[red vertex] (v4) at (4.95,0) {$v_4$};

        \draw[path edge] (v1) -- (v2);
        \draw[path edge] (v2) -- (v3);
        \draw[path edge] (v3) -- (v4);

        \node[note] at ($(v1)+(0,0.45)$) {head};
        \node[note] at ($(v4)+(0,0.45)$) {tail};

        \node[blue vertex] (p11) at (0,-2.25) {$p_{1,1}$};
        \node[blue vertex] (p12) at (1.65,-2.25) {$p_{1,2}$};
        \node[blue vertex] (p13) at (3.3,-2.25) {$p_{1,3}$};
        \draw[path edge] (p11) -- (p12);
        \draw[path edge] (p12) -- (p13);
        \draw[attach edge] (v1) -- (p11);

        \node[blue vertex] (p21) at (1.65,-1.25) {$p_{2,1}$};
        \draw[attach edge] (v2) -- (p21);

        \node[blue vertex] (p31) at (3.3,-1.25) {$p_{3,1}$};
        \node[blue vertex] (p32) at (4.95,-1.25) {$p_{3,2}$};
        \draw[path edge] (p31) -- (p32);
        \draw[attach edge] (v3) -- (p31);
    \end{tikzpicture}
}

\newcommand{\HTHCPathTemplateBare}{%
    \begin{tikzpicture}[
        vertex/.style={circle,draw,inner sep=0pt,minimum size=6mm,font=\scriptsize},
        path edge/.style={->,thick,>=Stealth},
        attach edge/.style={->,densely dotted,thick,>=Stealth}
    ]
        \node[vertex] (v1) at (0,0) {};
        \node[vertex] (v2) at (1.65,0) {};
        \node[vertex] (v3) at (3.3,0) {};
        \node[vertex] (v4) at (4.95,0) {};

        \draw[path edge] (v1) -- (v2);
        \draw[path edge] (v2) -- (v3);
        \draw[path edge] (v3) -- (v4);

        \node[vertex] (p11) at (0,-2.25) {};
        \node[vertex] (p12) at (1.65,-2.25) {};
        \node[vertex] (p13) at (3.3,-2.25) {};
        \draw[path edge] (p11) -- (p12);
        \draw[path edge] (p12) -- (p13);
        \draw[attach edge] (v1) -- (p11);

        \node[vertex] (p21) at (1.65,-1.25) {};
        \draw[attach edge] (v2) -- (p21);

        \node[vertex] (p31) at (3.3,-1.25) {};
        \node[vertex] (p32) at (4.95,-1.25) {};
        \draw[path edge] (p31) -- (p32);
        \draw[attach edge] (v3) -- (p31);
    \end{tikzpicture}
}

\newcommand{\HTHCPathDXTemplate}[2]{%
    \begin{tikzpicture}[
        vertex/.style={circle,draw,inner sep=0pt,minimum size=6mm,font=\scriptsize},
        path edge/.style={->,thick,>=Stealth},
        attach edge/.style={->,densely dotted,thick,>=Stealth}
    ]
        \node[vertex] (v1) at (0,0) {#1};
        \node[vertex] (v2) at (1.65,0) {#1};
        \node[vertex] (v3) at (3.3,0) {#1};
        \node[vertex] (v4) at (4.95,0) {#1};

        \draw[path edge] (v1) -- (v2);
        \draw[path edge] (v2) -- (v3);
        \draw[path edge] (v3) -- (v4);

        \node[vertex] (p11) at (0,-2.25) {#2};
        \node[vertex] (p12) at (1.65,-2.25) {#2};
        \node[vertex] (p13) at (3.3,-2.25) {#2};
        \draw[path edge] (p11) -- (p12);
        \draw[path edge] (p12) -- (p13);
        \draw[attach edge] (v1) -- (p11);

        \node[vertex] (p21) at (1.65,-1.25) {#2};
        \draw[attach edge] (v2) -- (p21);

        \node[vertex] (p31) at (3.3,-1.25) {#2};
        \node[vertex] (p32) at (4.95,-1.25) {#2};
        \draw[path edge] (p31) -- (p32);
        \draw[attach edge] (v3) -- (p31);
    \end{tikzpicture}
}

\newcommand{\HTHCPathXBlueHeadsDTemplate}{%
    \begin{tikzpicture}[
        vertex/.style={circle,draw,inner sep=0pt,minimum size=6mm,font=\scriptsize},
        blue vertex/.style={vertex,fill=blue!45},
        path edge/.style={->,thick,>=Stealth},
        attach edge/.style={->,densely dotted,thick,>=Stealth}
    ]
        \node[vertex] (v1) at (0,0) {$\X$};
        \node[vertex] (v2) at (1.65,0) {$\X$};
        \node[vertex] (v3) at (3.3,0) {$\X$};
        \node[vertex] (v4) at (4.95,0) {$\X$};

        \draw[path edge] (v1) -- (v2);
        \draw[path edge] (v2) -- (v3);
        \draw[path edge] (v3) -- (v4);

        \node[blue vertex] (p11) at (0,-2.25) {$\B$};
        \node[vertex] (p12) at (1.65,-2.25) {$\D$};
        \node[vertex] (p13) at (3.3,-2.25) {$\D$};
        \draw[path edge] (p11) -- (p12);
        \draw[path edge] (p12) -- (p13);
        \draw[attach edge] (v1) -- (p11);

        \node[blue vertex] (p21) at (1.65,-1.25) {$\B$};
        \draw[attach edge] (v2) -- (p21);

        \node[blue vertex] (p31) at (3.3,-1.25) {$\B$};
        \node[vertex] (p32) at (4.95,-1.25) {$\D$};
        \draw[path edge] (p31) -- (p32);
        \draw[attach edge] (v3) -- (p31);
    \end{tikzpicture}
}

\newcommand{\HTHCPathBadTwoColoringTemplate}{%
    \begin{tikzpicture}[
        vertex/.style={circle,draw,inner sep=0pt,minimum size=6mm,font=\scriptsize},
        red vertex/.style={vertex,fill=red!35},
        blue vertex/.style={vertex,fill=blue!45},
        path edge/.style={->,thick,>=Stealth},
        attach edge/.style={->,densely dotted,thick,>=Stealth}
    ]
        \node[red vertex] (v1) at (0,0) {$\R$};
        \node[blue vertex] (v2) at (1.65,0) {$\B$};
        \node[red vertex] (v3) at (3.3,0) {$\R$};
        \node[red vertex] (v4) at (4.95,0) {$\R$};

        \draw[path edge] (v1) -- (v2);
        \draw[path edge] (v2) -- (v3);
        \draw[path edge] (v3) -- (v4);

        \node[blue vertex] (p11) at (0,-2.25) {$\B$};
        \node[red vertex] (p12) at (1.65,-2.25) {$\R$};
        \node[blue vertex] (p13) at (3.3,-2.25) {$\B$};
        \draw[path edge] (p11) -- (p12);
        \draw[path edge] (p12) -- (p13);
        \draw[attach edge] (v1) -- (p11);

        \node[red vertex] (p21) at (1.65,-1.25) {$\R$};
        \draw[attach edge] (v2) -- (p21);

        \node[blue vertex] (p31) at (3.3,-1.25) {$\B$};
        \node[red vertex] (p32) at (4.95,-1.25) {$\R$};
        \draw[path edge] (p31) -- (p32);
        \draw[attach edge] (v3) -- (p31);
    \end{tikzpicture}
}

\newcommand{\DrawHTHCPlainLabeled}[1]{%
    \node[vertex] (#1v1) at (0,0) {$v_1$};
    \node[vertex] (#1v2) at (1.65,0) {$v_2$};
    \node[vertex] (#1v3) at (3.3,0) {$v_3$};
    \node[vertex] (#1v4) at (4.95,0) {$v_4$};

    \draw[path edge] (#1v1) -- (#1v2);
    \draw[path edge] (#1v2) -- (#1v3);
    \draw[path edge] (#1v3) -- (#1v4);

    \node[vertex] (#1p11) at (0,-2.25) {$p_{1,1}$};
    \node[vertex] (#1p12) at (1.65,-2.25) {$p_{1,2}$};
    \node[vertex] (#1p13) at (3.3,-2.25) {$p_{1,3}$};
    \draw[path edge] (#1p11) -- (#1p12);
    \draw[path edge] (#1p12) -- (#1p13);
    \draw[attach edge] (#1v1) -- (#1p11);

    \node[vertex] (#1p21) at (1.65,-1.25) {$p_{2,1}$};
    \draw[attach edge] (#1v2) -- (#1p21);

    \node[vertex] (#1p31) at (3.3,-1.25) {$p_{3,1}$};
    \node[vertex] (#1p32) at (4.95,-1.25) {$p_{3,2}$};
    \draw[path edge] (#1p31) -- (#1p32);
    \draw[attach edge] (#1v3) -- (#1p31);
}

\newcommand{\DrawHTHCPlainBare}[1]{%
    \node[vertex] (#1v1) at (0,0) {};
    \node[vertex] (#1v2) at (1.65,0) {};
    \node[vertex] (#1v3) at (3.3,0) {};
    \node[vertex] (#1v4) at (4.95,0) {};

    \draw[path edge] (#1v1) -- (#1v2);
    \draw[path edge] (#1v2) -- (#1v3);
    \draw[path edge] (#1v3) -- (#1v4);

    \node[vertex] (#1p11) at (0,-2.25) {};
    \node[vertex] (#1p12) at (1.65,-2.25) {};
    \node[vertex] (#1p13) at (3.3,-2.25) {};
    \draw[path edge] (#1p11) -- (#1p12);
    \draw[path edge] (#1p12) -- (#1p13);
    \draw[attach edge] (#1v1) -- (#1p11);

    \node[vertex] (#1p21) at (1.65,-1.25) {};
    \draw[attach edge] (#1v2) -- (#1p21);

    \node[vertex] (#1p31) at (3.3,-1.25) {};
    \node[vertex] (#1p32) at (4.95,-1.25) {};
    \draw[path edge] (#1p31) -- (#1p32);
    \draw[attach edge] (#1v3) -- (#1p31);
}

\newcommand{\DrawHTHCMiniBox}[3]{%
    \coordinate (#3center) at (#1,#2);
    \coordinate (#3attach) at ($(#3center)+(0,-0.52)$);
    \draw[instance box] ($(#3center)+(-0.62,-0.52)$) rectangle ($(#3center)+(0.62,0.42)$);
    \begin{scope}[
        shift={($(#3center)+(-0.45,0.18)$)},
        scale=0.18,
        transform shape,
        path edge/.style={-{Stealth[length=1.1pt,width=1pt]},line width=0.2pt},
        attach edge/.style={-{Stealth[length=1.1pt,width=1pt]},dash pattern=on 0.3pt off 0.45pt,line width=0.25pt}
    ]
        \DrawHTHCPlainBare{#3mini}
    \end{scope}
}

\newcommand{\HTHCFigureThree}{%
    \begin{figure}[t]
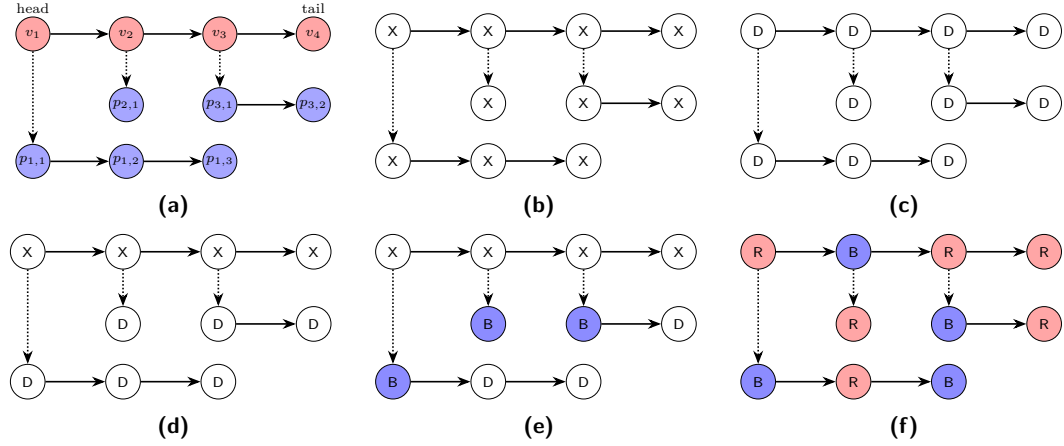

        \centering
        \begin{subfigure}[t]{0.31\textwidth}
            \centering
            \resizebox{\linewidth}{!}{\HTHCPathTemplate}
            \caption{}
        \end{subfigure}\hfill
        \begin{subfigure}[t]{0.31\textwidth}
            \centering
            \resizebox{\linewidth}{!}{\HTHCPathDXTemplate{$\X$}{$\X$}}
            \caption{}
        \end{subfigure}\hfill
        \begin{subfigure}[t]{0.31\textwidth}
            \centering
            \resizebox{\linewidth}{!}{\HTHCPathDXTemplate{$\D$}{$\D$}}
            \caption{}
        \end{subfigure}

        \vspace{1ex}

        \begin{subfigure}[t]{0.31\textwidth}
            \centering
            \resizebox{\linewidth}{!}{\HTHCPathDXTemplate{$\X$}{$\D$}}
            \caption{}
        \end{subfigure}\hfill
        \begin{subfigure}[t]{0.31\textwidth}
            \centering
            \resizebox{\linewidth}{!}{\HTHCPathXBlueHeadsDTemplate}
            \caption{}
        \end{subfigure}\hfill
        \begin{subfigure}[t]{0.31\textwidth}
            \centering
            \resizebox{\linewidth}{!}{\HTHCPathBadTwoColoringTemplate}
            \caption{}
        \end{subfigure}
        \caption{An $\HTHC$ instance with a top path of length four and attached bottom paths of lengths three, one, two, and zero from head to tail, respectively. Solid left-to-right arrows denote path connections, and dotted downward arrows denote attachments. (a) shows the input color for each vertex: all vertices in the top path have red input color, and all vertices in the bottom path have blue input color. (b)(c)(d)(e)(f) are possible output colorings for (a) that violate constraints 1,2,3,4,5 in $\HTHC$, respectively.}
        \label{fig:hthc-top-bottom-paths}
    \end{figure}
}

\newcommand{\HTHCFigureFour}{%
    \begin{figure}[t]
        \centering
        \begin{tikzpicture}[
            vertex/.style={circle,draw,inner sep=0pt,minimum size=5.5mm,font=\scriptsize},
            path edge/.style={->,thick,>=Stealth},
            attach edge/.style={->,densely dotted,thick,>=Stealth},
            instance box/.style={draw,thick}
        ]
            \DrawHTHCPlainLabeled{base}

            \DrawHTHCMiniBox{0}{1.7}{attone}
            \DrawHTHCMiniBox{1.65}{1.7}{atttwo}
            \DrawHTHCMiniBox{3.3}{1.7}{attthree}
            \DrawHTHCMiniBox{4.95}{1.7}{attfour}

            \draw[attach edge] (basev1) -- (attoneattach);
            \draw[attach edge] (basev2) -- (atttwoattach);
            \draw[attach edge] (basev3) -- (attthreeattach);
            \draw[attach edge] (basev4) -- (attfourattach);
        \end{tikzpicture}
        \caption{Fine-grained stacking construction. The large two-level $\HTHC$ instance is the base instance. Each vertex in the top path has one two-level $\HTHC$ instance attached, shown as a small boxed copy above the vertex.}
        \label{fig:hthc-top-attachments}
    \end{figure}
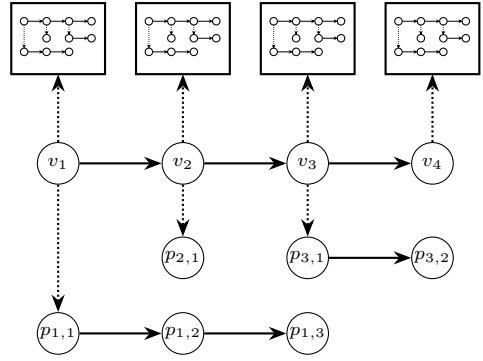
}

\HTHCFigureThree

Now we show how to craft the constraints to generate the effect described above. First, we follow the narrative of \cite{rosenbaum2020seeing} and add a variant of two-coloring constraints for each top and bottom path: Each vertex receives a red/blue input color, and the tail of a path must output its input color. Other vertices should output a color that is the same as their out-neighbors along the path. These constraints together require the head of a path to output the input color of the tail.

At the same time, we want an algorithm to be able to escape from the top path or some bottom paths when they are too long; otherwise, the complexity will be a clear $\Theta(n)$. We introduce two additional colors for each vertex: color $\D$ for decline, and color $\X$ for exempt. These two colors will be used when the algorithm thinks a path is too long. To make sure that an algorithm cannot abuse these colors, we add the following constraints to $\HTHC$: \begin{itemize}
    \item[1.] A vertex in any bottom path cannot be colored $\X$;
    \item[2.] A vertex in the top path cannot be colored $\D$;
    \item[3.] If a top vertex is colored $\X$, then either no bottom path is attached to it, or the head of its attached bottom path is not $\D$;
    \item[4.] If a vertex is not colored $\D$, then its descendant in the directed path is not colored $\D$.
    \item[5.] After removing vertices with colors $\D$ and $\X$, the set of directed paths follows the two-coloring constraints.
\end{itemize}

The first two constraints prevent the colorings in \cref{fig:hthc-top-bottom-paths}(b) and (c) that color all vertices with $\X$ or color all vertices with $\D$. The third constraint prevents the coloring in \cref{fig:hthc-top-bottom-paths}(d): if an algorithm tries to escape from the top path by putting an $\X$ color on some top-level vertex $v_i$, then the algorithm is required not to color the head of $P_i$ with $\D$ if it exists. With this in mind, constraint 4 then says that, if some vertex is not colored $\D$, so is its out-neighbor. This constraint propagates from the head of $P_i$ along the directed edges, and finally implies that any vertex in $P_i$ should not be colored $\D$ or $\X$. Finally, we need to color the path following the two-coloring constraints to satisfy constraint 5, so the head of $P_i$ must travel along the path to fetch the input color of the tail. 

The previous argument holds only when the top path uses color $\X$, but if it does not use $\X$, it cannot use $\D$ either from constraint 2, so according to constraint 5, the head of the top path needs to travel along the top path to fetch the input color on the tail. This gives the guarantee that we want: The algorithm needs to either visit all vertices in the top path to get the output for $v_1$, or visit all vertices in one of the bottom paths to get the output of the head of the bottom path.

To generalize the result to larger $k$, instead of a bottom path, one can attach an $\HTHC$ instance to each vertex in the top path. Following a similar argument, the algorithm is required to either solve the two-coloring along the top path or solve one of the attached $\HTHC$ instances. For instance, if we replace the bottom path with a two-level $\HTHC$ instance, we will end up with a three-level $\HTHC$ instance. In this case, when the top path has length $\ell$, the complexity will be $\tilde \Theta(\min(\ell, (n/\ell)^{1/2}))$, and the final complexity is $\tilde \Theta(n^{1/3})$ attained by $\ell = \Theta(n^{1/3})$.

As in the polylogarithmic setting, \cite{rosenbaum2020seeing} additionally constructs an input labeling called colored tree labeling to generalize the LCL to general bounded-degree trees and graphs. We follow the same approach in \cref{def:poly-t-labeling,def:poly-t-labeling-consistency}.

\subsubsection{The Stacking}

From the previous review, we can see how the stacking strategy in $\HTHC$ differs from that in $\PathToLeafK$: In the $\HTHC$ stacking, an algorithm only needs to solve one instance among all attached instances along the path, whereas in $\PathToLeafK$, the algorithm is required to solve all attached instances for each vertex along the path.

Same as the approach in the polylogarithmic setting, before the stacking, we first change the input and output format of $\HTHC$ to make sure that the output bit on the head of the top path $v_1$ must depend on inputs on $n^{\Omega(1)}$ vertices. We provide an additional input bit $b\insub(v)$ for each vertex in the instance, and we require any vertex with color other than $\D$ to output a bit $b\outsub(v)$ which is the xor of $b\insub(w)$ for all vertices $w$ in the following path starting from $v$: Suppose currently we are at some vertex $u$ with color other than $\D$. \begin{itemize}
    \item If $u$ is of output color $\X$, go to the head of its attached bottom path. If it has no bottom path attached, stop.
    \item If $u$ is of output color red or blue, go to its out-neighbor in the directed path. If it is the tail, stop.
\end{itemize}
According to constraints 3 and 4, this path will never visit vertices of color $\D$. At the same time, the length of this path from $v_1$ will be $\Omega(n^{1/2})$ for the $\HTHC$ construction above, since the path either never visits vertices of color $\X$ and visits all vertices in the top path, or visits some vertex with color $\X$ and visits all vertices of the attached bottom path. With the input $b\insub(v)$, we can remove the two-coloring constraint and the red/blue color input without changing the complexity. In \cref{sec:poly} we will identify $\R$ and $\B$ as a single color $\U$.

Now we try to stack $\HTHC$ instances. A possible first approach, similar to the stacking in the polylogarithmic setting, is to attach a $\HTHC$ instance, say of complexity $\tilde\Theta(m^{1/3})$, to every vertex in a base $\HTHC$ instance of complexity $\tilde\Theta(m^{1/2})$ of size $m$, and hope that it generates an LCL of probe complexity $\Theta(m^{1/2 + 1/3}) = \Theta(m^{5/6})$. Unfortunately, since the instance has size blown up to $n = m^2$ when each attached instance has size $m$, this is not true. By balancing the size of the base instance and attached instances, it turns out that the complexity is still $\tilde \Theta(m^{1/2})$.

Our approach is to deploy a more fine-grained stacking strategy, stacking different $\HTHC$ instances for different levels. Next, we will give one concrete example to show this idea. In the example, there will be a base instance, which is a two-level $\HTHC$ instance, to which other instances are attached. Instead of attaching an instance for each vertex in the base instance, we only attach one instance, specifically a fresh two-level $\HTHC$ instance, for each vertex in the top path of the base instance. Vertices in the bottom path of the base instances have no instance attached. See \cref{fig:hthc-top-attachments} for an illustration. As the construction in the polylogarithmic regime, the input bit on the vertices of the top path in the base instance will be a twisted one $b'\insub(v)$, and the algorithm needs to recover the actual input bit $b\insub(v)$ by xoring $b'\insub(v)$ and the output bit of the root of the attached instance, namely the head of the top path.

\HTHCFigureFour

Denote $\ell$ as the length of the top path in the base instance, and we assume that the algorithm knows the size of the attached two-level $\HTHC$ instances, denoted as $s_1,s_2,\dots,s_\ell$, and also $|P_i|$ for $1 \le i \le \ell$. Consider two possibilities for the output:

\begin{itemize}
    \item If the top path in the base instance has no color $\X$, then apart from $\ell$ input bits, the algorithm is required to solve $\ell$ two-level $\HTHC$ instances. The total number of probes will be $\tilde \Theta(\sum \sqrt{s_i})$. By Jensen's inequality, an adversary that wants to maximize the probe complexity will make $s_i$ roughly equal, leading to $\tilde\Theta(\ell (n/\ell)^{1/2})$ probe complexity;
    \item Otherwise, the algorithm is required to solve one of the bottom paths in the base instance. The algorithm can choose the shortest one to solve, and an adversary will make the bottom paths of equal length, leading to $\Theta(n/\ell)$ probe complexity.
\end{itemize}

Finally, an algorithm can choose one of the above two choices, so an adversary will take $\ell$ to maximize $\tilde \Theta(\min(\ell (n/\ell)^{1/2}, n/\ell))$. The maximum value is $\tilde \Theta(n^{2/3})$ taken by $\ell = \Theta(n^{1/3})$, and we would expect that the complexity of such an LCL is $\tilde \Theta(n^{2/3})$.

For the general case, we can replace the bottom path in the base instance with an instance of complexity $\tilde \Theta(n^q)$, and replace the attached two-level $\HTHC$ instance with an instance of complexity $\tilde \Theta(n^p)$ for $p < q$. Then, following the above analysis, the best algorithm would use $\tilde\Theta(\min(\ell (n/\ell)^p, (n/\ell)^q))$ probes, and an adversary maximizing the probe complexity will take $\ell = n^{(q-p)/(1+q-p)}$, leading to complexity $\tilde \Theta(n^{q/(1+q-p)})$. This suffices to generate all rational powers, as we will show in \cref{lem:poly-rational-to-ordered-tree}.

\subsubsection{Upper Bound and Lower Bound}

The lower bound analysis for the polynomial setting is much closer to the preceding discussion. We can follow the adversarial choice in the argument above to generate the underlying graph for the lower bound instances. There is still the issue where an algorithm can choose output colors after reading several input bits to be in favor of the output distribution for the root, yet we can use the same approach as in the polylogarithmic setting: we construct the distribution where each instance in the distribution has a definite output bit at the root, no matter how the underlying coloring is chosen. And then we use an induction to prove that for any vertex set of size $o(n^{q/(1+q-p)})$, revealing the input bits on those vertices does not provide any advantage to predict the output bit at the root.

The main obstacle for the upper bound proof is to remove the assumption that the algorithm knows the size of each attached instance and bottom path. To see why this is the case, we first provide a sketch of the matching $\tilde O(n^{1/2})$ upper bound algorithm for two-level $\HTHC$: For each vertex $v_i$ in the top path, consider the sub-path $v_i,v_{i+1},\dots,v_{i+\lceil 2\sqrt{n}\rceil}$. When $i+\lceil 2\sqrt{n}\rceil > \ell$, we can just color it with the input color at the tail of the top path with $O(\sqrt{n})$ probes. For the other case, at least half of the vertices in the sub-path have a bottom path of length at most $\sqrt{n}$ attached. In the algorithm, each vertex in the top path uses its local randomness to add itself to a sample set $S$ with probability $\Theta(\log n / \sqrt{n})$. Define $S_i = \{v_i,v_{i+1},\dots,v_{i+\lceil 2\sqrt{n}\rceil}\} \cap S$. By Chernoff bound, with high probability, $|S_i| = O(\log n)$ and some vertex in $S_i$ has a bottom path of length at most $\sqrt{n}$ attached. The algorithm finds the first vertex in $S_i$ that has an attached bottom path of length at most $\sqrt{n}$, which can be tested within $\tilde O(\sqrt{n})$ probes, and colors it with $\X$, generating the output accordingly. This uses $\tilde O(\sqrt{n})$ probes and concludes the algorithm.

Returning to the $\tilde \Theta(n^{2/3})$ example above, for the bottom paths in the base instance, we can use a similar approach with an alternate parameter: Each vertex in the top path of the base instance inspects $2n^{1/3}$ descendants, and if this does not hit the tail, use the sample set $S$ to find a vertex with an attached bottom path of length at most $n^{2/3}$. For the general case where the bottom paths are replaced by instances of complexity $\tilde \Theta(n^q)$, we will run the algorithm for attached instances with a \emph{maximum possible size} parameter set to $n^{1/(1+q-p)}$, and we require the algorithm to fail within $\tilde \Theta(n^{q/(1+q-p)})$ probes when the actual instance size is larger than the maximum possible size.

For the attached $\HTHC$ instances, things become more complicated since the algorithm described for $\tilde \Theta(n^{1/2})$ upper bound depends heavily on knowing the instance size. We use a standard binary lifting to guess the size of each instance, and this introduces an additional problem: To get the output for each vertex in the attached instance, each query will start an independent binary lifting guessing procedure, and different vertices may become satisfied and stop at different size guesses, raising potential inconsistency issues. To resolve this issue, we will give an additional size parameter to the algorithm, called the \emph{ideal size} of the instance, and most of the algorithm's decisions will rely on the ideal size parameter rather than the size guess in the binary lifting. This maximizes consistency between independent queries and allows us to make additional changes to the algorithm to provide full consistency within the desired probe complexity budget.

\section{Preliminaries}
\label{sec:prelim}

In this section, we provide basic definitions and formally define models and problems of interest.

\subsection{Graphs}

In this paper, we denote a graph as $G = (V, E)$, where $V$ is the vertex set and $E$ is the edge set. We will normally use $n$ to denote the number of vertices in a graph. Unless explicitly stated, all graphs are undirected without self-loops and parallel edges. For each vertex $v \in V$, we denote its degree by $\deg_G(v)$, its neighborhood by $\Gamma_G(v)$, and its $r$-hop neighborhood, defined as the set of vertices with distance at most $r$ to $v$, as $\Gamma^r_G(v)$. For a vertex set $V' \subseteq V$, define the induced subgraph of $V'$ in $G$ as $G[V'] = (V', \{(u,v) \in E: u, v \in V'\})$.

For all graph problems we will consider in the rest of the paper, we assume that each vertex $v \in V$ is given a unique identifier from the range $\{1,2,\dots,n^\alpha\}$ for some arbitrary fixed $\alpha \geq 1$, and all vertices have degree at most $\Delta$ for some fixed constant $\Delta \in \mathbb{N}$. Additionally, we assume that the input graph is equipped with a \emph{port ordering}. This means that, for each vertex $v$ and incident edge $(v, w)$, there is an associated number $p(v, w) \in \{1,2,\dots,\deg(v)\}$, the \emph{port number} of $(v, w)$ for $v$, such that $p$ is a bijection between edges incident to $v$ and $\{1,2,\dots,\deg(v)\}$. Notice that $p(v,w)$ and $p(w,v)$ may not be equal.

\subsection{Locally Checkable Labeling}

Here, we give the formal definition of a locally checkable labeling problem.

\begin{definition}[Locally Checkable Labeling]\label{def:LCL}
    For a fixed $\Delta$, a locally checkable labeling (LCL) problem $\Pi$ is a graph problem that can be specified by a tuple $(\Sigma\insub,\Sigma\outsub,r,C)$ such that
\begin{itemize}
	\item $\Sigma_{\mathrm{in}}$ and $\Sigma_{\mathrm{out}}$ are finite sets of labels;
	\item $r$ is an arbitrary positive integer, called the \emph{checkability radius} of $\Pi$;
	\item $C$ is the set of allowed local configurations, denoted by a finite set of tuples $(H = (V^H,E^H),x,\varphi\insub, \varphi\outsub)$, where:
	\begin{itemize}
		\item $H$ is a graph of degree at most $\Delta$, $x$ is a vertex of $H$, and each vertex has distance at most $r$ from $x$ in $H$;
		\item $\varphi\insub$ is a function from $V^H$ to $\Sigma\insub$, and $\varphi\outsub$ is a function from $V^H$ to $\Sigma\outsub$.
	\end{itemize}
\end{itemize}

    An \emph{instance} of an LCL $\Pi(\Sigma\insub, \Sigma\outsub, r, C)$ is a graph $G=(V,E)$ with $n$ vertices and maximum degree $\Delta$, and additionally an input labeling $\phi\insub: V \to\Sigma_{\mathrm{in}}$. The graph $G$ may not be connected. A \emph{solution} to the instance $(G, \phi\insub)$ is an output labeling $\phi\outsub: V \to \Sigma_{\mathrm{out}}$ such that for every vertex $v \in V$, there exists $(H=(V^H,E^H), x, \varphi\insub, \varphi\outsub) \in C$ and a bijection $f$ between $\Gamma_G^r(v)$ and $V^H$ with all the following properties:
    \begin{itemize}
        \item $f$ specifies a graph isomorphism between $G[\Gamma_G^r(v)]$ and $H$;
        \item $\forall u \in \Gamma_G^r(v)$, $\phi\insub(u) = \varphi\insub(f(u))$ and $\phi\outsub(u) = \varphi\outsub(f(u))$.
    \end{itemize}
\end{definition}

\begin{remark}
    In the rest of the paper, instead of specifying the set of allowed local configurations $C$ for an LCL explicitly, we will demonstrate a $\LOCAL$ algorithm $\mathcal{A}$ that collects the information for the $r$-hop neighborhood, including the graph structure, input labels, output labels and port numbers, but \emph{excluding the unique ID on each vertex}, and checks whether the output labeling satisfies the local constraints for the LCL. This implicitly defines $C$ as the set of local configurations that pass the check in $\mathcal{A}$.
\end{remark}

\subsection{Computation Models}

In the following, we formally define the models we consider in this paper, the randomized LCA and $\Vol$ model.

\begin{definition}[Randomized LCA]\label{def:rand-LCA}
    A randomized local computation algorithm (LCA) $\mathcal{A}$ for an LCL is given access to a deterministic graph oracle $\mathcal{O}^G$ for the input graph $G$ and input labeling $\phi\insub$, a tape $r$ of random bits shared across queries, and local read-write computation memory independent across queries. When given a query vertex $v$, $\mathcal{A}$ must assign an output label for $v$ using $r$, $\mathcal{O}^G$, and the knowledge of the size $n$ of graph $G$. The algorithm can access $\mathcal{O}^G$ by giving the unique identifier of a vertex $u$, and $\mathcal{O}^G(u)$ will return the input label $\phi\insub(u)$ along with its adjacency list and port number labeling in $G$. The algorithm is said to \emph{solve} an LCL $\Pi$ if for any instance $(G, \phi\insub)$ of $\Pi$, with probability at least $1-1/n$ over the random tape $r$, the algorithm generates an output label assignment to each vertex that forms a solution to the instance.
\end{definition}

\begin{definition}[Randomized $\Vol$]\label{def:rand-Vol}
    A randomized $\Vol$ algorithm $\mathcal{A}$ for an LCL is given access to a deterministic graph oracle $\mathcal{O}^G$ for the input graph $G$, input labeling $\phi\insub$, and an independent random string $r_v$ for each vertex $v$. The random strings are shared across queries. The algorithm is also equipped with local read-write computation memory independent across queries. When given a query vertex $v$, $\mathcal{A}$ must assign an output label for $v$ using $\mathcal{O}^G$ and the knowledge of the size $n$ of graph $G$. The algorithm can access $\mathcal{O}^G$ by giving the unique identifier of a vertex $u$, and $\mathcal{O}^G(u)$ will return the input label $\phi\insub(u)$, the random string $r_u$, along with its adjacency list and port number labeling in $G$, only when the following holds: \begin{itemize}
        \item Either $u$ is the vertex being queried, or some previous oracle access $\mathcal{O}^G(w)$ returns an adjacency list containing $u$.
    \end{itemize}
    The algorithm is said to \emph{solve} an LCL $\Pi$ if for any instance $(G, \phi\insub)$ of $\Pi$, with probability at least $1-1/n$ with respect to the local randomness $r_v$ for each vertex $v$, the algorithm generates an output labeling assignment to each vertex that forms a solution to the instance.
\end{definition}

The complexities of randomized LCA and $\Vol$ algorithms can be measured in several different ways. In this paper, we will use only the probe complexity, the maximum number of oracle accesses an algorithm uses to provide the output label for a single query, to measure an LCA or $\Vol$ algorithm. One reason is that it simplifies the analysis and enables us to focus on the effect of limited knowledge on the input graph in the algorithm design. Furthermore, it turns out that all algorithms given in this paper can also be implemented with time and space complexities roughly a factor of $O(\log n)$ larger than the probe complexities.

From the definition, we can see that a $\Vol$ algorithm is always an LCA, formalized by the following observation, but it is not the case in the reverse. It is still open whether $\Vol$ models and LCA are equivalent up to polynomial factors; see \cref{subsec:lca-vs-vol} for a related discussion.

\begin{observation}\label{obs:vol-weaker-than-lca}
    If a randomized $\Vol$ algorithm solves an LCL $\Pi$ with probe complexity $f(n)$ for some function $f$, then there is a randomized LCA that solves $\Pi$ with probe complexity $O(f(n))$.
\end{observation}
\begin{proof}
    Let $\mathcal{A}$ be the randomized $\Vol$ algorithm. We construct an LCA $\mathcal{B}$ that simulates $\mathcal{A}$ on the same query vertex. The LCA uses its shared random tape to generate the local random strings used by the $\Vol$ algorithm: for each vertex identifier $u$, reserve an infinite subsequence of bits on the shared tape for $r_u$. Since vertex identifiers come from a polynomial range in $n$, this reservation can be fixed in advance as a deterministic function of the identifier to make sure that the randomness is independent for every vertex.

    During the simulation, whenever $\mathcal{A}$ probes a vertex $u$, the LCA probes the same vertex through its graph oracle and additionally reads the block of random bits assigned to $u$. The $\Vol$ access rule only restricts which probes $\mathcal{A}$ is allowed to make; an LCA can make all of these probes as well. Therefore, the transcript of $\mathcal{B}$ has the same distribution as the transcript of $\mathcal{A}$ on every query and every instance. In particular, the output labels form a feasible solution with the same success probability. The number of graph-oracle probes is exactly the number of probes made by $\mathcal{A}$, so the probe complexity is $O(f(n))$.
\end{proof}

\begin{remark}\label{rem:las-vegas-to-monte-carlo}
    The algorithms described in the rest of this paper will be Las Vegas, with a guarantee of success, and the number of probes will vary with randomness. However, the randomized $\Vol$ model and randomized LCA are stated in Monte Carlo form, with a fixed probe budget and a small probability of failure. One can change the algorithm to Monte Carlo by running the algorithm with a fixed probing budget and output arbitrarily when the algorithm uses up its budget.
\end{remark}

\DeclareRobustCommand{\WalkToLeaf}{\ifmmode\mathop{\text{\normalfont\textsc{WalkToLeaf}}}\nolimits\else\textnormal{\textsc{WalkToLeaf}}\fi}
\DeclareRobustCommand{\SolvePTL}{\ifmmode\mathop{\text{\normalfont\textsc{SolvePTL}}}\nolimits\else\textnormal{\textsc{SolvePTL}}\fi}

\section{Construction in Polylogarithmic Complexity Regime}
\label{sec:log}

In this section, we construct LCLs $\PathToLeafK$ whose randomized probe complexities in both the LCA and $\Vol$ models are $\Theta(\log^k n)$ for every positive integer $k$. Here, $\mathsf{PTL}$ is an abbreviation of ``path to leaf''.

\subsection{Description of the LCL}

Before defining the LCL, we first define an input labeling called ``multi-level binary tree labeling'' in \cref{def:log-binary-tree-labeling} and define its associated local constraints in \cref{def:log-tree-labeling-consistency}. 
The purpose of this input labeling is to ``escape'' from the case where the input graph is irregular: As we will show in \cref{lem:log-pseudoforest}, all vertices satisfying the local constraints for the multi-level binary tree labeling, along with their outgoing edges, form a pseudo-forest in which each vertex has out-degree zero or two.

\begin{definition}[Multi-level binary tree labeling]\label{def:log-binary-tree-labeling}
Let $G=(V,E)$ be a graph of maximum degree at most $\Delta$, and let $\mathcal{P}=[\Delta]\cup\{\perp\}$, where $\perp$ is a special placeholder symbol. For a given positive integer $k$, a \emph{level-$k$ binary tree labeling} consists of a height labeling $h\colon V\to[k]$ and the following four labels for each $v\in V$: A parent $\P(v)\in\mathcal{P}$, a left child $\LC(v)\in\mathcal{P}$, a right child $\RC(v)\in\mathcal{P}$, and a new-instance child $\NC(v)\in\mathcal{P}$.
\end{definition}

\begin{remark}\label{rem:log-child-port-notation}
To make the set of possible inputs at each vertex constant, $\P(v),\LC(v),\RC(v)$, and $\NC(v)$ are given by the port number rather than the corresponding vertex ID in \cref{def:log-binary-tree-labeling}. However, to keep notation light, in the rest of this paper, we will abuse the notation a bit and use $\P(v), \LC(v),\RC(v),\NC(v)$ to also denote the corresponding neighbors reached through that port. If a port is $\perp$, then the corresponding vertex reached through that port is defined as $\perp$.
\end{remark}

\begin{definition}\label{def:log-tree-labeling-consistency}
For a level-$k$ binary tree labeling of $G$, a vertex $v$ is called \emph{consistent} if all the following conditions hold, and otherwise it is called \emph{inconsistent}:
\begin{enumerate}
    \item The non-$\perp$ ports among $\P(v),\LC(v),\RC(v),\NC(v)$ are pairwise distinct.
    \item For $w\in\{\LC(v),\RC(v),\NC(v)\} \backslash \{\perp\}$, $\P(w) = v$.
    \item Either both $\LC(v)$ and $\RC(v)$ are $\perp$, or neither of them is $\perp$.
    \item For $w\in\{\LC(v),\RC(v)\} \backslash \{\perp\}$, $h(w)=h(v)$.
    \item If $\NC(v)\ne\perp$, then $h(\NC(v))=h(v)-1$. In particular, if $h(v)=1$, then $\NC(v)=\perp$.
\end{enumerate}
We further use condition 3 to differentiate between consistent vertices: A consistent vertex is \emph{internal} if $\LC(v),\RC(v)\ne\perp$, and is a \emph{leaf} if $\LC(v)=\RC(v)=\perp$. For a consistent internal vertex $v$, define $\child(v,\L)=\LC(v)$ and $\child(v,\R)=\RC(v)$.
\end{definition}

\begin{lemma}\label{lem:log-pseudoforest}
For a graph $G$ and a level-$k$ binary tree labeling, define a directed graph $G' = (V,E')$, where $(u \to v) \in E'$ when $u$ is a consistent internal vertex and $v \in \{\LC(u), \RC(u)\}$. Then, each vertex in $G'$ has in-degree at most one and out-degree either zero or two. As a result, $G'$ is a pseudo-forest where each weak component has at most one directed cycle.
\end{lemma}

\begin{proof}
    For each vertex $v$, from constraint 2 of \cref{def:log-tree-labeling-consistency}, its only possible in-degree in $G'$ is from $\P(v)$ if $\P(v) \ne \perp$. From constraint 3 of \cref{def:log-tree-labeling-consistency}, each consistent internal vertex in $G'$ has out-degree two, while inconsistent vertices and consistent leaves have no out-neighbor. This concludes the statement.
\end{proof}

With \cref{def:log-binary-tree-labeling,def:log-tree-labeling-consistency}, we are ready to describe the LCL $\PathToLeafK$. 

\begin{definition}[$\PathToLeafK$]
    The problem $\PathToLeafK$ is defined in the following way: 

\textbf{Input:} A graph $G$, a level-$k$ binary tree labeling, and $b\insub(v)\in\{0,1\}$ for each $v\in V$.

\textbf{Output:} $d(v)\in\{\L,\R,\perp\}$ and $b\outsub(v)\in\{0,1\}$ for each $v\in V$.

\textbf{Constraint:} For each vertex $v\in V$, if $v$ is inconsistent, then $(d(v),b\outsub(v))=(\perp,0)$; Otherwise, $d(v) \in \{\L,\R\}$ for internal vertices $v$, $d(v) = \perp$ for leaves, and we have the following equation for $b\outsub(v)$:
\begin{equation}
    b\outsub(v)=b\insub(v)\xor b'(v) \xor b''(v).\label{eq:log-lcl-def}
\end{equation}
Here
\begin{gather}
    b'(v)=
    \begin{cases}
        0, & \NC(v)=\perp\\
        b\outsub(\NC(v)), & \NC(v)\ne\perp
    \end{cases}\label{eq:log-b'-def}
    \\
    b''(v)=
    \begin{cases}
        0, & \LC(v)=\RC(v) =\perp\\
        b\outsub(\child(v,d(v))), & \mathrm{otherwise}
    \end{cases}\label{eq:log-b''-def}
\end{gather}
    
\end{definition}

\begin{lemma}\label{lem:log-lcl}
    $\PathToLeafK$ is an LCL.
\end{lemma}

\begin{proof}
From the definition, since the maximum degree $\Delta$ and the parameter $k$ are constant, we know that the set of possible input and output labels has a constant size. Now we show that there is a $\LOCAL$ algorithm of radius one to check whether the local constraint is satisfied for each vertex. The consistency conditions in \cref{def:log-tree-labeling-consistency} inspect only $v$, its incident ports, and the labels of its neighbors. As a result, we can determine whether the current vertex is consistent, internal, or a leaf using input labels from the one-hop neighborhood. After that, constraints for the output label only involve the one-hop neighborhood. Hence, every constraint is checkable in radius one.
\end{proof}

In the rest of the section, we prove the following upper and lower bound results, which together imply \cref{thm:log-main}.

\begin{lemma}\label{lem:log-upper-bound}
    There is a randomized $\Vol$ algorithm that solves $\PathToLeafK$ in all bounded-degree graphs using $O(\log^k n)$ probes.
\end{lemma}

\begin{lemma}\label{lem:log-lower-bound}
    Any randomized LCA that solves $\PathToLeafK$ uses $\Omega(\log^k n)$ probes even when we restrict the instance to bounded-degree trees.
\end{lemma}

\begin{proof}[Proof of \cref{thm:log-main}]
    From \cref{obs:vol-weaker-than-lca}, the upper bound result \cref{lem:log-upper-bound} can be generalized from $\Vol$ to LCA, while the lower bound result \cref{lem:log-lower-bound} can be generalized from LCA to $\Vol$. As a consequence, in both models, we have a $\Theta(\log^k n)$ randomized algorithm for $\PathToLeafK$ on bounded-degree graphs and a matching lower bound for bounded-degree trees.
\end{proof}

\subsection{Upper Bound}

In this part, we prove \cref{lem:log-upper-bound}, providing a randomized $\Vol$ algorithm for $\PathToLeafK$ on bounded-degree graphs with probe complexity $O(\log^k n)$.

To begin with, we provide an equivalent formulation of the constraints in $\PathToLeafK$. Arbitrarily fix $d(v)$ for each consistent internal vertex $v$. For some consistent vertex $u$, apply $u \leftarrow \child(u, d(u))$ multiple times to generate a path, until either reaching a leaf, an inconsistent vertex, or returning to where it starts. The path is always a path in the graph $G'$ defined in \cref{lem:log-pseudoforest}, so according to \cref{lem:log-pseudoforest}, the procedure terminates in finite steps no matter how the $d$ output label is assigned. This procedure generates either a path from $u$ to some vertex $w$, or a simple cycle passing $u$. We will demonstrate the way to get rid of cycles, so let us focus on the other case. Applying \cref{eq:log-lcl-def} for each vertex along the path, we can conclude that the value of $b\outsub(u)$ is the xor of all $b\insub$ values along the path and all $b\outsub$ values of the new-instance children along the path. 

According to constraints 4 and 5 in \cref{def:log-tree-labeling-consistency}, all vertices along the path have the same height, while their new-instance children have height $h(v)-1$. If $d(v)$ is assigned in such a way that the path is of length $O(\log n)$, then one can recurse on all $O(\log n)$ new-instance children to get their $b\outsub$ value, and the recursion decreases the vertex height by one. This means that, when all possible paths generated by the procedure above have length $O(\log n)$, we get a $\Vol$ algorithm with probe complexity $O(\log^k n)$.

This is exactly the strategy of our randomized $\Vol$ algorithm given in \cref{alg:log-solve}, where the output pair $(d,z)$ of $\SolvePTL^k(v)$ will be the final value of $(d(v), b\outsub(v))$. We assume that each vertex has a uniform random direction $\rho(v)\in\{\L,\R\}$, and recall that in the $\Vol$ model, the local random direction $\rho(v)$ is fixed for each vertex $v$ and shared across queries.

The algorithm first assigns $(\perp, 0)$ for inconsistent vertices. For other vertices, it calls $\WalkToLeaf(v)$, described in \cref{alg:log-walk}, to find the path from $v$ to either an inconsistent vertex or a leaf using only $\LC$ and $\RC$ ports. The following observation is straightforward from the algorithm description.

\begin{observation}\label{obs:alg-log-walk-in-G'}
    For any consistent vertex $v$, the path $\pi_v$ generated by $\WalkToLeaf(v)$ is a path in the graph $G'$ defined in \cref{lem:log-pseudoforest}.
\end{observation}

\cref{alg:log-walk} treats the randomness $\rho(v)$ as the initial assignment for $d(v)$ and traces the path $\pi_v$ by following $\rho$. If the path ends at a vertex for which the termination condition holds, then we assign $d(v)$ with $\rho(v)$ in \cref{alg:log-solve}. For the other case when $\rho$ directs $v$ into a cycle, corresponding to the branch in line \ref{line:alg-log-walk-in-cycle}, \cref{alg:log-walk} will abandon all visited vertices and use the direction other than $\rho(v)$ to restart a walk. The other direction is guaranteed to arrive at a leaf or an inconsistent vertex, as otherwise we can find some vertex with in-degree two in the graph $G'$ defined in \cref{lem:log-pseudoforest}. In this case, we assign $d(v)$ to that other direction. After getting the path $\pi_v$, \cref{alg:log-solve} assigns $b\outsub(v)$ by what we have discussed above: gathers all $b\insub$ values in $\pi_v$, and gets all $b\outsub$ values by recursion.

\begin{algorithm}[p]
\caption{$\WalkToLeaf(v)$}
\label{alg:log-walk}
\begin{algorithmic}[1]
\REQUIRE A random direction $\rho(u) \in \{\L, \R\}$ for each vertex $u$ in the graph.
\ENSURE A path $\pi$ from $v$ to a leaf or inconsistent vertex using only $\LC$ and $\RC$ ports.
\STATE $u\gets v$, $\pi\gets()$
\WHILE{$u$ is a consistent internal vertex}
    \IF{$u=v$ and $\pi$ is nonempty} \label{line:alg-log-walk-in-cycle}
        \STATE $\pi\gets(v)$
        \STATE $a\gets$ the unique direction in $\{\L,\R\}\backslash\{\rho(v)\}$
    \ELSE
        \STATE append $u$ to the end of $\pi$
        \STATE $a\gets\rho(u)$
    \ENDIF
    \STATE $u\gets\child(u,a)$
\ENDWHILE
\STATE append $u$ to $\pi$
\RETURN $\pi$
\end{algorithmic}
\end{algorithm}

\begin{algorithm}[p]
\caption{$\SolvePTL^k(v)$}
\label{alg:log-solve}
\begin{algorithmic}[1]
\ENSURE The output label $(d(v), b\outsub(v))$ for $v$
\IF{$v$ is inconsistent}
    \RETURN $(\perp,0)$
\ENDIF
\STATE $\pi\gets\WalkToLeaf(v)$
\IF{$v$ is a leaf}
    \STATE $d\gets\perp$
\ELSE
    \STATE $d\gets$ the first direction taken by $\pi$ from $v$ \label{line:alg-log-solve-assign-direction}
\ENDIF
\STATE $z\gets 0$
\FOR{each consistent vertex $u$ in $\pi$}
    \STATE $z\gets z\xor b\insub(u)$
    \IF{$\NC(u)\ne\perp$}
        \STATE $(-,z')\gets\SolvePTL^k(\NC(u))$
        \STATE $z\gets z\xor z'$
    \ENDIF
\ENDFOR
\RETURN $(d,z)$
\end{algorithmic}
\end{algorithm}

We first show that \cref{alg:log-walk} behaves correctly within $O(\log n)$ probes with high probability, which leads to the $O(\log^k n)$ probe complexity of \cref{alg:log-solve}.

\begin{lemma}\label{lem:log-algo-log-walk-probes}
    For each consistent vertex $v \in G$, $\WalkToLeaf(v)$ generates a path $\pi_v$ from $v$ to a leaf or an inconsistent vertex within $O(\log n)$ probes with probability $1-1/n^3$. As a result, $|\pi_v| = O(\log n)$ with probability $1-1/n^3$.
\end{lemma}

\begin{proof}
    By \cref{obs:alg-log-walk-in-G'}, all vertices visited by \cref{alg:log-walk} lie on a directed walk in the graph $G'$ from \cref{lem:log-pseudoforest}. We first record a simple consequence of the in-degree bound in \cref{lem:log-pseudoforest}. Fix a vertex $s$ and expose a walk that starts at $s$ and, whenever it is at a consistent internal vertex, follows an independently chosen random direction in $\{\L,\R\}$. Before the walk reaches a leaf, reaches an inconsistent vertex, or returns to $s$, two different direction sequences of length $t$ must end at two different vertices. Indeed, if two such prefixes first meet at a vertex other than $s$, then that vertex has two distinct in-neighbors in $G'$, contradicting \cref{lem:log-pseudoforest}. Hence, for $t$ steps, the number of direction sequences that keep the walk ``alive'', meaning that all vertices traveled by the walk are consistent internal vertices, is at most $n$, and the probability that this happens is at most $n2^{-t}$.

    Choose a sufficiently large constant $C$. Applying the previous bound with $t=C\log n$ shows that the walk followed by \cref{alg:log-walk} lasts more than $C\log n$ steps with probability at most $1/(2n^3)$ before the walk arrives at a leaf, arrives at an inconsistent vertex, or finds a cycle and triggers the branch in line \ref{line:alg-log-walk-in-cycle}. If this walk reaches a leaf or an inconsistent vertex, the algorithm stops. Otherwise, the first walk has returned to $v$, so it has found the unique directed cycle through $v$. The algorithm then restarts from $v$ using the other child of $v$ as the first step. By \cref{lem:log-pseudoforest}, this second walk cannot return to $v$: the predecessor of $v$ on the first cycle already accounts for the only possible in-neighbor of $v$ in $G'$. It also cannot enter a directed cycle not containing $v$, because entering such a cycle from outside would give the entry vertex two distinct in-neighbors in $G'$. Therefore, the second walk must end at a leaf or an inconsistent vertex.

    The same counting argument, now with the first direction fixed and all later directions random, shows that this second walk has length more than $C\log n$ with probability at most $1/(2n^3)$ after increasing $C$ if necessary. By a union bound, with probability at least $1-1/n^3$, both walks considered by the algorithm have length $O(\log n)$. The algorithm probes only the vertices on these walks, up to a constant number of additional oracle accesses needed to check consistency and read the relevant ports. Thus \cref{alg:log-walk} uses $O(\log n)$ probes and returns a path $\pi_v$ of length $O(\log n)$ with probability at least $1-1/n^3$.
\end{proof}

\begin{corollary}\label{cor:log-algo-log-solve-probes}
    For each $v \in G$, $\SolvePTL^k(v)$ probes $O(\log^k n)$ vertices with probability $1-1/n^2$.
\end{corollary}

\begin{proof}
From \cref{lem:log-algo-log-walk-probes} and union bound over all vertices, with probability $1-1/n^2$, for any consistent vertex $v$ in the instance, \cref{alg:log-walk} generates a path $\pi_v$ of length at most $C \log n$ using at most $C \log n$ probes for some constant $C$. In the following, we assume that the event happens.

Let $T(n,h^\star)$ be the maximum possible number of probes for \cref{alg:log-solve} to generate $(d(v), b\outsub(v))$ for a consistent vertex $v$ with $h(v) = h^\star$ in an $n$-vertex instance. For the base case $T(n,1)$, from constraint 5 in \cref{def:log-tree-labeling-consistency}, none of the vertices in $\pi_v$ have a new-instance child, so there is no recursion. In conclusion, the number of probes in \cref{alg:log-solve} is bounded by $C \log n$.

For $h^\star \ge 2$, apart from the call to \cref{alg:log-walk}, \cref{alg:log-solve} recurses on $\NC(u)$ for each $u \in \pi_v$. Constraints 4 and 5 of \cref{def:log-tree-labeling-consistency} tell us that all vertices in $\pi_v$ have height $h^\star$, while all their new-instance children have height $h^\star - 1$. As a consequence, we have the following recurrence:
\begin{equation}
    T(n,h^\star)\le C\log n\cdot T(n,h^\star - 1)+C\log n,\label{eq:log-recurse}
\end{equation}
with $T(n,1)\le C\log n$. As $C$ is constant with respect to $n$, solving \cref{eq:log-recurse} gives $T(n,h^\star)=O(\log^{h^\star} n)$. Finally, since the maximum height of a vertex is $k$ in $\PathToLeafK$, the probe complexity of \cref{alg:log-solve} is $O(\log^k n)$. 
\end{proof}

Now we show that \cref{alg:log-solve} assigns labels satisfying all constraints in $\PathToLeafK$.

\begin{lemma}\label{lem:log-algo-log-walk-correctness}
    If $\pi_v = (v,w,p_1,p_2,\dots,p_\ell)$ is the output of $\WalkToLeaf(v)$ for a consistent internal vertex $v$ for some $\ell \ge 0$ and $w$ is consistent, then $\WalkToLeaf(w)$ outputs $\pi_w = (w, p_1,p_2,\dots, p_\ell)$.
\end{lemma}

\begin{proof}
    Let $a$ be the first direction taken by $\pi_v$, so $w=\child(v,a)$. After this first step, \cref{alg:log-walk} follows the random direction $\rho(u)$ at every subsequent consistent internal vertex $u$ on $\pi_v$. This is true whether the first direction is $a=\rho(v)$ or the algorithm first discovers a cycle through $v$ and then restarts with the direction $a\ne\rho(v)$.

    If $w$ is a leaf, then running \cref{alg:log-walk} on $w$ immediately returns $(w)$, which is the desired suffix. Otherwise, $w$ is a consistent internal vertex. Running \cref{alg:log-walk} from $w$ starts by following $\rho(w)$, and then follows the same random directions as the run from $v$ did after reaching $w$ due to shared randomness across queries. Thus, the two executions trace the same directed walk unless the execution from $w$ returns to $w$ and triggers the branch in line \ref{line:alg-log-walk-in-cycle}. This cannot happen along the suffix of $\pi_v$: if some later vertex on the suffix pointed back to $w$, then $w$ would have two distinct in-neighbors in $G'$, namely $v$ and that later vertex, contradicting \cref{lem:log-pseudoforest}. Hence, the execution from $w$ will never enter the branch in line \ref{line:alg-log-walk-in-cycle}, and it stops exactly at the same terminal leaf or inconsistent vertex as the execution from $v$. Therefore its output is $\pi_w=(w,p_1,p_2,\dots,p_\ell)$.
\end{proof}

\begin{lemma}\label{lem:log-algo-log-solve-correctness}
    $\SolvePTL^k$ generates labels satisfying all constraints in $\PathToLeafK$.
\end{lemma}

\begin{proof}
    We distinguish the different types of vertices to prove the statement. Recall the definition of $b'$ and $b''$ in \cref{eq:log-b'-def,eq:log-b''-def}.
    \begin{itemize}
        \item For inconsistent vertices, the algorithm assigns $(\perp, 0)$ to it, satisfying the constraint.
        \item For a leaf $v$, $d(v)$ is assigned $\perp$, $b''(v) = 0$ and $\pi_v = (v)$. \cref{alg:log-solve} collects $b\insub(v)$ and also $b\outsub(\NC(v))$ when $\NC(v)$ exists, so $b\outsub(v) = b\insub(v) \oplus b'(v)$, which is \cref{eq:log-lcl-def} with $b''(v) = 0$.
        \item For an internal vertex $v$, define $w = \child(v,d(v))$. When $w$ is inconsistent, we know that $\pi_v= (v,w)$, and $(d(w), b\outsub(w)) = (\perp, 0)$, meaning that $b''(v) = 0$ and $b\outsub(v) = b\insub(v) \oplus b'(v)$; Otherwise, from \cref{lem:log-algo-log-walk-correctness}, $\pi_w$ is a suffix of $\pi_v$. Compare the processes of computing the output bit between $\SolvePTL^k(v)$ and $\SolvePTL^k(w)$, the only difference is that $\SolvePTL^k(v)$ additionally collects $b\insub(v)$ and $b\outsub(\NC(v))$ if $\NC(v)$ exists. As a result, $b\outsub(v) = b\outsub(w) \oplus b\insub(v) \oplus b'(v)$, which is exactly \cref{eq:log-lcl-def}.
    \end{itemize}
\end{proof}

\begin{proof}[Proof of \cref{lem:log-upper-bound}]
    The Las Vegas algorithm given by \cref{alg:log-solve} can be converted to the required Monte Carlo guarantee by \cref{rem:las-vegas-to-monte-carlo}. The algorithm succeeds with high probability within $O(\log^k n)$ probes according to \cref{cor:log-algo-log-solve-probes,lem:log-algo-log-solve-correctness}.
\end{proof}

\subsection{Lower Bound}

In this part, we prove the lower bound for $\PathToLeafK$ for LCAs via the following lemma.

\begin{lemma}\label{lem:log-hard-root-bit}
    For any positive integer $k$ and for every sufficiently large $n$, there is a set $\mathcal{S}_{k,n}$ of $\PathToLeafK$ instances of size at most $n$ with the following properties:
    \begin{enumerate}[(i)]
        \item All instances in $\mathcal{S}_{k, n}$ have the same underlying graph and binary tree labeling, in which the graph is a tree, every vertex is consistent, and there is a unique root vertex $x$ with $h(x) = k$ and $\P(x) = \perp$;
        \item For any instance $I$ in $\mathcal{S}_{k,n}$, there exists $b_I \in \{0,1\}$ such that any solution satisfying all $\PathToLeafK$ constraints for $I$ must have $b\outsub(x) = b_I$. Define $\mathcal{S}_{k,n}^B$ as the subset of $\mathcal{S}_{k,n}$ with $b_I = B$;
        \item For any set $V_0$ of vertices of size $o(\log^k n)$, there exists a bijection $f_{V_0}$ between $\mathcal{S}_{k,n}^0$ and $\mathcal{S}_{k,n}^1$, where for each $I \in \mathcal{S}_{k,n}$ and $u \in V_0$, $b\insub(u)$ is the same between $I$ and $f_{V_0}(I)$.
    \end{enumerate}
\end{lemma}

We first show how \cref{lem:log-hard-root-bit} implies \cref{lem:log-lower-bound}.

\begin{proof}[Proof of \cref{lem:log-lower-bound}]
    Suppose for contradiction that there is a randomized LCA that solves $\PathToLeafK$ with $o(\log^k n)$ probes and succeeds in every instance over bounded-degree trees of size at most $n$ with probability larger than $1/2$. According to Yao's minimax principle, for any distribution $\mathcal{D}$ over such instances, there is a deterministic LCA $\mathcal{A}$ with the same probe complexity that succeeds with probability larger than $1/2$ over $\mathcal{D}$.

    Now set $\mathcal{D}$ to be the uniform distribution over $\mathcal{S}_{k,n}$. For any instance $I \in \mathcal{S}_{k,n}$, define $V_0$ as the set of vertices probed by $\mathcal{A}(x)$, we have $|V_0| = o(\log^k n)$. According to condition (iii) in \cref{lem:log-hard-root-bit}, there exists $I'$ with $b_I \ne b_I'$, but all inputs in $V_0$ are the same across two instances. As $\mathcal{A}$ is a deterministic algorithm, the transcript of $\mathcal{A}(x)$ on $I$ and $I'$ will be exactly the same, so $\mathcal{A}$ will label $x$ the same in both instances, and one of the two solutions must be wrong according to condition (ii). Since they have the same probability weight in $\mathcal{D}$, $\mathcal{A}$ cannot succeed with probability larger than $1/2$, a contradiction.
\end{proof}

To construct the set of instances, we first provide the construction for the underlying graph and binary tree labeling. The construction consists of $k$ layers of perfect binary trees.

\begin{definition}
    For a positive integer sequence $(a_1,a_2,\dots,a_\ell)$, define the graph $G_{(a_1,a_2,\dots,a_\ell)}$ along with its binary tree labeling recursively in the following way: \begin{itemize}
        \item For the base case, $G_{(1)}$ consists of only one root vertex $x$ with $h(x) = 1$.
        \item $G_{(a_1,a_2,\dots,a_\ell)}$ consists of a vertex $x$ with $h(x) = \ell$ and the following ports for $x$: \begin{itemize}
            \item If $\ell \ne 1$, attach the graph and labeling $G_{(a_2,\dots,a_\ell)}$ to $\NC(x)$, otherwise $\NC(x) = \perp$;
            \item If $a_1 \ne 1$, attach two independent copies of $G_{(a_1-1,a_2,\dots,a_\ell)}$ to $\LC(x)$ and $\RC(x)$ respectively, otherwise $\LC(x) = \RC(x) = \perp$.
        \end{itemize}
    \end{itemize}
    For ``attaching a graph to a port'', we mean creating a copy of the graph and making the port point to the root of the graph. Finally, we set up $\P$ ports to make every vertex consistent. It is easy to check that there is a unique way to do so.
\end{definition}

The following observation can be proved by a standard induction argument.

\begin{observation}\label{obs:log-lower-bound-construction-property}
    $G_{(a_1,\dots,a_\ell)}$ satisfies the following property: \begin{itemize}
        \item $G_{(a_1,\dots,a_\ell)}$ is a tree with $\Theta(\prod_{i=1}^\ell 2^{a_i})$ vertices and maximum degree 4;
        \item With respect to the binary tree labeling, all vertices in the graph are consistent;
        \item The maximum height in the binary tree labeling is $\ell$.
    \end{itemize}
\end{observation}

Now we construct sets $\mathcal{S}^0_{(a_1,a_2,\dots,a_\ell)}$ and $\mathcal{S}^1_{(a_1,a_2,\dots,a_\ell)}$ of $\PathToLeaf^\ell$ instances for each sequence $\{a_1,\dots,a_\ell\}$, in which all instances have $G_{\{a_1,\dots,a_\ell\}}$ as the underlying graph and binary tree labeling. The only difference between these instances lies in $b\insub$ inputs for each vertex, and we construct the inputs also recursively, in the following way: \begin{itemize}
    \item For the base case, $\mathcal{S}_{(1)}^B$ contains the only instance with $b\insub(x) = B$ for $B \in \{0,1\}$;
    \item For the general case, $\mathcal{S}_{(a_1,a_2,\dots,a_\ell)}^B$ contains all instances that can be generated by the following procedure: \begin{itemize}
        \item If $\ell \ne 1$, $\NC(x)$ has a copy of $G_{(a_2,\dots,a_\ell)}$ attached. Arbitrarily pick $b' \in \{0,1\}$ and a $b\insub$ assignment in $\mathcal{S}_{(a_2,\dots,a_\ell)}^{b'}$ to assign $b\insub$ values for the subtree rooted at $\NC(x)$. If $\ell = 1$, set $b'$ to be zero.
        \item If $a_1 \ne 1$, $\LC(x)$ and $\RC(x)$ both have a copy of $G_{(a_1-1,a_2,\dots,a_\ell)}$ attached. Arbitrarily pick $b'' \in \{0,1\}$ and two $b\insub$ assignments in $\mathcal{S}_{(a_1-1,a_2,\dots,a_\ell)}^{b''}$ to assign $b\insub$ values for the subtrees rooted at $\LC(x)$ and $\RC(x)$, respectively. Notice that two copies use the same bit $b''$ but possibly different $b\insub$ assignments. If $a_1 = 1$, set $b''$ to be zero.
        \item Finally, assign $b\insub(x)$ with $B \oplus b' \oplus b''$.
    \end{itemize}
\end{itemize}

Finally, $\mathcal{S}_{\{a_1,\dots,a_\ell\}}$ is the union of $\mathcal{S}^0_{(a_1,a_2,\dots,a_\ell)}$ and $\mathcal{S}^1_{(a_1,a_2,\dots,a_\ell)}$.

\begin{lemma}\label{lem:log-hard-root-bit-by-construction}
    For any positive integer sequence $(a_1,a_2,\dots,a_\ell)$, the sets $\mathcal{S}^0_{(a_1,a_2,\dots,a_\ell)}$ and $\mathcal{S}^1_{(a_1,a_2,\dots,a_\ell)}$ satisfy the following:
    \begin{enumerate}[(i)]
        \item All instances in $\mathcal{S}_{(a_1,a_2,\dots,a_\ell)}$ have the same underlying graph and binary tree labeling, in which the graph is a tree, every vertex is consistent, and there is a unique root vertex $x$ with $h(x) = \ell$ and $\P(x) = \perp$;
        \item For $B\in \{0,1\}$ and any instance $I$ in $\mathcal{S}_{(a_1,a_2,\dots,a_\ell)}^B$, any solution satisfying all $\PathToLeaf^\ell$ constraints for $I$ satisfies $b\outsub(x) = B$;
        \item For any set $V_0$ of vertices of size less than $\prod_{i=1}^\ell a_i$, there exists a bijection $f_{V_0}$ between $\mathcal{S}_{(a_1,a_2,\dots,a_\ell)}^0$ and $\mathcal{S}_{(a_1,a_2,\dots,a_\ell)}^1$, where for each $I \in \mathcal{S}_{(a_1,a_2,\dots,a_\ell)}$ and $u \in V_0$, $b\insub(u)$ is the same between $I$ and $f_{V_0}(I)$.
    \end{enumerate}
\end{lemma}

\begin{proof}
    We prove the lemma by induction on $\sum_{i=1}^\ell a_i$. In the base case, the sequence is $(1)$, and condition (i) is trivial. The graph consists only of the root $x$, and $\mathcal{S}_{(1)}^B$ fixes $b\insub(x)=B$. Hence every legal solution has $b\outsub(x)=B$ by \cref{eq:log-lcl-def}, proving condition (ii). Finally, in condition (iii), the set $V_0$ is required to have size less than one, so $V_0=\emptyset$, and the bijection between $\mathcal{S}_{(1)}^0$ and $\mathcal{S}_{(1)}^1$ is trivial.

    Now consider a sequence $A=\{a_1,\dots,a_\ell\}$ with $\sum_{i=1}^\ell a_i>1$, and let $x$ be the root of $G_A$. Condition (i) follows from the construction of $G_A$ and \cref{obs:log-lower-bound-construction-property}.

    To prove condition (ii), fix $B\in\{0,1\}$ and an instance $I\in\mathcal{S}_A^B$. Let $b'$ and $b''$ be the auxiliary bits used to assign $b\insub$ inputs for $I$. If $\NC(x)\ne\perp$, the $b\insub$ input is assigned according to some instance in $\mathcal{S}_{(a_2,\dots,a_\ell)}^{b'}$, so by the induction hypothesis applied to the subtree rooted at $\NC(x)$, every solution of $I$ has $b\outsub(\NC(x))=b'$. This means that the $b'(x)$ value defined in \cref{eq:log-b'-def} always equals the auxiliary bit $b'$. When $\ell = 1$, we also have $b'(x) = 0 = b'$.
    
    Similarly, if $x$ is internal and $w=\child(x,d(x))$, then $w$ is the root of either the left or right copy of $G_{(a_1-1,a_2,\dots,a_\ell)}$, and both copies have $b\insub$ input generated from some instance in $\mathcal{S}_{(a_1-1,a_2,\dots,a_\ell)}^{b''}$. By induction, every solution of $I$ has $b\outsub(w)=b''$, so the $b''(x)$ value in \cref{eq:log-b''-def} always equals to $b''$, too. When $a_1 = 1$, we also have $b''(x) = 0 = b''$.

    In sum, we have $b'(x) = b'$ and $b''(x) = b''$. Since $b\insub(x) = B \oplus b' \oplus b''$, \cref{eq:log-lcl-def} gives $b\outsub(x) = B$, implying (ii).

    For condition (iii), fix a set $V_0$ of size less than $\prod_{i=1}^\ell a_i$. If $x\notin V_0$, the bijection $f_{V_0}$ can simply flip $b\insub(x)$, so we suppose $x\in V_0$. Let $V_{\NC}$ be the intersection of $V_0$ with the subtree rooted at $\NC(x)$ when it exists, and $V_{\NC} = \varnothing$ if $\NC(x) = \perp$. Similarly define $V_{\LC}$ and $V_{\RC}$. 
    
    If $\ell>1$ and $|V_{\NC}|<\prod_{i=2}^\ell a_i$, applying the induction hypothesis to the subtree rooted at $\NC(x)$, there is a bijection $f_{\NC}$ between $\mathcal{S}_{\{a_2,a_3,\dots,a_\ell\}}^0$ and $\mathcal{S}_{\{a_2,a_3,\dots,a_\ell\}}^1$ that preserves the input for vertices in $V_{\NC}$. Now we define the bijection $f_{V_0}$ in the following way: For each instance $I$, we apply the bijection $f_{\NC}$ to the subtree rooted at $\NC(x)$, and keep other inputs unchanged. From the induction hypothesis, $f_{V_0}$ is a map that preserves inputs among $V_0$. Finally, following \cref{eq:log-lcl-def}, applying $f_{V_0}$ changes $b'(x)$, while $b''(x)$ and $b\insub(x)$ are unchanged, so $b\outsub(x)$ is flipped, showing that $f_{V_0}$ is a bijection between $\mathcal{S}_A^0$ and $\mathcal{S}_A^1$.

    Consider the remaining case, where either $\ell=1$, or $|V_\NC| \ge \prod_{i=2}^\ell a_i$. Since $x\in V_0$ and $|V_0|<a_1\prod_{i=2}^\ell a_i$, we have 
    \[
        |V_{\LC} \cup V_{\RC}| = |V_0 \backslash (\{x\} \cup V_{\NC})| = |V_0| - 1 - |V_{\NC}|  < \prod_{i=1}^\ell a_i - 1 - \prod_{i=2}^\ell a_i  < (a_1 - 1) \prod_{i=2}^\ell a_i.
    \]
    Apply the induction hypothesis to the subtrees rooted at $\LC(x)$ and $\RC(x)$, we know there are bijections $f_\LC, f_\RC$ between $\mathcal{S}_{(a_1-1,a_2,\dots,a_\ell)}^0$ and $\mathcal{S}_{(a_1-1,a_2,\dots,a_\ell)}^1$ which preserve the input among $V_{\LC}$ and $V_{\RC}$, respectively. For the bijection $f_{V_0}$ between $\mathcal{S}_A^0$ and $\mathcal{S}_A^1$, we apply $f_{\LC}$ to the subtree rooted at $\LC(x)$, apply $f_{\RC}$ to the subtree rooted at $\RC(x)$, and keep other inputs unchanged. From the induction hypothesis, $f_{V_0}$ preserves inputs among $V_0$. Finally, from \cref{eq:log-lcl-def}, applying $f_{V_0}$ changes $b''(x)$ regardless of the choice of $d(x)$, while $b'(x)$ and $b\insub(x)$ remain unchanged, so $b\outsub(x)$ is flipped, showing that $f_{V_0}$ is a bijection between $\mathcal{S}_A^0$ and $\mathcal{S}_A^1$.
\end{proof}

\begin{proof}[Proof of \cref{lem:log-hard-root-bit}]
    Recall that $k$ is a constant. Let $q = \Theta(\log n)$ and set $\mathcal{S}_{k,n}$ to be $\mathcal{S}_{\{q,q,\dots,q\}}$ where the sequence has length $k$. According to \cref{obs:log-lower-bound-construction-property}, the graph has size $\Theta(2^{qk}) = n^{\Theta(1)}$, and the size is less than $n$ when the constant on $q$ is small enough. We have $q^k=\Theta(\log^k n)$ and all three conditions in \cref{lem:log-hard-root-bit} follow from \cref{lem:log-hard-root-bit-by-construction}.
\end{proof}

\DeclareRobustCommand{\SolveHC}{\ifmmode\mathop{\text{\normalfont\textsc{SolveHC}}}^T\nolimits\else\textnormal{\textsc{SolveHC}}^T\fi}
\DeclareRobustCommand{\Prepare}{\ifmmode\mathop{\text{\normalfont\textsc{Prepare}}}^T\nolimits\else\textnormal{\textsc{Prepare}}^T\fi}
\DeclareRobustCommand{\Collect}{\ifmmode\mathop{\text{\normalfont\textsc{Collect}}}^T\nolimits\else\textnormal{\textsc{Collect}}^T\fi}
\DeclareRobustCommand{\Traverse}{\ifmmode\mathop{\text{\normalfont\textsc{Traverse}}}^T\nolimits\else\textnormal{\textsc{Traverse}}^T\fi}
\DeclareRobustCommand{\Try}
{\ifmmode\mathop{\text{\normalfont\textsc{Try}}}^T\nolimits\else\textnormal{\textsc{Try}}^T\fi}

\section{Construction in Polynomial Complexity Regime}
\label{sec:poly}

In this section, we provide LCL constructions whose randomized $\Vol$ and LCA complexity is $\tilde \Theta(n^x)$ for any $x \in \mathbb{Q} \cap (0,1]$.

\subsection{Preliminaries}

In this section, our construction is parameterized by an \emph{ordered binary tree} $T$. For an \emph{ordered} binary tree, we mean that the tree is rooted, and we distinguish the left and right children of a vertex. In the following, we provide basic notations around ordered binary trees.

Define the root of an ordered binary tree $T$ as $r_T$ and the subtree of a vertex $v \in T$ as $T_v$. For each vertex $v$, the left child of a vertex $v \in T$ is denoted as $\LC_T(v)$, while the right child is denoted as $\RC_T(v)$. When the vertex does not have a left child, denote $\LC_T(v) = \perp$, and similarly for the right child. The parent of the vertex $v$ is denoted as $\P_T(v)$, while $\P_T(r_T)$ is defined as $\perp$. For notation simplicity, we define $\LC_T(\perp) = \RC_T(\perp) = \P_T(\perp) = \perp$. 

\begin{remark}\label{rem:poly-never-omit-subscript}
    To avoid confusion between $\LC_T, \RC_T$ and $\P_T$ notation in an ordered binary tree and $\LC, \RC, \P$ ports in the LCL (See \cref{def:log-binary-tree-labeling} and \cref{def:poly-t-labeling}), we will \emph{never} omit subscript $T$ on any notation related to ordered binary trees, such as $\LC_T, \RC_T, \P_T$ and $r_T$.
\end{remark}

Define $T_L$ and $T_R$ as the left and right subtrees of $r_T$ in $T$. Denote the unique ordered binary tree with no vertex as $\perp$, and the unique ordered binary tree with only one vertex as $\bullet$. We define $r_\perp = \perp$. 

Now we introduce an important definition in this section, the \emph{value} of an ordered binary tree.

\begin{definition}[Value]\label{def:poly-value-binary-tree}
The \emph{value} of an ordered binary tree $T$ is defined recursively as follow: \begin{equation}
    \val(T) = \begin{cases}
        0, & T = \perp \\
        1, & T = \bullet, \\
        \frac{\val(T_R)}{1 + \val(T_R) - \val(T_L)}, & \text{Otherwise}.
    \end{cases} \label{eq:value-recurrence}
\end{equation}
\end{definition}

We say an ordered binary tree $T$ is \emph{good}, if for every vertex $v \in T$ that is not a leaf, its left subtree has a value strictly smaller than that of its right subtree. In the following, we will only focus on good ordered binary trees. The following observation can be easily proved via induction.

\begin{observation}\label{obs:val-rational-in-0-1}
    $\val(T)$ is always a rational value in $(0,1]$ for any nonempty good ordered binary tree $T$.
\end{observation}

The main theorem we will prove in this section is the following. This theorem, along with \cref{lem:poly-rational-to-ordered-tree}, concludes \cref{thm:poly-main}. 

\begin{theorem}\label{thm:poly-HTC-complexity}
    For every good ordered binary tree $T$ that is not $\perp$, there exists an LCL whose $\Vol$ and LCA complexity is $\tilde\Theta \left(n^{\val(T)}\right)$.
\end{theorem}

\begin{lemma}\label{lem:poly-rational-to-ordered-tree}
For any rational $x \in \mathbb{Q}\cap(0,1]$, there exists a good ordered binary tree $T$ with $\val(T) = x$.
\end{lemma}

\begin{proof}
We induct on $c$ to prove the following claim: For all $0 \le a \le b \le c$, there is a good ordered binary tree $T$ with value $a/b$. If $a = 0$, then we can set $T = \perp$, and we can set $T = \bullet$ when $a = b$, so the base case $c=1$ is proved. 

Now we assume the induction hypothesis holds for $c-1$ and consider the statement for $c$. We only need to consider cases where $b=c$ and $a \in [1,b-1]$. Let $p=(a-1)/(b-1)$ and $q=a/(b-1)$. We have $0\le p < q\le 1$, and both $p$ and $q$ have denominators strictly smaller than $c$. By induction hypothesis, there are two good ordered binary trees $T_L$, $T_R$ whose values are $p$ and $q$ respectively. Now we consider the tree $T$ where the left subtree of the root is $T_L$ while the right subtree is $T_R$. According to \cref{eq:value-recurrence},
\[
  \val(T) = \frac{q}{1+q-p}
  =
  \frac{\frac{a}{b-1}}{1+\frac{a}{b-1}-\frac{a-1}{b-1}}
  =
  \frac{\frac{a}{b-1}}{\frac{b}{b-1}}
  =
  \frac ab.
\]
It is easy to check that $T$ is good, and this concludes the induction.
\end{proof}

\begin{proof}[Proof of \cref{thm:poly-main}]
For every rational $x\in(0,1]$, choose the good ordered binary tree $T$ with $\val(T)=x$, whose existence is guaranteed by \cref{lem:poly-rational-to-ordered-tree}, and apply \cref{thm:poly-HTC-complexity} to generate an LCL with complexity in the LCA and $\Vol$ model $\tilde \Theta(n^{\val(T)}) = \tilde \Theta(n^{x})$.
\end{proof}

\subsection{Description of the LCL}

In the rest of this section, we fix a good ordered binary tree $T$. $T$ does not grow with the size of the LCL instance we will construct, so the size of $T$ is considered a constant.

Similar to \cref{def:log-binary-tree-labeling,def:log-tree-labeling-consistency} in the polylogarithmic setting, we will first provide a set of input labels called $T$-labeling and the associated local constraints to rule out irregularities in the instance.

\begin{definition}[$T$-labeling]\label{def:poly-t-labeling}
Let $G$ be a graph of maximum degree at most $\Delta$, and $\mathcal{P}=[\Delta]\cup\{\perp\}$. For a given ordered binary tree $T$, a \emph{$T$-labeling} consists of a type labeling $t\colon V(G)\to V(T) \cup \{\perp\}$ and the following four labels for each $v\in V(G)$: a parent $\P(v)\in\mathcal{P}$, a left child $\LC(v)\in\mathcal{P}$, a right child $\RC(v)\in\mathcal{P}$, and a direct child $\DC(v)\in\mathcal{P}$. We call $t(v)$ the \emph{type} of $v$. 
\end{definition}

Same as \cref{rem:log-child-port-notation}, we will abuse the notation and use $\P(v), \LC(v), \RC(v)$ and $\DC(v)$ to also denote the vertex reached from the corresponding port.

\begin{definition}\label{def:poly-t-labeling-consistency}
For a $T$-labeling, a vertex $v$ is \emph{consistent} if all the following constraints hold; otherwise, it is \emph{inconsistent}:
\begin{enumerate}
    \item The non-$\perp$ ports among $\P(v),\LC(v),\RC(v)$, and $\DC(v)$ are pairwise distinct.
    \item For $w\in\{\LC(v),\RC(v),\DC(v)\} \backslash \{\perp\}$, $\P(w) = v$.
    \item If $\P(v)\ne \perp$, then $v\in\{\LC(\P(v)),\RC(\P(v)),\DC(\P(v))\}$.
    \item If $\LC(v) \ne \perp$, $t(\LC(v)) = \LC_T(t(v))$. Specifically, $t(\LC(v)) = \perp$ when $\LC_T(t(v)) = \perp$.
    \item If $\RC(v) \ne \perp$, $t(\RC(v)) = \RC_T(t(v))$. Specifically, $t(\RC(v)) = \perp$ when $\RC_T(t(v)) = \perp$.
    \item If $\DC(v) \ne \perp$, $t(\DC(v)) = t(v)$.
\end{enumerate}
\end{definition}

With \cref{def:poly-t-labeling,def:poly-t-labeling-consistency}, we are ready to define our LCL $\HC^T$, where $\mathsf{HC}$ is an abbreviation of $\mathsf{HierarchicalColoring}$. The name comes from the LCL $\mathsf{Hierarchical\ 2\frac{1}{2}\ Coloring}$ introduced in previous works \cite{chang2019time, rosenbaum2020seeing}.

\begin{definition}[$\HC^T$]
    The problem $\HC^T$ is defined in the following way:
    
\textbf{Input:} A graph $G$, a $T$-labeling and a bit $b\insub(v)\in\{0,1\}$ for each $v\in V(G)$.

\textbf{Output:} A color $c\outsub(v)\in\{\D,\X,\U,\Cyc,\perp\}$ and $b\outsub(v)\in\{0,1\}$ for each $v\in V(G)$.

\textbf{Constraint on $c\outsub(v)$:} If $v$ is inconsistent, then $c\outsub(v)=\perp$. Otherwise, $c\outsub(v) \ne \perp$ and the coloring should follow the following constraints:
\begin{enumerate}
    \item If $c\outsub(v)=\Cyc$, then $\DC(v) \ne \perp$ and $c\outsub(\DC(v)) = \Cyc$.
    \item If $c\outsub(v) = \U$, then $\DC(v)=\perp$ or $c\outsub(\DC(v)) \in \{\X, \U,\perp\}$.
    \item If $c\outsub(v)=\X$, then \emph{all} of the following hold:\begin{enumerate}
        \item $t(v) = \perp$ or $t(v)$ is not a leaf in $T$.
        \item $\RC(v) = \perp$ or $c\outsub(\RC(v))\ne\D$.
    \end{enumerate} 
    \item If $c\outsub(v)=\D$, then $t(v) \ne r_T$;
    \item If $c\outsub(v) \in \{\X, \U\}$, then $\LC(v) = \perp$ or $c\outsub(\LC(v)) \ne \D$.
\end{enumerate}

\textbf{Constraint on $b\outsub(v)$:} If $c\outsub(v)\in\{\Cyc,\D,\perp\}$ or $t(v) = \perp$, then $b\outsub(v)=0$. Otherwise, we have the following equation:
\begin{equation}
  b\outsub(v)=b\insub(v)\xor b'(v)\xor b''(v),\label{eq:poly-lcl-def}
\end{equation}
where
\begin{gather}
    b'(v) =
    \begin{cases}
        0, & c\outsub(v) = \U \text{ and }\DC(v)=\perp,\\
        b\outsub(\DC(v)), & c\outsub(v) = \U \text{ and }\DC(v)\ne\perp,\\
        0, & c\outsub(v)=\X\text{ and }\RC(v)=\perp,\\
        b\outsub(\RC(v)), & c\outsub(v)=\X\text{ and }\RC(v)\ne\perp,
    \end{cases}\label{eq:poly-b'-def}\\
    b''(v) =
    \begin{cases}
        0, & \LC(v)=\perp,\\
        b\outsub(\LC(v)), & \LC(v)\ne\perp.
    \end{cases}\label{eq:poly-b''-def}
\end{gather}
\end{definition}

\begin{remark}\label{rem:hierarchical-coloring}
    For readers familiar with previous works \cite{chang2019time, rosenbaum2020seeing}, the $\HTHC$ LCL introduced in these papers is a special case of $\HC^T$, where all vertices in $T$ only have a right child, a new output color $\Cyc$ is introduced, and colors $\R$ and $\B$ are identified to a single color $\U$. Left and right children in $\HTHC$ correspond to $\DC$ and $\RC$ ports in $\HC^T$. The extra color $\Cyc$ is to avoid the possibility that none of the vertices in a cycle want to output $\X$, and there may be no consistent solution for $b\outsub$ along the cycle. 
\end{remark}

\begin{observation}\label{obs:poly-lcl}
    $\HC^T$ is an LCL.
\end{observation}

\begin{proof}
As $|T|$ is a constant, the set of possible input and output labels has a constant size. Now we show that there is a $\LOCAL$ algorithm of radius one to check whether the local constraint is satisfied for each vertex. The consistency conditions in \cref{def:poly-t-labeling-consistency} inspect only $v$, its incident ports, and the labels of its neighbors. As a result, we can determine whether a vertex is consistent within one round. After that, constraints for the output label only involve the one-hop neighborhood. Hence, every constraint is checkable in radius one.
\end{proof}

In the rest of this section, we will show the following two lemmas, jointly proving \cref{thm:poly-HTC-complexity}.

\begin{lemma}\label{lem:poly-upper-bound}
    There is a randomized $\Vol$ algorithm that solves $\HC^T$ over bounded-degree graphs using $\tilde O(n^{\val(T)})$ probes.
\end{lemma}

\begin{lemma}\label{lem:poly-lower-bound}
    Any randomized LCA that solves $\HC^T$ uses $\Omega(n^{\val(T)})$ probes even when we restrict the instance to bounded-degree trees.
\end{lemma}

\begin{proof}[Proof of \cref{thm:poly-HTC-complexity}]
    From \cref{obs:vol-weaker-than-lca}, the upper bound result \cref{lem:poly-upper-bound} can be generalized from $\Vol$ to LCA, while the lower bound result \cref{lem:poly-lower-bound} can be generalized from LCA to $\Vol$. As a consequence, in both models, we have a $\tilde {O} (n^{\val(T)})$ randomized algorithm for $\HC^T$ over bounded-degree graphs and a matching lower bound for bounded-degree trees up to a polylogarithmic factor.
\end{proof}

\subsection{Upper Bound}

Throughout the rest of this section, fix a good ordered binary tree $T$. Set $p = \val(T_L)$ and $q = \val(T_R)$, and furthermore
\[
    \alpha=\val(T)=\frac{q}{1+q-p},
    \qquad
    \beta=\frac{q-p}{1+q-p}.
\]
We have the following equalities:
\begin{gather}
    \beta+p(1-\beta)=q(1-\beta)=\alpha \label{eq:poly-exponent-identities} \\
    (p-1)(1-\beta) = \frac{p-1}{1+q-p} = \frac{q-(1+q-p)}{1+q-p} = \alpha - 1. \label{eq:poly-upper-bound-(p-1)(1-beta)}
\end{gather}

As we see in \cref{sec:tech-overview}, $\HC^T$ instances are formed by paths (corresponding to $\DC$ paths and cycles in which all vertices have type $r_T$) and some attached $\HC^{T_L}$ and $\HC^{T_R}$ instances. We formalize this structure by the definition of a \emph{sub-instance}.

\begin{definition}[Sub-instance]
    For a $\HC^T$ instance $I$, define a graph $G_I = (V_I, E_I)$ where $V_I$ contains all consistent vertices with type other than $r_T$ and $\perp$ in $I$, and $E_I = \{(u,\P(u)) \mid u, \P(u) \in V_I\}$. A connected component of $G_I$ induces an instance, which we call a \emph{sub-instance} in $I$. A sub-instance is a $\HC^{T_L} (\HC^{T_R})$ sub-instance if all vertices in the sub-instance have type in $T_L (T_R)$.
    
    For each consistent vertex $u$ of type $r_T$, $\LC(u)$ and $\RC(u)$ will be in a $\HC^{T_L}$ and a $\HC^{T_R}$ sub-instance when these two vertices are consistent. We will call the two sub-instances the $\HC^{T_L}(\HC^{T_R})$ instance \emph{attached to $u$}.
\end{definition}
\begin{lemma}\label{lem:poly-upper-bound-sub-instance-categorization}
    Every vertex $v$ in a $\HC^T$ instance $I$ falls in exactly one of the categories: an inconsistent vertex, a consistent vertex of type $\perp$, a consistent vertex of type $r_T$, a vertex in some $\HC^{T_L}$ sub-instance, or a vertex in some $\HC^{T_R}$ sub-instance.
\end{lemma}
\begin{proof}
    If $v$ is inconsistent, or if $v$ is consistent with $t(v)=\perp$ or $t(v)=r_T$, then $v$ falls into exactly one of the first three categories. It remains to consider a consistent vertex $v$ with $t(v)\notin\{\perp,r_T\}$. Since $T$ is an ordered binary tree, $t(v)$ belongs to exactly one of $T_L$ and $T_R$. The vertex $v$ is therefore included in $V_I$ and lies in a unique connected component of $G_I$.

    We now show that this component cannot contain types from both $T_L$ and $T_R$. Consider an edge $\{x,y\}$ of $G_I$. By the definition of $G_I$, $x$ and $y$ are consistent vertices, so up to swapping $x$ and $y$, we have $y=\P(x)$ and $x\in\{\LC(y),\RC(y),\DC(y)\}$. If $x=\DC(y)$, then $t(x)=t(y)$. If $x=\LC(y)$ or $x=\RC(y)$, then $t(x)$ is a child of $t(y)$ in $T$. In all cases, as $x,y\in V_I$, the two types lie in the same one of the two subtrees $T_L$ and $T_R$. Hence every connected component of $G_I$ is contained entirely in one side, and $v$ belongs to exactly one $\HC^{T_L}$ or $\HC^{T_R}$ sub-instance.
\end{proof}

In the rest of this section, a randomized $\Vol$ algorithm $\mathcal{A}$ for $\HC^T$ will receive three inputs $(u, N_1, N_2)$, where $u$ is the queried vertex, and $N_1, N_2$ are positive size parameters with $N_1\le N_2$. The algorithm $\mathcal{A}$ will be used to solve both an $\HC^T$ instance, in which $N_1 = N_2 = n$, and a $\HC^T$ sub-instance inside a larger $\HC^{T'}$ instance, in which $N_1$ and $N_2$ give clues about the size of the current sub-instance.
\begin{itemize}
    \item $N_1$ denotes the \emph{ideal size} of the $\HC^T$ sub-instance, and this parameter is to ensure the correct probe complexity and consistency. Most decisions in the algorithm that are related to the size of the instance, such as the probability of each vertex entering the sample set in \cref{alg:poly-traverse}, will use $N_1$ as the size parameter.
    \item $N_2$ is the \emph{maximum possible size} of the $\HC^T$ sub-instance. When the algorithm has a high confidence that the current sub-instance is even larger than $N_2$, then it could use $(\D, 0)$ to label the current vertex being queried. Otherwise, the algorithm should make the best effort to use a label other than $(\D, 0)$.
\end{itemize}
Within the algorithm, we will utilize binary lifting to estimate the correct size of the sub-instance, with the estimation denoted as $N$. The actual size of the sub-instance is generally unknown to the algorithm, and we denote it as $n^\star$. $n$ denotes the size of the whole instance, which is known to the algorithm. One may think of the algorithm as solving a $\HC^T$ sub-instance of unknown size $n^\star$ inside a large instance of known size $n$. Additionally, the algorithm $\mathcal{A}$ may execute the algorithm for $\HC^{T_L}$ and $\HC^{T_R}$ to deal with the $\HC^{T_L}$ and $\HC^{T_R}$ sub-instances of the current instance $I$.

The following definition provides the guarantee for the algorithm in this section.

\begin{definition}[Strong Algorithm]\label{def:poly-strong-algo}
    Fix a $\HC^T$ instance $I$ with $n^\star$ vertices, a randomized $\Vol$ algorithm $\mathcal{A}$, a randomness setup $\mathcal{R} = \{r_v\}_{v \in I}$, and $n \ge n^\star$. For a vertex $u\in I$ and positive integers $N_1\le N_2$, define $a_{u,N_1,N_2}$ as the output of $\mathcal{A}(u,N_1,N_2)$ using $\mathcal{R}$ as the randomness, and define $b_{u,N_1,N_2}$ as the number of probes used by $\mathcal{A}(u,N_1,N_2)$. $\mathcal{A}$ knows the value of $n$ but does not know $n^\star$.
    
    The algorithm $\mathcal{A}$ is called \emph{strong} if for all $\HC^T$ instances $I$ of size $n^\star$, $n \ge n^\star$ and $1 \le N_1 \le n$, with probability $1-n^\star/n^2$ over $\mathcal{R}$ all the following conditions hold:
    \begin{enumerate}[(i)]
        \item (Stability over $N_2$) For every vertex $u\in I$, there exists $N_2^\star(u,N_1)\ge N_1$ such that 
        \begin{equation}
            a_{u,N_1,N_2} = \begin{cases}
            (\D,0), & N_2<N_2^\star(u,N_1), \\
            a_{u,N_1,N_2^\star(u,N_1)}, & N_2\ge N_2^\star(u,N_1).
            \end{cases}
        \end{equation}
        \item (Correctness for small $N_2$) For all $N_2\ge N_1$, all constraints in $\HC^T$ are satisfied by the output labeling given by $a_{v,N_1,N_2}$, except constraints of type 4 for $c\outsub$ (vertices of label $(\D, 0)$ should not have type $r_T$).
        \item (Correctness for large $N_2$) For every $N_2\ge \max(N_1,n^\star)$ and every consistent vertex $u\in I$ with $t(u)=r_T$, $a_{u,N_1,N_2}\ne(\D,0)$.
        \item (Probe complexity) For every $N_2\ge N_1$,
        \begin{equation}
            \max_{u \in I} b_{u,N_1,N_2} = \begin{cases}
                O(1), & T = \perp, \\ 
                O(N_2N_1^{\alpha-1}\log^{d_T}n), & T \ne \perp.
            \end{cases} \label{eq:poly-strong-algo-complexity-universal}
        \end{equation}
        In addition, when $T \ne \perp$, then for every consistent vertex $u\in I$ with $t(u)=r_T$,
        \begin{equation}
            b_{u,N_1,N_2} = O(\max(N_1,n^\star)N_1^{\alpha-1}\log^{d_T}n). \label{eq:poly-strong-algo-complexity-root}
        \end{equation}
        Here $d_T$ is the depth of the good ordered binary tree $T$ defined as the number of edges in the longest path in $T$ from $r_T$ to some vertex in the tree.
    \end{enumerate}
\end{definition}

\begin{remark}
    We explain here the intuition of the four requirements in \cref{def:poly-strong-algo}.
    \begin{enumerate}[(i)]
        \item Condition (i) says that an algorithm can label a vertex $u$ with output label $(\D, 0)$ when $N_2$ is smaller than the actual size $n^\star$. However, once $N_2$ is large enough for the algorithm to make a decision, increasing $N_2$ will not change the output. Notice that we allow the algorithm to always output $(\D, 0)$ no matter how large $N_2$ is.
        \item Condition (ii) says that even if $N_2$ is too small compared to $n^\star$, the output labels satisfy all constraints in $\HC^T$ except that some vertices with type $r_T$ may output $(\D, 0)$.
        \item Condition (iii) says that once $N_2$ is at least $n^\star$, every root-type vertex is not labeled $(\D, 0)$ with high probability. Combining condition (ii), the output labeling in the sub-instance satisfies all $\HC^T$ constraints with probability $1-n^\star/n^2$.
        \item Condition (iv) contains two complexity guarantees, \cref{eq:poly-strong-algo-complexity-universal,eq:poly-strong-algo-complexity-root}. \cref{eq:poly-strong-algo-complexity-universal} applies to the case where $n^\star \gg N_2$: in this case, we require the algorithm to produce the output of \emph{every} vertex in the instance with a probe complexity depending on $N_2$ rather than $n^\star$. $N_1$ roughly measures the proportion of vertices that the algorithm is going to shave. When $N_1 = N_2 = n^\star$ and $T \ne \perp$, since $d_T$ is a constant, we have that with high probability $b_{u, n^\star, n^\star} = \tilde O\left((n^\star)^{1+\alpha - 1}\right) = \tilde O((n^\star)^\alpha)$, which is the desired complexity for $\HC^T$. \cref{eq:poly-strong-algo-complexity-root} handles the case where $n^\star \ll N_2$. In this case, we require the algorithm to produce a labeling for \emph{root-type vertices} faster than the bound claimed by \cref{eq:poly-strong-algo-complexity-universal}.
    \end{enumerate}
\end{remark}

We will show the following claim in the rest of this section, and it immediately implies \cref{lem:poly-upper-bound}.

\begin{theorem}\label{thm:poly-upper-bound-strong}
    For every good ordered binary tree $T$, there is a strong randomized $\Vol$ algorithm for $\HC^T$.
\end{theorem}

\begin{proof}[Proof of \cref{lem:poly-upper-bound}]
    Apply \cref{thm:poly-upper-bound-strong} with $N_1=N_2=n^\star=n$. By conditions (ii) and (iii) of \cref{def:poly-strong-algo}, the output labeling $\{a_{u,n,n}\}_{u \in I}$ satisfies $\HC^T$ with probability at least $1-1/n$. By condition (iv), the probe complexity is $\tilde O(n^{\val(T)})$ with probability $1 - 1/n$. Finally, we can use \cref{rem:las-vegas-to-monte-carlo} to translate this algorithm from Las Vegas to Monte Carlo.
\end{proof}

We build the algorithm for \cref{thm:poly-upper-bound-strong} inductively, and a strong algorithm for $\HC^T$ will use strong algorithms for $\HC^{T_L}$ and $\HC^{T_R}$. We first handle the base cases $T=\perp$ and $T=\bullet$.

\begin{lemma}\label{lem:poly-upper-bound-base-perp}
    There is a strong deterministic $\Vol$ algorithm for $\HC^\perp$.
\end{lemma}
\begin{proof}
    The algorithm checks whether the queried vertex $u$ is consistent using \cref{def:poly-t-labeling-consistency}. If $u$ is inconsistent, it is labeled $(\perp,0)$; otherwise it is labeled $(\X,0)$. The output is not $(\D, 0)$ and is independent of the parameter $N_2$, so condition (i) of \cref{def:poly-strong-algo} is satisfied.

    Since $T=\perp$, every consistent vertex has type $\perp$. Thus $\X$ is allowed by item 3(a) in the $c\outsub$ constraint, and item 3(b) is always satisfied since output color $\D$ does not appear in the instance. This shows that all vertices satisfy the $\HC^T$ constraints on $c\outsub$. The constraint on $b\outsub$ enforces every vertex $u$ in the instance to have $b\outsub(u)=0$. This means that conditions (ii) and (iii) hold. Finally, the probe complexity is $O(1)$, satisfying condition (iv).
\end{proof}

\begin{lemma}\label{lem:poly-upper-bound-base-bullet}
    There is a strong deterministic $\Vol$ algorithm for $\HC^\bullet$.
\end{lemma}
\begin{proof}
    The algorithm first performs the same consistency and type-$\perp$ checks as in \cref{lem:poly-upper-bound-base-perp}. Thus inconsistent vertices output $(\perp,0)$, type-$\perp$ consistent vertices output $(\X,0)$, and constraints on these vertices are satisfied. Now suppose $u$ is consistent and $t(u)=r_T$, then colors $\X$ and $\D$ are both prohibited at $u$. With the input parameter $N_2$, the algorithm follows the $\DC$ path from $u$ (i.e. the path $u, \DC(u), \DC(\DC(u)) \dots$) for $N_2$ steps and outputs the first applicable label below:
    \begin{itemize}
        \item If the walk revisits $u$, $u$ is labeled $(\Cyc,0)$.
        \item If the walk reaches a consistent vertex $w$ such that $\DC(w)=\perp$ or $\DC(w)$ is inconsistent, $u$ is labeled $\left(\U,\bigoplus_{z\in P_{u,w}} b\insub(z)\right)$, where $P_{u,w}$ is the $\DC$ path from $u$ to $w$.
        \item If neither event above is found within $N_2$ steps, $u$ is labeled $(\D,0)$.
    \end{itemize}
    
    Now we check the conditions for strong algorithms. For a fixed $N_1$, the threshold $N_2^\star(u, N_1)$ is the smallest value $N_2$ for which the walk from $u$ sees one of the two terminating events, hence condition (i) holds.
    
    Now we check conditions (ii) and (iii). Fix $N_2\ge N_1$, and distinguish different possibilities for a vertex $u$ in the instance. Inconsistent vertices and consistent type-$\perp$ vertices have already been handled. Let $u$ be a consistent vertex with $t(u)=r_T$. If the algorithm outputs $(\D,0)$, then the only possible violation is item 4 in the constraint on $c\outsub$, which is precisely the exception allowed in condition (ii).

    Suppose the algorithm outputs $(\Cyc,0)$ at $u$. Then the $\DC$ walk from $u$ returns to $u$ within $N_2$ steps. For every vertex $z$ on this directed cycle, the same walk starting from $z$ also returns to $z$ within $N_2$ steps before encountering any terminating event, so $z$ is also labeled $(\Cyc,0)$. Hence item 1 in the constraint on $c\outsub$ is satisfied at $u$, and the bit constraint is satisfied because vertices of color $\Cyc$ have output bit zero.

    Finally, suppose the algorithm outputs $\left(\U,\bigoplus_{z\in P_{u,w}}b\insub(z)\right)$ at $u$, where $w$ is the first consistent vertex on the $\DC$ path from $u$ such that $\DC(w)=\perp$ or $\DC(w)$ is inconsistent. If $u=w$, then item 2 in the constraint on $c\outsub$ is immediate. Otherwise, the vertex $\DC(u)$ reaches the same terminating vertex $w$ within fewer steps, so it is also labeled $\U$. Thus item 2 is satisfied at $u$. Since $T=\bullet$, vertices of type $r_T$ have left and right children of type $\perp$ whose output bit must be zero. As a result, \cref{eq:poly-lcl-def} reduces to
    \[
        b\outsub(u)=b\insub(u)\oplus b\outsub(\DC(u)),
    \]
    with $b\outsub(\DC(u))=0$ when $\DC(u)=\perp$ or $\DC(u)$ is inconsistent. The bit given by the algorithm is exactly this recurrence expanded along the path from $u$ to $w$.

    To prove condition (iii), let $N_2\ge\max(N_1,n^\star)$ and $u$ be a consistent vertex with $t(u)=r_T$. Along the $\DC$ walk from $u$, the first repeated consistent vertex must be $u$: if the first repeated vertex were $z\ne u$, then $z$ would have two different predecessors, contradicting the consistency definition in \cref{def:poly-t-labeling-consistency}. Therefore, within at most $n^\star$ steps, the walk either returns to $u$ or reaches a vertex whose $\DC$ port is $\perp$ or points to an inconsistent vertex. Since $N_2\ge n^\star$, the algorithm sees one of these terminating events and does not output $(\D,0)$.
    
    Finally, the walk uses $O(\min(n^\star, N_2))$ probes for every vertex. This is $O(N_2)$ for all vertices and $O(\max(N_1,n^\star))$ for root-type vertices. Since $\alpha = 1$ and $d_T = 0$, condition (iv) holds.
\end{proof}

For the induction step, assume that there are strong $\Vol$ algorithms $\mathcal{A}_L$ and $\mathcal{A}_R$ for $\HC^{T_L}$ and $\HC^{T_R}$ respectively. We also assume that $T\ne\bullet$, so $r_T$ is not a leaf and the vertices of type $r_T$ are not prohibited from using color $\X$.

\subsubsection{Algorithm Description}

The description of the algorithm is given in \cref{alg:poly-upper-bound}. Before explaining the algorithm, we first explain its subroutines, \cref{alg:poly-prepare,alg:poly-collect,alg:poly-traverse}. Inside the algorithm, we define $\DC^j(u)$ as the $j$-th vertex in the $\DC$ path starting from $u$. Namely, $\DC^0(u) = u$ and $\DC^j(u) = \DC(\DC^{j-1}(u))$ for $j \ge 1$.

\begin{algorithm}[p]
\caption{$\Prepare(u,N_1,N)$}
\label{alg:poly-prepare}
\begin{algorithmic}[1]
\ENSURE Either the output label $(c\outsub(u),b\outsub(u))$, or $\perp$ indicating no output.
\IF{$u=\perp$}
    \RETURN $(\X,0)$
\ELSIF{$u$ is inconsistent}
    \RETURN $(\perp,0)$
\ELSIF{$t(u)=\perp$}
    \RETURN $(\X,0)$
\ELSIF{$t(u)\in T_L$}
    \RETURN $\mathcal{A}_L(u,\lceil N_1^{1-\beta} \rceil,N)$
\ELSIF{$t(u)\in T_R$}
    \RETURN $\mathcal{A}_R(u,\lceil N_1^{1-\beta} \rceil, \lceil NN_1^{-\beta} \rceil)$
\ELSE
    \RETURN $\perp$
\ENDIF
\end{algorithmic}
\end{algorithm}

\begin{algorithm}[p]
\caption{$\Collect(u,k,N_1,N)$}
\label{alg:poly-collect}
\begin{algorithmic}[1]
\ENSURE $\perp$ when the algorithm fails, or a bit $b\in\{0,1\}$ when it succeeds.
\STATE Set a probe budget of size $\Theta(NN_1^{\alpha-1} \log^{d_T}n)$ with a sufficiently large constant. Once the whole $\Collect$ procedure uses up the budget, stop all recursions and return $\perp$.
\STATE $b\leftarrow 0$
\FOR{$j$ from $0$ to $k$}
    \STATE $(c',b')\leftarrow \Prepare(\LC(\DC^j(u)),N_1, N)$ \label{line:alg-poly-collect-solve-LC}
    \IF{$(c',b') = (\D, 0)$}
        \RETURN $\perp$
    \ENDIF
    \STATE $b\leftarrow b\oplus b\insub(\DC^j(u))\oplus b'$
\ENDFOR
\RETURN $b$
\end{algorithmic}
\end{algorithm}

\begin{algorithm}[p]
\caption{$\Traverse(u,N_1,N)$}
\label{alg:poly-traverse}
\begin{algorithmic}[1]
\ENSURE A sample set $S_u$ and a backup label $(c_0,b_0)$.
\STATE $S\leftarrow\varnothing$, $(c_0,b_0)\leftarrow(\D,0)$
\FOR{$j$ from $0$ to $\lceil 2N_1^\beta\rceil$}
    \STATE $w\leftarrow \DC^j(u)$
    \IF{$w=\perp$ or $w$ is inconsistent or $(j>0\text{ and }w=u)$} \label{line:alg-poly-traverse-check-terminate}
        \IF{$j>0\text{ and }w=u$}
            \STATE $(c_0,b_0)\leftarrow(\Cyc,0)$
        \ELSIF{$\Collect(u,j-1,N_1,N)\ne\perp$} \label{line:alg-traverse-call-collect}
            \STATE $(c_0,b_0)\leftarrow(\U,\Collect(u,j-1,N_1,N))$
        \ENDIF
        \STATE \textbf{break}
    \ENDIF
    \STATE With probability $\min(1,4N_1^{-\beta}\log n)$, include $j$ in $S$. The same randomness is used between different calls when $N_1$ is the same. \label{line:alg-poly-traverse-sample}
\ENDFOR
\RETURN $S,(c_0,b_0)$
\end{algorithmic}
\end{algorithm}

\begin{algorithm}[p]
\caption{$\SolveHC(u,N_1,N_2)$}
\label{alg:poly-upper-bound}
\begin{algorithmic}[1]
\ENSURE An output labeling $(c\outsub(u),b\outsub(u))$ for $u$.
\IF{$N_1 > N_2$}
    \RETURN $(\D,0)$
\ENDIF
\IF{$\Prepare(u,N_1,2N_2)\ne\perp$}
    \RETURN $\Prepare(u,N_1,2N_2)$
\ENDIF
\STATE $k \leftarrow \lceil \log_2(N_2 / N_1) \rceil$
\FOR{$p$ from $0$ to $k$}
    \IF{$\Try(u, N_1, p) \ne \perp$}
        \RETURN $\Try(u, N_1, p)$
    \ENDIF
\ENDFOR
\RETURN $(\D,0)$
\end{algorithmic}
\end{algorithm}

\begin{algorithm}[p]
\caption{$\Try(u,N_1,p)$}
\label{alg:poly-try}
\begin{algorithmic}[1]
\ENSURE An output labeling $(c\outsub(u),b\outsub(u))$ other than $(\D, 0)$, or $\perp$ to denote a failure
\STATE $N \leftarrow 2^p N_1$
\STATE $S_u,(c_0,b_0)\leftarrow \Traverse(u,N_1,N)$
\STATE $S_{u,p}' \leftarrow \{j\in S_u\mid \Prepare(\RC(\DC^j(u)),N_1, N) \ne(\D,0)\}$ \label{line:alg-poly-try-check-X}
\IF{$S_{u,p}' \ne \varnothing$}
    \STATE $j^\star\leftarrow\min S_{u,p}'$, $w \leftarrow \DC^{j^\star}(u)$
    \STATE $(c^\star, b^\star, B) \leftarrow (\D, 0, \perp)$
    \IF{$j^\star = 0$}
        \STATE $c^\star \leftarrow \X$
        \STATE $(-, b^\star) \leftarrow \Prepare(\RC(w),N_1,N)$
        \STATE $B \leftarrow \Collect(u,0,N_1,N)$
    \ELSE
        \STATE $c^\star \leftarrow \U$
        \IF{$\SolveHC(w, N_1, N/2) \ne (\D, 0)$} \label{line:alg-poly-upper-bound-prev-success-block}
            \STATE $(-, b^\star) \leftarrow \SolveHC(w, N_1, N/2)$
            \STATE $B \leftarrow \Collect(u,j^\star-1,N_1,N)$ 
        \ELSE \label{line:alg-poly-upper-bound-prev-fail-block}
            \STATE $(-,b^\star)\leftarrow \Prepare(\RC(w),N_1, N)$
            \STATE $B \leftarrow \Collect(u,j^\star,N_1,N)$ 
        \ENDIF
    \ENDIF
    \IF{$B = \perp$}
        \RETURN $\perp$
    \ELSE
        \RETURN $(c^\star, b^\star \oplus B)$
    \ENDIF
\ELSIF{$(c_0,b_0)\ne(\D,0)$}
    \RETURN $(c_0,b_0)$ \label{line:alg-poly-upper-bound-output-backup}
\ELSE
    \RETURN $\perp$
\ENDIF
\end{algorithmic}
\end{algorithm}

The algorithm $\Prepare$ described in \cref{alg:poly-prepare} is a subroutine that handles edge cases and sub-instances. For inconsistent vertices or vertices of type $\perp$, we label them in the same way as in the algorithm for $\HC^\perp$ described in \cref{lem:poly-upper-bound-base-perp}. For consistent vertices $u$ where $t(u) \ne r_T$, according to \cref{lem:poly-upper-bound-sub-instance-categorization}, it lies in either a $\HC^{T_L}$ or a $\HC^{T_R}$ sub-instance. We use $\mathcal{A}_L$ and $\mathcal{A}_R$ to solve them. Notice that parameter $N_2$ is different between $\mathcal{A}_L$ and $\mathcal{A}_R$ calls: $\HC^{T_R}$ sub-instances will be solved in a best-effort basis, meaning that if the sub-instance size is larger than $\lceil N N_1^{-\beta} \rceil$, it is allowed to label $(\D, 0)$; Instead, the maximum possible size for $\HC^{T_L}$ sub-instances is a conservative value $N$.

The algorithm $\Collect$ described in \cref{alg:poly-collect} is used when the algorithm plans to color the $\DC$ path from $u$ to $\DC^{k-1}(u)$ with $\U$ and color $\DC^k(u)$ with $\U$ or $\X$. When we follow the coloring plan, expanding $b\outsub(u)$ using \cref{eq:poly-lcl-def} along the path, we have that 
\[
    b\outsub(u)=b'(\DC^k(u)) \oplus \left(\bigoplus_{j=0}^k
        \left(b\insub(\DC^j(u))\oplus b''(\DC^j(u))\right)\right),
\]
and $\Collect$ tries to compute the term $\left(\bigoplus_{j=0}^k \left(b\insub(\DC^j(u))\oplus b''(\DC^j(u))\right)\right)$. This requires $\Collect$ to get the output bit for $\LC(\DC^j(u))$ for each $0 \le j \le k$ if it exists. Each of these vertices is either handled directly by $\Prepare$ or lies in a $\HC^{T_L}$ sub-instance, so we call $\Prepare$ with size parameters $(N_1, N)$ to get them.

To avoid wasting too many probes when the total size of the attached $\HC^{T_L}$ sub-instances is much larger than $N$, we set up a probe budget of size $\tilde \Theta(NN_1^{\alpha-1})$. Once those calls use up the probe budget, we immediately stop this process and all recursions generated by it, returning $\perp$ to indicate a failure. Another possibility of failure is when any $\Prepare$ call returns $(\D, 0)$: in this case, we cannot execute the original coloring plan for the $\DC$ path due to item 5 in the constraint on $c\outsub$ (a vertex of color $\X$ or $\U$ cannot have a left child of color $\D$). In this case, $\Collect$ also returns $\perp$.

The algorithm $\Traverse$ described in \cref{alg:poly-traverse} travels along the $\DC$ path from $u$ for at most $\lceil 2N_1^\beta \rceil$ steps. For each of the vertices traveled, we include it in a sample set $S$ with probability $\min(1,4 N_1^{-\beta} \log n)$. This sampling should be consistent across different $\Traverse$ calls with the same $N_1$, meaning that for two calls $\Traverse(u, N_1, N)$ and $\Traverse(v, N_1, N')$ that both visit $w$, $w$ will be in both or neither of the sample sets. Apart from the sample set $S$, like the algorithm in \cref{lem:poly-upper-bound-base-bullet} for $\HC^\bullet$, $\Traverse$ produces a label $(c_0,b_0)$ for $u$ when the walk returns itself or reaches a terminating vertex. 

Finally, $\SolveHC$ described in \cref{alg:poly-upper-bound} first calls $\Prepare$ to deal with edge cases and sub-instances. Vertices left are consistent vertices with type $r_T$ due to \cref{lem:poly-upper-bound-sub-instance-categorization}. For each of them, $\SolveHC$ uses binary lifting to guess $n^\star$, the actual size of the current $\HC^T$ (sub-)instance, from $N_1$ to $2N_2$. For each guess $p \in [0, \lceil \log_2 (N_2 / N_1) \rceil]$, $\SolveHC$ calls $\Try(u, N_1, p)$ described in \cref{alg:poly-try} to see whether the size guess $N = 2^pN_1$ can generate for $u$ an output label other than $(\D, 0)$. If none of them work, the algorithm then has a high confidence that $N_2$ is actually much smaller than $n^\star$. In this case, condition (ii) in the definition of a strong algorithm allows $\SolveHC$ to give output label $(\D, 0)$ for $u$, and it outputs $(\D, 0)$. For the other case, if any of the $\Try$ calls succeeds, it gives the label from the first successful $\Try$ call to $u$.

Each $\Try(u,N_1,p)$ call tries to generate an output label other than $(\D, 0)$ for $u$ with size guess $N = 2^p N_1$. It first gets the sample set $S_u$ and a backup label $(c_0,b_0)$ from $\Traverse$. Notice that the set $S_u$ does not carry a $p$ subscript because, for fixed $u$ and $N_1$, the walk inspected by $\Traverse$ and the sampling choices are independent of the size guess $N=2^pN_1$. Indices in $S_u$ are candidates for an output color $\X$. To color a candidate $\DC^j(u)$ with $\X$, item 3(b) in $c\outsub$ constraints requires the algorithm to give an output label other than $(\D, 0)$ to $\RC(\DC^j(u))$. $\RC(\DC^j(u))$ is either handled directly by $\Prepare$ or lies in a $\HC^{T_R}$ sub-instance. As a result, for each $j \in S_u$, we check whether $\Prepare(\RC(\DC^j(u)),N_1, N)$ returns an output label other than $(\D, 0)$. We collect all such indices to the set $S_{u,p}'$ in line \ref{line:alg-poly-try-check-X}.

$N$ may be too small, or we are simply unlucky, and $S_{u,p}'$ may then be empty. In this case, we enter line \ref{line:alg-poly-upper-bound-output-backup} and try to output the backup label if it is not $(\D, 0)$. If the backup label is also $(\D, 0)$, then we indicate a failure by returning $\perp$.

If $S_{u,p}'$ is not empty, we pick $j^\star$ as the minimum in $S_{u,p}'$, and plan to color $w = \DC^{j^\star}(u)$ with $\X$ and the path between $\P(w)$ and $u$ with $\U$ when $w \ne u$. For the case where $w = u$, we use $\Collect(u,0,N_1,N)$ to compute the contribution of $b\insub(v)$ and $b''(v)$ defined in \cref{eq:poly-b''-def}, denoted by $B$.

When $w \ne u$, to ensure local consistency, we need to make sure that $w$, whose output label is generated by an independent query $\SolveHC(w, N_1, N_2)$, agrees with this plan, thinking that it is colored $\X$. To address this issue, we will prove that if $\SolveHC(w, N_1, N_2)$ does not follow the coloring plan, then $\SolveHC(w, N_1, N_2)$ must get a label from a smaller size guess. This is why in line \ref{line:alg-poly-upper-bound-prev-success-block}, we call $\SolveHC(w, N_1, N/2)$ to see whether this is the case. 

If the call does not return $(\D, 0)$, then we know the output label for $w$; otherwise, we know that $\SolveHC(w, N_1, N_2)$ will follow the same coloring plan and color $w$ with $\X$. In both cases, we will color the path between $\P(w)$ and $u$ with $\U$ and use $\Collect$ to get their contributions to $b\outsub(u)$. When the $\Collect$ call fails, $\Try$ will return $\perp$.

\subsection{Algorithm Analysis}

Throughout the analysis, we assume the $\HC^T$ instance $I$ of size $n^\star$, $N_1$, and $n \ge N_1$ are arbitrary but fixed. Define $n_0$ as the number of consistent vertices of type $r_T$, sequences $(I_{l,1},I_{l,2},\dots,I_{l,n_l})$ and $(I_{r,1},I_{r,2},\dots,I_{r,n_r})$ as the sequences of $\HC^{T_L}$ and $\HC^{T_R}$ sub-instances, and $(s_{l,1},s_{l,2},\dots,s_{l,n_l})$, $(s_{r,1},s_{r,2},\dots,s_{r,n_r})$ as their size sequences. According to \cref{lem:poly-upper-bound-sub-instance-categorization}, $n^\star \ge n_0 + \sum_{i=1}^{n_l} s_{l,i} + \sum_{i=1}^{n_r} s_{r,i}.$

We apply induction hypothesis for $\mathcal{A}_L$ and $\mathcal{A}_R$ to each $\HC^{T_L}$ and $\HC^{T_R}$ sub-instance. This means that,
\begin{itemize}
    \item For each $1 \le i \le n_l$, the output labeling given by $\SolveHC(\star, N_1, N_2)$ for $I_{l,i}$, which is further the output labeling given by $\mathcal{A}_L(\star, \lceil N_1^{1-\beta} \rceil, 2N_2)$, satisfies the four conditions given in \cref{def:poly-strong-algo} with probability $1-s_{l,i}/n^2$;
    \item For each $1 \le i \le n_r$, the output labeling given by $\SolveHC(\star, N_1, N_2)$ for $I_{r,i}$, which is further the output labeling given by $\mathcal{A}_R(\star, \lceil N_1^{1-\beta} \rceil, \lceil 2N_2 N_1^{-\beta} \rceil)$, satisfies the four conditions given in \cref{def:poly-strong-algo} with probability $1-s_{r,i}/n^2$;
\end{itemize}
We additionally require the following event to occur for every consistent vertex of type $r_T$:

\begin{lemma}\label{lem:poly-sampling-success}
    For every consistent vertex $u$ of type $r_T$, the following event happens with probability $1-1/n^2$: For every $N \ge N_1$, if the call $\Traverse(u, N_1, N)$ never enters the if-block at line \ref{line:alg-poly-traverse-check-terminate} (i.e. within $\lceil 2N_1^\beta\rceil$ steps, the walk never returns to itself or visits a vertex $w=\perp$ or an inconsistent vertex $w$), then the call returns a sample set $S$ of size $O(\log n)$ in which there exists $j$ such that the $\HC^{T_R}$ instance attached to $\DC^j(u)$ has size at most $n^\star/N_1^\beta$. For fixed $u$ and $N_1$, this sample set is the same for all values of $N$.
\end{lemma}
\begin{proof}
    The behavior of $\Traverse(u,N_1,N)$ before entering the if-block in line \ref{line:alg-poly-traverse-check-terminate} is independent of $N$, and the sampling probability is also independent of $N$. Hence it is enough to analyze the single random sample set on the first $\lceil 2N_1^\beta\rceil+1$ positions of the $\DC$ walk from $u$, and all $\Traverse(u, N_1, N)$ calls will give the same sampling set.

    Let $\delta = \min(1,4N_1^{-\beta}\log n)$ be the sampling probability. If $\delta=1$, then $N_1^\beta=O(\log n)$, so the sample set has size $O(\log n)$. Suppose now that $\delta<1$. We have $\mathbb{E}[|S|] = (\lceil 2N_1^\beta\rceil+1)(4N_1^{-\beta}\log n) = \Theta(\log n),$
    and a Chernoff bound implies that the sample size is $O(\log n)$ with probability at least $1-1/(2n^2)$.

    Now assume that $\Traverse(u,N_1,N)$ never enters the if-block. Then the vertices $\DC^j(u)$ for $0\le j\le\lceil 2N_1^\beta\rceil$ are distinct consistent vertices of type $r_T$. Consider the $\HC^{T_R}$ sub-instances attached to them. Since these sub-instances are vertex-disjoint and all contained in $I$, their total size is at most $n^\star$. This means that at least $N_1^\beta$ of them have a $\HC^{T_R}$ sub-instance of size at most $n^\star/N_1^\beta$ attached. The probability that none of these vertices is sampled is at most $(1-\delta)^{N_1^\beta} \le \exp(-4\log n) \le \frac{1}{2n^2}$ for all sufficiently large $n$. A union bound over two events proves the lemma.
\end{proof}

A union bound shows that with probability at least $1-n^\star/n^2$, all sub-instances have an output label satisfying \cref{def:poly-strong-algo}, and for every consistent vertex $u$ of type $r_T$, the event in \cref{lem:poly-sampling-success} occurs. In the rest of the upper bound analysis, all lemmas assume that these events happen. 

Now we verify that $\SolveHC$ follows all conditions in \cref{def:poly-strong-algo}. We first verify stability.

\begin{lemma}\label{lem:poly-upper-bound-stability}
    $\SolveHC$ satisfies condition (i) in \cref{def:poly-strong-algo}.
\end{lemma}
\begin{proof}
    Fix a vertex $u$ and the first size parameter $N_1$. If $u=\perp$, if $u$ is inconsistent, or if $u$ is consistent with $t(u)=\perp$, then $\Prepare(u,N_1,2N_2)$ returns an output independent of $N_2$. Thus we may take $N_2^\star(u,N_1)=N_1$.

    Suppose next that $u$ lies in a $\HC^{T_L}$ sub-instance. Then $\SolveHC(u,N_1,N_2)$ returns the output of $\mathcal{A}_L(u,\lceil N_1^{1-\beta}\rceil,2N_2)$. By condition (i) for $\mathcal{A}_L$, there is a threshold $M^\star$ for the third parameter of $\mathcal{A}_L$. Taking $N_2^\star(u,N_1)$ to be the smallest integer $N_2\ge N_1$ with $2N_2\ge M^\star$ gives the desired statement. The argument for $u$ in a $\HC^{T_R}$ sub-instance is identical, using the threshold $M^\star$ for $\mathcal{A}_R(u,\lceil N_1^{1-\beta}\rceil,\star)$ and taking the smallest $N_2\ge N_1$ such that $\lceil 2N_2N_1^{-\beta}\rceil\ge M^\star$.

    It remains to consider a consistent vertex $u$ with $t(u)=r_T$. For every integer $p\ge 0$, the output of $\Try(u,N_1,p)$ depends only on $u,N_1,p$ and the fixed randomness, and not on the value of $N_2$. If there exists $p$ such that $\Try(u, N_1, p) \ne \perp$, take the minimum of them, denoted as $p^\star$, and take $N_2^\star(u, N_1)$ to be the smallest $N_2 \ge N_1$ such that $\SolveHC(u, N_1, N_2)$ calls $\Try(u,N_1, p^\star)$; Otherwise the output label will be $(\D, 0)$ whatever the value of $N_2$. Both cases satisfy condition (i).
\end{proof}

For condition (ii), we need the following lemma showing that $\Collect$ behaves correctly when the parameter $N$ is large enough.

\begin{lemma}\label{lem:poly-collect-success}
    Fix $N_2$, so the output labeling for $u$ is generated from $\SolveHC(u, N_1, N_2)$. During the execution of $\SolveHC(u,N_1,N_2)$ for some consistent vertex $u$ of type $r_T$, if a call $\Collect(u,k,N_1,N)$ returns a bit $B$, then
    \[
        B=\bigoplus_{j=0}^k
        \left(b\insub(\DC^j(u))\oplus b''(\DC^j(u))\right),
    \]
    where $b''$ is defined in \cref{eq:poly-b''-def}. Moreover, if the total size of the $\HC^{T_L}$ sub-instances attached to $\DC^j(u)$ for $0\le j\le k$ is at most $N$, then $\Collect(u,k,N_1,N)$ returns a bit rather than $\perp$.
\end{lemma}
\begin{proof}
    For each call $\Collect(u,k,N_1,N)$ in $\SolveHC(u,N_1,N_2)$, we first have $N_1 \le N \le 2N_2$. Additionally, $k \le \lceil 2N_1^\beta \rceil$ and $\DC^j(u)$ is consistent for each $0 \le j \le k$: whenever $\Collect$ is called from $\Traverse$, all these vertices were visited before the first terminating event, and whenever $\Collect$ is called from $\Try$, $j^\star \in S_u$ already indicates that $j^\star \le \lceil 2N_1^\beta \rceil$ and all indices between zero and $j^\star$ passes the termination check in line \ref{line:alg-poly-traverse-check-terminate} of $\Traverse$.

    For the first claim, $\Collect$ uses $\Prepare$ with size parameter $(N_1, N)$ to obtain the output bit of each $\LC(\DC^j(u))$, while the output label of $\LC(\DC^j(u))$ is generated by $\Prepare(\LC(\DC^j(u)), N_1, 2N_2)$ inside $\SolveHC(\LC(\DC^j(u)), N_1, N_2)$. If $\LC(\DC^j(u))$ falls into an edge case handled directly by $\Prepare$, its output is independent of the size parameter. Otherwise, condition (i) for $\mathcal{A}_L$ implies that, if $\Prepare(\LC(\DC^j(u)), N_1, N)$ returns an output label other than $(\D, 0)$, this label must be the final output label for $\LC(\DC^j(u))$. When $\LC(\DC^j(u)) = \perp$, $\Prepare$ outputs $(\X,0)$. In sum, no matter when $\Collect$ returns a bit, every $\Prepare$ call has returned an output label other than $(\D,0)$, and the output bit is indeed $b''(\DC^j(u))$ defined in \cref{eq:poly-b''-def}. If a call returns $(\D,0)$, then $\Collect$ fails and returns $\perp$ instead of a bit. This concludes the first claim.

    For the second claim, for each $0 \le j \le k$, $\Collect(u, N_1, N)$ generates a call $\mathcal{A}_L(\LC(\DC^j(u)), \lceil N_1^{1-\beta}\rceil, N)$ when $\LC(\DC^j(u))$ does not fall into the edge cases. For each $0 \le j \le k$, we define $m_j$ to be the size of the $\HC^{T_L}$ sub-instance attached to $\DC^j(u)$, and $m_j = 0$ when it does not exist. These sub-instances are vertex-disjoint, so we have $\sum_{j=0}^k m_j \le N$. This means that $m_j \le N$ for each $0 \le j \le k$, and additionally $\lceil N_1^{1-\beta}\rceil\le N_1\le N$, so we can apply condition (iii) for $\mathcal{A}_L$ to show that all of the $\mathcal{A}_L$ calls return labels other than $(\D, 0)$. 
    
    We need to show further that the probe budget will not be used up by these $\mathcal{A}_L$ calls. If $T_L=\perp$, then each call uses $O(1)$ probes, and the total cost is $O(N_1^\beta) = O(NN_1^{\alpha-1})$ because $N\ge N_1$ and $\alpha = \beta$ when $p = 0$. The case $T_L = \bullet$ is impossible since then $p = 1 \ge q$, contradicting the goodness of $T$. In all remaining cases, the vertices queried are consistent vertices of type $r_{T_L}$ in $\HC^{T_L}$ sub-instances, so we apply \cref{eq:poly-strong-algo-complexity-root} in condition (iv), and the total number of probes is
    \begin{align*}
      & \sum_{j=0}^k O\left(\max(\lceil N_1^{1-\beta}\rceil, m_j)\lceil N_1^{1-\beta}\rceil^{p-1} \log^{d_{T_L}} n\right) \\
     = & O\left(\left(N_1 \cdot N_1^{(1-\beta)(p-1)}+N N_1^{(1-\beta)(p-1)}\right)\log^{d_T} n\right) \tag{$\lceil N_1^{1-\beta}\rceil=\Theta(N_1^{1-\beta})$, $k=O(N_1^\beta)$ and $\sum_{j=0}^k m_j\le N$} \\
     = & O \left(N N_1^{(p-1)(1-\beta)} \log^{d_T} n\right) \tag{$N \ge N_1$} \\
     = & O \left(N N_1^{\alpha - 1} \log^{d_T} n\right) \tag{\cref{eq:poly-upper-bound-(p-1)(1-beta)}}
    \end{align*}
    This is within the probe budget when we set the probe budget constant in $\Collect$ to be sufficiently large.
\end{proof}

We will prove condition (ii) for $\SolveHC$ by considering each constraint in $\HC^T$ one by one. During the analysis of condition (ii), we will fix the value $N_2$. According to the induction hypothesis, constraints are satisfied on inconsistent vertices, vertices of type $\perp$, and vertices in a sub-instance, so it remains to consider constraints touching consistent vertices of type $r_T$. If the output label for such a vertex $u$ is $(\D, 0)$, then it only violates item 4 in $c\outsub$ constraints, which is allowed for condition (ii). So we will focus on cases where $u$ receives a label other than $(\D, 0)$. 

For each consistent vertex $u$ of type $r_T$ that does not receive output label $(\D, 0)$, define $p_u^\star$ as the minimum nonnegative integer $p$ such that $\Try(u,N_1, p)$ does not return $\perp$ and define $N_u=2^{p_u^\star}N_1$. According to the algorithm description, $p_u^\star$ exists and is bounded by $\lceil \log_2 (N_2 / N_1) \rceil$, and $\SolveHC(u,N_1,N_2)$ returns the output label from $\Try(u,N_1, p_u^\star)$. Now we distinguish different possibilities in $\Try(u, N_1, p_u^\star)$ that generate an output label for $u$, and show that in all these cases, constraints on $u$ are satisfied. We start with an observation around $S_{u,p}'$ in $\Try$.

\begin{observation}\label{obs:poly-upper-bound-sample-set-monotone}
    For a consistent vertex $u$ of type $r_T$ and $p \ge 0$, $S_{u,p}' \subseteq S_{u,p+1}'$, where $S_{u,p}'$ is defined in line \ref{line:alg-poly-try-check-X} of $\Try(u, N_1, p)$.
\end{observation}
\begin{proof}
    Let $N=2^pN_1$. The sample set $S_u$ is the same between $\Try(u,N_1,p)$ and $\Try(u,N_1,p+1)$ according to \cref{lem:poly-sampling-success}, so it is enough to prove that every $j\in S_{u,p}'$ belongs to $S_{u,p+1}'$. Write $x=\RC(\DC^j(u))$, and compare the calls $\Prepare(x, N_1, N)$ and $\Prepare(x, N_1, 2N)$.

    If $x$ lies in edge cases handled directly by $\Prepare$, then the output is independent of the size parameter. This covers the cases $x=\perp$, $x$ is inconsistent, and $t(x)=\perp$. The case where $\Prepare$ returns $\perp$ is impossible here: if $x$ is consistent, then $x$ is the right child of a consistent type-$r_T$ vertex and therefore has type $\RC_T(r_T)$, or type $\perp$ when $\RC_T(r_T)=\perp$. Hence $\Prepare(x,N_1,2N)$ is an output label other than $(\D,0)$ whenever $\Prepare(x,N_1,N)$ is one, and $j \in S_{u,p+1}'$.

    It remains to consider the case where $x$ lies in a sub-instance. If $x$ lies in a $\HC^{T_L}$ sub-instance, then $\Prepare(x,N_1,N)$ calls $\mathcal{A}_L(x,\lceil N_1^{1-\beta}\rceil,N)$. Since this value is not $(\D,0)$, condition (i) for $\mathcal{A}_L$ implies that $\Prepare(x,N_1,2N)$ would give the same output. The argument for a $\HC^{T_R}$ sub-instance is identical. Thus $j\in S_{u,p+1}'$, completing the proof.
\end{proof}

We first consider the case where $\Try(u, N_1, p_u^\star)$ returns the backup label $(c_0,b_0)$ from $\Traverse$.

\begin{lemma}\label{lem:poly-upper-bound-consistency-cyc}
    For a consistent vertex $u$ of type $r_T$, if $\Try(u, N_1, p_u^\star)$ returns the backup label $(c_0,b_0)$ from $\Traverse$, then $\HC^T$ constraints on $u$ are satisfied.
\end{lemma}
\begin{proof}
    Let $v = \DC(u)$. The backup label is returned only when $S_{u,p_u^\star}'=\varnothing$ and the backup label produced by $\Traverse(u,N_1,N_u)$ is not $(\D,0)$. According to \cref{obs:poly-upper-bound-sample-set-monotone}, we also know that $S_{u,p}' = \varnothing$ for every $0 \le p < p_u^\star$.

    First suppose this label is $(\Cyc,0)$. Then for some $\ell \in [1, \lceil 2N_1^{\beta} \rceil]$, $\DC^\ell(u)=u$, and all vertices on this directed cycle are consistent vertices of type $r_T$. Additionally, every $\Traverse(u,N_1,N)$ call will see this cycle and produce backup label $(\Cyc, 0)$ whatever the value of $N$ is, so $p_u^\star = 0$. We claim that $p_v^\star$ is also zero and the output label for $v$ will be $(\Cyc, 0)$, then the $\HC^T$ constraints on $u$ are satisfied.
    
    To see this, consider the call $\Try(v, N_1, 0)$. It calls $\Traverse(v,N_1,N_1)$, which will see the cycle and prepare $(\Cyc, 0)$ for the backup label. Additionally, this call will visit exactly the same set of vertices as in $\Traverse(u,N_1,N_1)$, and the sampling is consistent across calls of the same $N_1$ value, this means that $S_v$ will point to the same set of vertices as $S_u$. Since $S_{u,0}' = \varnothing$, $S_{v,0}'$ will also be empty and $\Try(v, N_1, 0)$ will output the backup label $(\Cyc, 0)$.

    Now suppose the backup label is $(\U,b_0)$. This means $\Traverse$ reaches a vertex $\DC^j(u)$ that is either $\perp$ or inconsistent, $u, \DC(u), \DC^2(u), \dots, \DC^{j-1}(u)$ are all consistent vertices of type $r_T$, and $\Collect(u,j-1,N_1,N_u)$ returns $b_0$. If $j = 1$, $\DC(u)$ is $\perp$ or colored $(\perp, 0)$, satisfying item 2 in the $c\outsub$ constraints for $u$, and from \cref{lem:poly-collect-success}, $b\outsub(u) = b\insub(u) \oplus b''(u)$, which is just \cref{eq:poly-lcl-def} with $b'(u) = 0$.
    
    Next, we consider the case where $j > 1$, then $\Traverse(v, N_1, N_u)$ will see the same terminating event by traveling along a suffix of the DC path traveled by $u$, and the corresponding $\Collect$ call collects a subset of bits from $\Collect(u,j-1,N_1,N_u)$. This means that $\Traverse(v, N_1, N_u)$ successfully generates a backup label $(\U, b_0')$. It is possible that the backup label can be produced with a smaller guess: Define $p_v'$ as the smallest guess $p$ such that $\Traverse(v, N_1, 2^pN_1)$ produces a backup label other than $(\D, 0)$. From the argument above, we know $p_v' \le p_u^\star$. In the remaining, we show that $p_v^\star = p_v'$ and the output label for $v$ makes $u$ satisfy its $\HC^T$ constraints.
    
    We have $S_v \subseteq S_u$, so similar to the argument in the cycle case, $S_{v,p}' = \varnothing$ for $0 \le p \le p_u^\star$. This means that for each such $p$, $\Try$ will check whether the backup label is not $(\D, 0)$. As $p_v'$ is the smallest possible $p$ value that generates a backup label other than $(\D, 0)$ and $p_v' \le p_u^\star$, $\Try(v,N_1,p_v')$ will output the backup label while smaller $p$ outputs $\perp$. This proves $p_v' = p_v^\star$. Since the backup label colors $v$ with $\U$, item 2 in the $c\outsub$ constraint is satisfied on $u$. For the $b\outsub$ constraint, according to \cref{lem:poly-collect-success}, we have
    \begin{gather*}
        b\outsub(u) = \bigoplus_{r=0}^{j-1}\left(b\insub(\DC^r(u))\oplus b''(\DC^r(u))\right), \\
        b\outsub(v) = \bigoplus_{r=1}^{j-1}\left(b\insub(\DC^r(u))\oplus b''(\DC^r(u))\right).
    \end{gather*}
    As a result, we have $b\outsub(u) = b\outsub(v) \oplus b\insub(u) \oplus b''(u)$, which is \cref{eq:poly-lcl-def}.

    In both cases, it remains to consider item 5 in the $c\outsub$ constraints that $\LC(u)$ should not be labeled $(\D, 0)$. This comes from the fact that the $\Collect$ call, which would fail when one of the left children is labeled $(\D, 0)$, outputs a bit rather than a failure.
\end{proof}

The other possibility that $\Try(u, N_1, p_u^\star)$ generates an output label is when $S_{u,p_u^\star}'$ is nonempty. The algorithm will then pick $j^\star$ as the minimum item in $S_{u,p_u^\star}'$ and set $w = \DC^{j^\star}(u)$. We first consider the case where $j^\star = 0$, so $w = u$.

\begin{lemma}\label{lem:poly-upper-bound-consistency-x}
    For a consistent vertex $u$ of type $r_T$, if $0 \in S_{u,p_u^\star}'$, then $u$ is colored $\X$ and the $\HC^T$ constraints on $u$ are satisfied.
\end{lemma}
\begin{proof}
    Since $0\in S_{u,p_u^\star}'$, the minimum index $j^\star$ chosen is zero. Thus, the branch for $j^\star=0$ sets $c^\star=\X$. As $\Try(u,N_1,p_u^\star)$ succeeds, $\Collect(u,0,N_1,N_u)$ returns a bit, and the algorithm returns a label of color $\X$.

    The color $\X$ is allowed at $u$ because $u$ has type $r_T$ and $r_T$ is not a leaf. As $0\in S_{u,p_u^\star}'$, the call $\Prepare(\RC(u),N_1,2^{p_u^\star}N_1)$ returns a label other than $(\D,0)$. Since $2^{p_u^\star}N_1\le 2N_2$, by condition (i) of $\SolveHC$, this will be the output label for $\RC(u)$. This proves item 3(b) in the constraint on $c\outsub$. Similarly, the successful $\Collect$ call implies, by \cref{lem:poly-collect-success}, that the output label of $\LC(u)$ is not $(\D,0)$, proving item 5. In sum, $c\outsub$ constraints on $u$ are satisfied.

    Let $b^\star$ be the output bit  of $\Prepare(\RC(u),N_1,2^{p_u^\star}N_1)$, then $b^\star$ equals $b\outsub(\RC(u))$ when $\RC(u)$ exists, and equals zero if $\RC(u)=\perp$. This is exactly the definition of $b'(u)$ in \cref{eq:poly-b'-def}. By \cref{lem:poly-collect-success}, $\Collect(u,0,N_1,2^{p_u^\star}N_1)$ returns $b\insub(u)\oplus b''(u)$. Therefore $b\outsub(u) = b\insub(u)\oplus b''(u)\oplus b^\star$, which is exactly \cref{eq:poly-lcl-def} for a vertex of color $\X$.
\end{proof}

When $j^\star = \min S_{u,p_u^\star}' \ge 1$, we have $w \ne u$ and $\DC(u)$ is not $\perp$. In this case, the call $\Try(u,N_1,p_u^\star)$ will always color $u$ with $\U$. We define $v=\DC(u)$. Since $j^\star$ is sampled only after the termination check in line \ref{line:alg-poly-traverse-check-terminate}, all vertices $\DC^j(u)$ with $0\le j\le j^\star$ are consistent vertices of type $r_T$, so $p^\star_{\DC^j(u)}$ is well-defined. We begin with the following lemma that all vertices between $\P(w)$ and $u$ will agree on the coloring plan and color itself with $\U$.

\begin{lemma}\label{lem:poly-upper-bound-consistency-u-help-lemma}
    For a consistent vertex $u$ of type $r_T$, suppose $S_{u,p_u^\star}'\ne\varnothing$ and $0\notin S_{u,p_u^\star}'$. Define $j^\star$ as the minimum element in $S_{u,p_u^\star}'$, then for every $0 \le j < j^\star$, $c\outsub(\DC^j(u)) = \U$.
\end{lemma}
\begin{proof}
    Define $y=\DC^j(u)$. When $c\outsub(y) = \Cyc$, the cycle containing $y$ must contain $u$ as otherwise the consistency will be violated. Then from the argument in \cref{lem:poly-upper-bound-consistency-cyc}, $c\outsub(u)$ should also be $\Cyc$, a contradiction.

    Another possibility is that $c\outsub(y) = \X$. This happens only when $0\in S_{y,p_y^\star}'$, which by consistent sampling would imply $j \in S_{u,p_y^\star}'$. According to \cref{obs:poly-upper-bound-sample-set-monotone}, to make sure $j^\star > j$ is the minimum among $S_{u,p_u^\star}'$, $p_y^\star > p_u^\star$ must hold.

    However, we will prove next that $p_y^\star \le p_u^\star$, leading to a contradiction. Notice that $j^\star$ is the minimum value in $S_{u,p_u^\star}'$ implies that $j^\star-j$ is the minimum value in $S_{y,p_u^\star}'$. As a result, $\Try(y, N_1, p_u^\star)$ will agree on the same vertex $w$ as $\Try(u, N_1, p_u^\star)$, the check $\SolveHC(w, N_1, N_u/2) \ge 0$ has the same result, and the corresponding $\Collect$ call in $\Try(y,N_1,p_u^\star)$ collects a subset of bits compared to the successful $\Collect$ call in $\Try(u, N_1, p_u^\star)$. This means that $\Try(y,N_1,p_u^\star)$ returns a label rather than a failure and $p_y^\star \le p_u^\star$.

    Since we have proved $p_y^\star \le p_u^\star$, $c\outsub(y)$ cannot be $\D$ either. So the only possible color for $y$ is $\U$, concluding the statement.
\end{proof}

Now we will prove the following lemma, which is sufficient to prove that $u$ is locally consistent:

\begin{lemma}\label{lem:poly-upper-bound-consistency-u}
    For a consistent vertex $u$ of type $r_T$, suppose $S_{u,p_u^\star}'\ne\varnothing$ and $0\notin S_{u,p_u^\star}'$. Define $i^\star_u$ to be the minimum value $i$ such that $\DC^i(u)$ is inconsistent, $\DC^i(u) = \perp$ or $c\outsub(\DC^i(u)) = \X$. Then:
    \begin{itemize}
        \item[1)] $i^\star_u$ exists, and for every $0\le i<i^\star_u$, $\DC^i(u)$ is a consistent vertex of type $r_T$ colored $\U$;
        \item[2)] The $c\outsub$ constraint on $u$ is satisfied, meaning that $\DC(u)$ is either $\perp$ or colored using one of the colors in $\{\X, \U, \perp\}$, and $\LC(u)$ is not colored $\D$ if exists.
        \item[3)] If $\DC^{i^\star_u}(u)$ is colored $\X$, then
        \begin{equation}
        b\outsub(u) = b'(\DC^{i^\star_u}(u)) \oplus \left( \bigoplus_{j=0}^{i^\star_u} (b\insub(\DC^j(u)) \oplus b''(\DC^j(u)))\right),\label{eq:poly-upper-bound-consistency-eq-1}
        \end{equation}
        where $b'$ and $b''$ are defined in \cref{eq:poly-b'-def,eq:poly-b''-def}; otherwise, if $\DC^{i^\star_u}(u)$ is inconsistent or $\DC^{i^\star_u}(u) = \perp$, then
        \begin{equation}
        b\outsub(u) = \bigoplus_{j=0}^{i^\star_u-1} (b\insub(\DC^j(u)) \oplus b''(\DC^j(u))).\label{eq:poly-upper-bound-consistency-eq-2}
        \end{equation}
    \end{itemize}
\end{lemma}
\begin{proof}
    We prove the statement by induction on $p_u^\star$. This means that, to show the statement for a vertex $u$, we assume that the induction hypothesis holds for every consistent vertex $u'$ of type $r_T$ with $p_{u'}^\star < p_u^\star$. Let $j^\star=\min S_{u,p_u^\star}'$ and $w=\DC^{j^\star}(u)$, then $j^\star \ge 1$. $\Try$ then checks whether $\SolveHC(w,N_1,N_u/2)\ne(\D,0)$, and enters one of the two branches starting at lines \ref{line:alg-poly-upper-bound-prev-success-block} and \ref{line:alg-poly-upper-bound-prev-fail-block}. As $N_u = 2^{p_u^\star} N_1$, this is equivalent to checking whether $p_w^\star < p_u^\star$.

    Consider first the branch where $\SolveHC(w,N_1,N_u/2)\ne(\D,0)$. Then, according to condition (i) for $\SolveHC$, the label returned by this call is the output label of $w$, and we know $p_w^\star < p_u^\star$. The case where $c\outsub(w) = \Cyc$ cannot exist: as otherwise the $\DC$ cycle from $w$ should pass $u$, and $c\outsub(u)$ should also be $\Cyc$ according to \cref{lem:poly-upper-bound-consistency-cyc}, a contradiction. $p_w^\star < p_u^\star$ implies that $c\outsub(w) \ne \D$, thus $w$ is colored $\X$ or $\U$.

    If $w$ is colored $\X$, we set $i^\star_u=j^\star$. \cref{lem:poly-upper-bound-consistency-u-help-lemma} shows that for all $0 \le j < j^\star$, $c\outsub(\DC^j(u)) = \U$, so $j^\star$ is the correct value for $i^\star_u$. If $w$ is colored $\U$, $w$ receives this label from two possible places: one is from the backup label produced by $\Traverse$, in which, according to the proof in \cref{lem:poly-upper-bound-consistency-cyc}, there exists a value $i^\star_w$ such that $\DC^{i^\star_w}(w)$ is inconsistent or $\perp$ and all vertices between $w$ and $\DC^{i^\star_w - 1}(w)$ are colored $\U$. As a result, $i^\star_u$ exists and equals $j^\star+i^\star_w$. The final case is that $S_{w,p_w^\star}' \ne \varnothing$ and $0 \not\in S_{w,p_w^\star}'$. In this case, since $p_w^\star < p_u^\star$, we can apply the induction hypothesis to get $i^\star_w$, and assign $i^\star_u=i^\star_w + j^\star$. The arguments above in sum prove item 1).
    
    The item 2 in the $c\outsub$ constraint on $u$ follows because $\DC(u)$ is either $\perp$, inconsistent, or colored $\X$ when $i^\star_u=1$, and colored $\U$ when $i^\star_u>1$; also, the successful $\Collect$ call implies that $\LC(u)$ is not colored $\D$ if it exists, so item 5 in the $c\outsub$ constraint is also satisfied. This proves item 2).

    For items 3), notice that the $\Collect$ call in this branch uses parameter $(u,j^\star-1,N_1,N_u)$, and it returns a bit $B$ because $\Try(u,N_1,p_u^\star)$ succeeds. By \cref{lem:poly-collect-success}, the algorithm returns $b\outsub(u)=b\outsub(w)\oplus B$ with
    \[
        B=\bigoplus_{i=0}^{j^\star-1}
        \left(b\insub(\DC^i(u))\oplus b''(\DC^i(u))\right).
    \]

    If $w$ is colored $\X$, then substituting \cref{eq:poly-lcl-def} for $w$ proved in \cref{lem:poly-upper-bound-consistency-x} gives \cref{eq:poly-upper-bound-consistency-eq-1}. If $w$ is colored $\U$ and receives the backup label from $\Try(w,N_1,p_w^\star)$, the proof of \cref{lem:poly-upper-bound-consistency-cyc} implies that $b\outsub(w) = \bigoplus_{i=0}^{i^\star_w - 1}(b\insub(\DC^i(w)) \oplus b''(\DC^i(w)))$, and this gives \cref{eq:poly-upper-bound-consistency-eq-2}. Finally, the case where $w$ receives a label since $S_{w,p_w^\star}' \ne \varnothing$ and $0 \not\in S_{w,p_w^\star}'$ follows from the induction hypothesis.

    It remains to consider the branch where $\SolveHC(w,N_1,N_u/2)=(\D,0).$ In this case, we know $p_w^\star \ge p_u^\star$. Since $j^\star \in S_{u,p_u^\star}'$, the consistent sampling implies that $0\in S_{w,p_u^\star}'$. Moreover, $\Collect(w,0,N_1,N_u)$ will succeed since it collects a subset of bits compared to the successful call $\Collect(u,j^\star,N_1,N_u)$. As a result, from \cref{lem:poly-upper-bound-consistency-x}, we know that $p_w^\star = p_u^\star$, the output color for $w$ will be $\X$, and its local constraints are satisfied.

    We set $i^\star_u=j^\star$, then \cref{lem:poly-upper-bound-consistency-u-help-lemma} immediately implies item 1). Also, the successful $\Collect(u,j^\star,N_1,N_u)$ call implies that $\LC(u)$ is not colored $\D$ if it exists, so the $c\outsub$ constraint on $u$ is satisfied, implying item 2).

    For item 3), notice that the successful $\Collect(u,j^\star,N_1,N_u)$ call returns a bit $B$ with
    \[
        B=\bigoplus_{j=0}^{j^\star}
        \left(b\insub(\DC^j(u))\oplus b''(\DC^j(u))\right)
    \]
    according to \cref{lem:poly-collect-success}. At the same time, condition (i) for $\SolveHC$ implies that the bit $b^\star$ obtained from $\Prepare(\RC(w), N_1, N)$ is exactly $b'(w)$. Since $\Try(u,N_1,p_u^\star)$ returns $b^\star\oplus B$, this gives \cref{eq:poly-upper-bound-consistency-eq-1}.
\end{proof}

Wrapping up, we are finally able to show condition (ii) for $\SolveHC$.

\begin{lemma}\label{lem:poly-upper-bound-consistency}
    $\SolveHC$ satisfies condition (ii) in \cref{def:poly-strong-algo}.
\end{lemma}
\begin{proof}
    Fix $N_2\ge N_1$ and consider the output labeling generated by $\SolveHC(u,N_1,N_2)$. Inconsistent vertices, vertices of type $\perp$, and vertices in $\HC^{T_L}$ or $\HC^{T_R}$ sub-instances satisfy all relevant constraints by the definition of $\Prepare$ and the induction hypothesis. Hence it remains to consider a consistent vertex $u$ with $t(u)=r_T$.

    If $\SolveHC(u,N_1,N_2)$ returns $(\D,0)$, then the only possible violation is item 4 in the constraint on $c\outsub$, which is explicitly allowed in condition (ii). Otherwise $p_u^\star$ is defined. If $\Try(u,N_1,p_u^\star)$ returns the backup label from $\Traverse$, the constraints on $u$ are satisfied by \cref{lem:poly-upper-bound-consistency-cyc}. Otherwise $S_{u,p_u^\star} \ne \varnothing$. If $0\in S_{u,p_u^\star}'$, constraints on $u$ are satisfied by \cref{lem:poly-upper-bound-consistency-x}. The remaining case is that $S_{u,p_u^\star}'\ne\varnothing$ and $0\notin S_{u,p_u^\star}'$, which is handled by \cref{lem:poly-upper-bound-consistency-u}. These lemmas cover every possible cases, so condition (ii) follows.
\end{proof}

Now we proceed to conditions (iii) and (iv).

\begin{lemma}\label{lem:poly-upper-bound-success}
    $\SolveHC$ satisfies condition (iii) in \cref{def:poly-strong-algo}.
\end{lemma}
\begin{proof}
    Let $N_2\ge \max(N_1,n^\star)$, and let $u$ be a consistent vertex with $t(u)=r_T$. Set $p=\left\lceil \log_2(N_2/N_1)\right\rceil$ and $N=2^pN_1$, then $N_2\le N\le 2N_2$, and in particular $N\ge n^\star$. Consider the call $\Try(u,N_1,p)$, we will show that this call must return a label rather than a failure, so the output label for $u$ is not $(\D, 0)$.

    First of all, for every $\Collect$ call during $\Try(u, N_1, p)$, all $\HC^{T_L}$ sub-instances are vertex-disjoint, so their total size is at most $n^\star \le N$. According to \cref{lem:poly-collect-success}, every $\Collect$ call will succeed.

    Now consider every possibility that $u$ gets a label. If $\Traverse(u,N_1,N)$ reaches a terminating event (i.e. reaching $\perp$, an inconsistent vertex, or returning to itself) and produces a backup label, the corresponding $\Collect$ call always succeed, so a backup label other than $(\D, 0)$ is produced. This means that, when $S_{u,p}' = \varnothing$, $\Try(u, N_1, p)$ still output a label other than $(\D, 0)$. When $S_{u,p}' \ne \varnothing$, $\Try$ always outputs a label other than $(\D, 0)$ no matter the result of $\SolveHC(w, N_1, N/2)$ is.

    For the other case that $\Traverse(u,N_1,N)$ does not produce a backup label, according to \cref{lem:poly-sampling-success}, there is an index $j\in S_u$ whose attached $\HC^{T_R}$ sub-instance has size at most $n^\star/N_1^\beta$. $\Prepare(\RC(\DC^j(u)),N_1,N)$ will invoke $\mathcal{A}_R(\RC(\DC^j(u)), \lceil N_1^{1-\beta}\rceil, \lceil NN_1^{-\beta}\rceil)$. As $\lceil N_1^{1-\beta}\rceil\le \lceil NN_1^{-\beta}\rceil$ and $n^\star / N_1^\beta \le \lceil NN_1^{-\beta} \rceil$, condition (iii) for $\mathcal{A}_R$ implies that this $\mathcal{A}_R$ call returns a label other than $(\D, 0)$. Hence $S_{u,p}'$ is nonempty, and $\Try$ outputs a label other than $(\D, 0)$ following the same argument as above.
\end{proof}

\begin{lemma}\label{lem:poly-upper-bound-probes}
    $\SolveHC$ satisfies condition (iv) in \cref{def:poly-strong-algo}.
\end{lemma}
\begin{proof}
    Apart from the induction hypothesis for $\mathcal{A}_L$ and $\mathcal{A}_R$, we apply an induction on $d$ to prove the following three claims simultaneously for the current instance of size $n^\star$:
    \begin{itemize}
        \item[A($d$)] For every $0\le r\le d$,
        \[
            \Try(u,N_1,r)
            \text{ uses }
            O(2^rN_1^\alpha\log^{d_T}n)
            \text{ probes.}
        \]
        \item[B$_{\mathrm{all}}(d)$] For every $N_1\le N_2\le N_12^d$ and every vertex $u$,
        \[
            \SolveHC(u,N_1,N_2)
            \text{ uses }
            O(N_2N_1^{\alpha-1}\log^{d_T}n)
            \text{ probes.}
        \]
        \item[B$_{\mathrm{root}}(d)$] For every $N_1\le N_2\le N_12^d$ and every consistent vertex $u$ with $t(u)=r_T$,
        \[
            \SolveHC(u,N_1,N_2)
            \text{ uses }
            O(\max(N_1,n^\star)N_1^{\alpha-1}\log^{d_T}n)
            \text{ probes.}
        \]
    \end{itemize}
    The claims B$_{\mathrm{all}}(d)$ and B$_{\mathrm{root}}(d)$ for all $d\ge 0$ are exactly the two bounds in condition (iv). The base case for $d=-1$ is trivial. We show B$_{\mathrm{all}}(d-1)\Rightarrow$ A($d$), and then A($d)\Rightarrow$ B$_{\mathrm{all}}(d)$ and B$_{\mathrm{root}}(d)$, concluding the induction.

    We first show B$_{\mathrm{all}}(d-1)\Rightarrow$ A($d$). Fix $0\le r\le d$ and set $N=2^rN_1$. Consider the subroutines that $\Try(u,N_1,r)$ calls. It makes only constantly many $\Collect$ calls, each with budget 
    \[
        O(NN_1^{\alpha-1}\log^{d_T}n)
        =
        O(2^rN_1^\alpha\log^{d_T}n).
    \] 
    The $\Traverse$ call uses $O(N_1^\beta)\le O(N_1^\alpha)\le O(N_1^\alpha)$ probes. For each sampled index $j\in S_u$, the call to $\Prepare(\RC(\DC^j(u)),N_1,N)$ is either handled directly in $O(1)$ probes, or invokes $\mathcal{A}_R(\RC(\DC^j(u)),\lceil N_1^{1-\beta}\rceil,\lceil NN_1^{-\beta}\rceil)$.
    By induction hypothesis, applying \eqref{eq:poly-strong-algo-complexity-universal} for $\mathcal{A}_R$, we know that each of these calls costs
    \begin{align*}
        &O\left(NN_1^{-\beta}\cdot N_1^{(1-\beta)(q-1)}\log^{d_T - 1}n\right) \tag{$d_{T_R} \le d_T - 1$}\\
        ={}&O\left(NN_1^{\alpha-1}\log^{d_T-1}n\right). \tag{$-\beta+(1-\beta)(q-1) = \frac{-(q-p)}{1+q-p} + \frac{q-1}{1+q-p} = \frac{p-1}{1+q-p} = \alpha-1$}
    \end{align*}
    Since $|S_u|=O(\log n)$ by \cref{lem:poly-sampling-success}, all calls together cost $O(NN_1^{\alpha-1}\log^{d_T}n)$. Finally, $\Try$ may call $\SolveHC(w,N_1,N/2)$. If $r=0$, this call returns immediately because $N/2<N_1$. Otherwise, B$_{\mathrm{all}}(d-1)$ applies and gives cost 
    \[
        O((N/2)N_1^{\alpha-1}\log^{d_T}n) 
        = 
        O(NN_1^{\alpha-1}\log^{d_T}n).
    \]
    In sum, $\Try(u, N_1, r)$ has a cost of $O(2^rN_1^\alpha \log^{d_T} n)$. This proves A($d$).

    We now prove B$_{\mathrm{all}}(d)$ from A$(d)$. The initial $\Prepare(u,N_1,2N_2)$ may call one of $\mathcal{A}_L$ and $\mathcal{A}_R$. By \cref{eq:poly-strong-algo-complexity-universal} applied to both $\mathcal{A}_L$ and $\mathcal{A}_R$, the call $\mathcal{A}_L(u,\lceil N_1^{1-\beta}\rceil,2N_2)$ costs 
    \[
        O(N_2N_1^{(1-\beta)(p-1)}\log^{d_{T_L}}n)
        =
        O(N_2N_1^{\alpha-1}\log^{d_T}n), \tag{\cref{eq:poly-upper-bound-(p-1)(1-beta)}}
    \]
    and the call $\mathcal{A}_R(u, \lceil N_1^{1-\beta}\rceil,\lceil 2N_2N_1^{-\beta}\rceil)$ costs
    \[
        O(N_2N_1^{-\beta}N_1^{(1-\beta)(q-1)}\log^{d_{T_R}}n)
        =
        O(N_2N_1^{\alpha-1}\log^{d_T}n).
    \]
    If $\Prepare$ returns a label, this proves B$_{\mathrm{all}}(d)$. Otherwise, $u$ is a consistent vertex of type $r_T$, and $\SolveHC$ runs $\Try(u,N_1,r)$ over a geometric sequence of size guesses all the way up to $2N_2$. By A($d$), the total cost of all these calls is $O(N_2N_1^{\alpha-1}\log^{d_T}n)$. This proves B$_{\mathrm{all}}(d)$.

    It remains to prove B$_{\mathrm{root}}(d)$ from A$(d)$. Let $u$ be a consistent vertex with $t(u)=r_T$. Then $\Prepare(u,N_1,2N_2)$ returns $\perp$ with $O(1)$ probes. When $N_2<n^\star$, B$_{\mathrm{all}}(d)$ already implies B$_{\mathrm{root}}(d)$, so we consider the other case where $N_2\ge n^\star$. The proof of \cref{lem:poly-upper-bound-success} shows that $\Try(u,N_1,r)$ succeeds as soon as $2^rN_1 \ge n^\star$. Therefore the largest size guess used by $\SolveHC$ is $O(\max(N_1,n^\star))$, and the probe complexity forms a geometric sequence according to A($d$). As a result, they sum up to
    \[
        O(\max(N_1,n^\star)N_1^{\alpha-1}\log^{d_T}n),
    \]
    which is the desired bound. This proves B$_{\mathrm{root}}(d)$ and completes the proof of condition (iv).
\end{proof}

\begin{proof}[Proof of \cref{thm:poly-upper-bound-strong}]
    The proof is by induction on the size of $T$. The base cases are \cref{lem:poly-upper-bound-base-perp,lem:poly-upper-bound-base-bullet}. For the induction part, assume that for some tree $T$ there is a strong algorithm for $\HC^{T_L}$ and $\HC^{T_R}$, then the algorithm is given by $\SolveHC$. For the conditions, when all the events described in the start of the analysis section happens, which is of probability at least $1-n^\star/n^2$, all four conditions in \cref{def:poly-strong-algo} are satisfied shown in \cref{lem:poly-upper-bound-stability,lem:poly-upper-bound-consistency,lem:poly-upper-bound-success,lem:poly-upper-bound-probes}.
\end{proof}

\subsection{Lower bound}

We will show the lower bound via induction on a stronger claim that one specific output bit in a dedicated distribution is indistinguishable for any deterministic LCA with $o(n^{\val(T)})$ probes.

    \begin{lemma}\label{lem:poly-hard-root-bit}
    Let $T\ne\perp$ be a good ordered binary tree. For every sufficiently large $n$, there is a set $\mathcal{S}_{T,n}$ of $\HC^T$ instances of size at most $n$ with the following properties:
    \begin{enumerate}[(i)]
        \item All instances in $\mathcal{S}_{T, n}$ have the same underlying graph and $T$-labeling, in which the underlying graph is a degree-4 tree, every vertex is consistent, and there is a unique root vertex $x$ with $t(x) = r_T$ and $\P(x) = \perp$;
        \item For any instance $I$ in $\mathcal{S}_{T,n}$, there exists $b_I \in \{0,1\}$ such that any solution satisfying all $\HC^T$ constraints for $I$ must have $b\outsub(x) = b_I$. Define $\mathcal{S}_{T,n}^B$ as the set of instances $I$ in $\mathcal{S}_{T,n}$ with $b_I = B$;
        \item For any set $V_0$ of vertices of size $o(n^{\val(T)})$, there exists a bijection $f_{V_0}$ between $\mathcal{S}_{T,n}^0$ and $\mathcal{S}_{T,n}^1$, where for each $I \in \mathcal{S}_{T,n}$ and $u \in V_0$, $b\insub(u)$ is the same between $I$ and $f_{V_0}(I)$.
    \end{enumerate}
\end{lemma}

\begin{proof}
    We prove the lemma by induction on the size of $T$. 
    
    First, consider the base case $T=\bullet$. $\mathcal{S}_{\bullet,n}$ contains all instances with the following property: The underlying graph is a directed $\DC$ path $u_1,u_2,\ldots,u_n$, where $u_1=x$; in the $T$-labeling, all vertices have type $r_T$, all $\LC$ and $\RC$ ports are $\perp$. There is no constraint on $b\insub$ inputs on each vertex, so $\mathcal{S}_{\bullet,n}$ contains $2^n$ instances, one different $b\insub$ assignment per instance.
    
    Condition (i) is satisfied by definition. Since $r_T$ is a leaf and a root, $\X$ and $\D$ are prohibited. Moreover, no vertex can output $\Cyc$. Thus, every legal solution colors each vertex with $\U$, and \cref{eq:poly-lcl-def} gives $b_I = b\outsub(x)=\bigoplus_{i=1}^n b\insub(u_i)$. This proves condition (ii). Finally, for any $V_0$ of size $o(n)$, choose $w\notin V_0$ and map an instance to the instance obtained by flipping only $b\insub(w)$. This map is clearly a bijection between $\mathcal{S}_{\bullet,n}^0$ and $\mathcal{S}_{\bullet,n}^1$, and it preserves all input bits on $V_0$. Hence condition (iii) holds.

    Now assume that $T \ne \bullet$, and write
    \[
        p=\val(T_L),\qquad q=\val(T_R),\qquad
        \alpha=\val(T),\qquad
        \beta=\frac{q-p}{1+q-p}.
    \]
    Since $T$ is good, we have $p<q$. Let $\ell=\Theta(n^\beta)$ and $m=\Theta(n^{1-\beta})$, and we define $\mathcal{S}_{T,n}$ as all instances that can be generated by the following procedure:
    \begin{itemize}
        \item Create a directed $\DC$ path $x = u_1,u_2,\ldots,u_\ell$ of length $\ell$, all of which have type $r_T$.
        \item Generate $(A_1,\ldots,A_\ell) \in \{0,1\}^\ell$ and define $S_i=A_{i+1}\oplus A_{i+2}\oplus\cdots\oplus A_\ell$ for $0 \le i < \ell$ and $S_\ell = 0$.
        \item If $T_L=\perp$, then for every vertex $u_i (1 \le i \le \ell)$, set $\LC(u_i)=\perp$ and $b\insub(u_i)=A_i$. Otherwise, we set $b\insub(u_i)=0$ and attach to $\LC(u_i)$ an arbitrary instance drawn from $\mathcal{S}_{T_L,m}^{A_i}$.
        \item For every $i \in [1,\ell]$, attach to $\RC(u_i)$ an arbitrary instance drawn from $\mathcal{S}_{T_R,m}^{S_i}$.
    \end{itemize}
      For ``attaching an instance to a port'', we mean creating a copy of the instance, making the port point to the root of the copy, and setting the parent port of that copied root back to the vertex $u_i$. The constants hidden in $\Theta(\cdot)$ are chosen so that the final instance has at most $n$ and $\Theta(n)$ vertices. This completes the definition of $\mathcal{S}_{T,n}$.
    
    According to the induction hypothesis, the structure and $T$-labeling of the attached instance is fixed, while the structure and $T$-labeling of the top-level $\DC$ path is also fixed. Hence all instances in $\mathcal{S}_{T,n}$ have the same underlying graph and $T$-labeling, in which it is easy to prove that the graph is a tree, every vertex is consistent, and the only type-$r_T$ vertex with parent $\perp$ is $x=u_1$. Thus condition (i) holds.

    Now we consider condition (ii). Define the contribution of a vertex $u_i$, denoted as $c(u_i)$, as $b''(u_i) \oplus b\insub(u_i)$, where $b''$ is defined in \cref{eq:poly-b''-def}. According to the construction, $c(u_i) = A_i$ for each $1 \le i \le \ell$. Now, for any legal output, define $k = \min\{i \mid c\outsub(u_i) = \X\}$, and $k = \infty$ when there is no such $i$. Notice that $\Cyc$ is prohibited since there is no cycle in the instance. When $k \ne \infty$, expanding \cref{eq:poly-lcl-def} gives 
    \[
        b\outsub(x)=c(u_1) \oplus \dots \oplus c(u_k) \oplus S_k = A_1\oplus\cdots\oplus A_k\oplus S_k= \oplus_{i=1}^\ell A_i.
    \]
    If $k = \infty$, then we have
    \[
        b\outsub(x)=c(u_1) \oplus \dots \oplus c(u_\ell) = A_1\oplus\cdots\oplus A_\ell=\oplus_{i=1}^\ell A_i.
    \]
    In all, condition (ii) is satisfied by $b_I = \oplus_{i=1}^\ell A_i$.

    It remains to prove condition (iii). Define $L_i$ and $R_i$ as the $\HC^{T_L}$ and $\HC^{T_R}$ instance attached to $u_i$, with $L_i=\varnothing$ when $T_L=\perp$. For a fixed set of vertices $V_0$ of size $o(n^{\val(T)})$, let $V_i^L = V_0 \cap L_i$ and $V_i^R = V_0 \cap R_i$. We have
    \[
    |V_i^R| = o(n^\alpha) = o(n^{q(1-\beta)}) = o(m^q),
    \]
    so according to the induction hypothesis, for each $i \in [1,\ell]$ there exists a bijection $f^R_i$ between $\mathcal{S}_{T_R, m}^0$ and $\mathcal{S}_{T_R, m}^1$ satisfying condition (iii) for $V_i^R$. Additionally, when $T_L \ne \perp$, we have 
    \[
        \min_i |V_i^L| \le \frac{1}{\ell} \sum_{i=1}^\ell |V_i^L| \le \frac{o(n^\alpha)}{\ell} = o(n^{\alpha - \beta}) = o(n^{p(1-\beta)}) = o(m^p).
    \]
    Take $i^\star = \arg \min_i |V_i^L|$, then according to the induction hypothesis, there exists a bijection $f^L_{i^\star}$ between $\mathcal{S}_{T_L, m}^0$ and $\mathcal{S}_{T_L, m}^1$ satisfying condition (iii) for $V_{i^\star}^L$. For the case where $T_L = \perp$, $\alpha = q/(1+q) = \beta$, so $|V_0| = o(n^\beta)$ and there exists some $i^\star$ where $u_{i^\star} \not\in V_0$. We let $f_{i^\star}^L$ denote the operation that flips $b\insub(u_{i^\star})$.

    Now we construct the bijection $f_{V_0}$ between $\mathcal{S}_{T,n}^0$ and $\mathcal{S}_{T,n}^1$. For each instance $I$ in $\mathcal{S}_{T,n}$ with bits $A_1, A_2, \dots, A_\ell$, we flip $A_{i^\star}$ by replacing the $\HC^{T_L}$ instance attached to $u_{i^\star}$ with its image in $f_{i^\star}^L$ if $T_L\ne\perp$, or by flipping $b\insub(u_{i^\star})$ if $T_L=\perp$. For every $1 \le i < i^\star$, the value $S_i$ also flips, so we replace the $\HC^{T_R}$ instance attached to $u_i$ with its image under $f_i^R$. All other inputs are unchanged. This completes the construction.

    Since $f_{i^\star}^L$ and all $f_i^R$ are bijections, and the same components are changed when we apply the map again, $f_{V_0}$ is an bijection. It remains to check that $f_{V_0}$ flips $b_I$. According to the induction hypothesis, the map changes exactly $A_{i^\star}$ and $S_1, S_2, \dots,S_{i^\star-1}$. By following the same analysis for condition (ii), we know that all solutions for $f_{V_0}(I)$ must have $b\outsub(x) = b_I \oplus 1$, so $f_{V_0}$ is indeed a bijection between $\mathcal{S}_{T,n}^0$ and $\mathcal{S}_{T,n}^1$.

    Finally, we show that $I$ and $f_{V_0}(I)$ have the same input labels on $V_0$. For replaced $\HC^{T_L}$ and $\HC^{T_R}$ copies, this follows from the induction hypothesis applied to $V_{i^\star}^L$ and to the sets $V_i^R$. All other attached instances are identical. Vertices in the top $\DC$ chain keep the same input when $T_L \ne \perp$, and when $T_L = \perp$, the only input bit that changes is at $u_{i^\star}\notin V_0$. This concludes condition (iii).
\end{proof}

\begin{proof}[Proof of \cref{lem:poly-lower-bound}]
    Suppose for contradiction that there is a randomized LCA that solves $\HC^T$ with $o(n^{\val(T)})$ probes and succeeds in every $\HC^T$ instance of size at most $n$ with probability larger than $1/2$. Pick any distribution $\mathcal{D}$ over instances of size at most $n$. According to Yao's minimax principle, there is a deterministic LCA $\mathcal{A}$ with the same probe bound that succeeds with probability larger than $1/2$ over $\mathcal{D}$.

    Now set $\mathcal{D}$ to be a uniform distribution over $\mathcal{S}_{T,n}$. Since all instances in the support have the same graph and $T$-labeling, the transcript of $\mathcal{A}(x)$ is determined by the set of probed vertices and their $b\insub$ values. Fix any possible transcript $\tau$, and let $V(\tau)$ be the set of vertices probed in this transcript. We have $|V(\tau)|=o(n^{\val(T)})$. Applying condition (iii) of \cref{lem:poly-hard-root-bit} to $V(\tau)$ gives a bijection between $\mathcal{S}_{T,n}^0$ and $\mathcal{S}_{T,n}^1$ that preserves all input bits seen in $\tau$. Since all instances have the same probability of being chosen, and by condition (ii), a correct output must have root bit $b_I$, we have $\Pr[b_I=0 \mid \tau]=\Pr[b_I=1 \mid \tau]=1/2$, and the success probability of $\mathcal{A}$ over $\mathcal{D}$ is at most $1/2$, a contradiction.
\end{proof}

\section{Discussion and Future Directions}
\label{sec:discussion}

In this section, we provide discussions and propose several possible future directions.

\subsection{Unifying Constructions Between \texorpdfstring{$\LOCAL$}{LOCAL} and \texorpdfstring{$\Vol$}{VOLUME}}
\label{subsec:local-vs-vol}

As we have mentioned in the introduction, our construction is completely different from the previous approach to provide a similar result in the distributed $\LOCAL$ model \cite{balliu2018new}. 

The $\LOCAL$ constructions in \cite{balliu2018almost,balliu2018new} encode the execution of some special Turing machines into the graph, and the time complexity of such Turing machines will determine the round complexity of the problem. The current realization of such encodings crucially relies on the fact that each vertex reads \emph{all} vertices in its local neighborhood and checks whether the input graph follows certain nice structure locally, so the LCLs require $\Omega(n)$ probes in the $\Vol$ and LCA model.

Conversely, our construction does not yield dense round complexities in the $\LOCAL$ model for the polynomial regime. For the problem $\HC^T$, follow the same argument in \cref{sec:tech-overview}, if the left and right children have complexity $\tilde\Theta(n^p)$ and $\tilde\Theta(n^q)$, then the round complexity for a top path of length $\ell$ will be $\tilde O(\min(\ell + (n/\ell)^p, (n/\ell)^q))$. The term balancing always gives $\tilde \Theta(n^{1/(p+1)})$ round complexity for $p<q$, so the construction only gives $\tilde \Theta(n^{1/k})$ round complexity for some positive integer $k$.

As a consequence, we raise the following open question to find a possible way to integrate density results in the $\LOCAL$ and $\Vol$ model.

\begin{question}
    For any pair of rationals $0 < r/s < p/q < 1$, construct a LCL $\Pi_{r/s,p/q}$ with round complexity $\tilde \Theta(n^{r/s})$ in the $\LOCAL$ model and probe complexity $\tilde \Theta(n^{p/q})$ in the $\Vol$ or LCA model.
\end{question}

\subsection{LCAs and \texorpdfstring{$\Vol$}{VOLUME} Model}
\label{subsec:lca-vs-vol}

From \cref{obs:vol-weaker-than-lca}, the $\Vol$ model is an LCA with additional locality restrictions. However, are these additional restrictions necessary? Can we find a problem that separates LCAs and $\Vol$ algorithms? This question is trivial if we allow arbitrary problems to be under consideration (for example, consider the problem that each vertex in a graph needs to output the input label on vertex $1$). But for LCLs, this question is still open.

\begin{question}
    Provide an LCL that separates the LCA and $\Vol$ model more than polynomially, or prove that such an LCL does not exist.
\end{question}

The constraint of connectivity for the probe region is shown to be not important in \cite{goos2015non}: Any randomized LCA $\mathcal{A}$ for an LCL can be transformed into another randomized LCA $\mathcal{A}'$ with polynomial probe complexity blowup, with the additional property that $\mathcal{A}'$ always keeps its probed region a connected component in the graph. So the remaining question is to show whether the difference in randomness is important. 

The question about different randomness models is also explicitly mentioned as an open question in \cite{rosenbaum2020seeing}. There is work discussing the benefit of shared randomness in the distributed setting \cite{balliu2025shared, hadad2025shared}, and it has been shown that in the $\LOCAL$ model, shared randomness helps exponentially for certain LCLs. Yet it is currently unknown whether this is also the case in the $\Vol$ model and LCAs.

\subsection{LCLs and CSPs}
\label{subsec:lcl-vs-csp}

The definition of LCLs is almost the same as that of finite-domain CSPs, except that LCLs are defined on bounded-degree constraint graphs. A natural question is whether it is possible to utilize existing tools in CSP complexity classification for LCL complexity classification. From decades of research, it turns out that the complexity of a CSP is largely related to its \emph{polymorphism}, an important definition in universal algebra. See \cite{barto2017polymorphisms} for an introduction to the application of polymorphism in complexity theory and \cite{barto2015constraint} for a literature review. 

Previous work suggests indirect connections between universal algebra tools and distributed graph algorithms: \cite{butti2024complexity} applies universal algebra techniques to distributed CSPs \cite{yokoo1992distributed}, but the distributed model is not a standard one. An emergent line of research, connecting distributed graph algorithms with descriptive combinatorics \cite{bernshteyn2023distributed, bernshteyn2025borel, brandt2021local, grebik2023local}, may also have implicit connections to a universal algebraic approach \cite{thornton2022algebraic}. Yet, we are not aware of a directly application of universal algebra to the classification of LCL complexities in standard distributed models. As a result, we raise the following question.

\begin{question}
    Provide a way to apply a universal algebraic approach to LCL complexity classification in distributed settings, or provide evidence that these algebraic tools may not be suitable to classify LCLs.
\end{question}

\bibliographystyle{plain}
\bibliography{ref}

\end{document}